\documentclass{article}

\usepackage[T1]{fontenc}
\usepackage{lmodern}
\usepackage{times}
\usepackage{amsfonts}
\usepackage{amssymb}
\usepackage{amsmath}
\usepackage{verbatim}
\usepackage{lineno}

\DeclareMathAlphabet{\mathpzc}{OT1}{pzc}{m}{it}

\usepackage{colordvi}
\usepackage{color}
\newcommand{\red}{\textcolor{black}}

\numberwithin{equation}{section}

\newcommand{\dom}{{\mbox{\rm dom}}}
\newcommand{\ran}{\mathrm{ran}}

\newcommand{\supp}{{\mbox{\rm supp}}}

\newcommand{\tr}{{\mathrm{tr}}}
\newcommand{\card}{{\mathrm{card}}}

\newcommand{\slim}{\,\mbox{\em s-}\hspace{-2pt} \lim}

\newtheorem{theorem}{Theorem}[section]
\newtheorem{lemma}[theorem]{Lemma}
\newtheorem{corollary}[theorem]{Corollary}
\newtheorem{proposition}[theorem]{Proposition}
\newtheorem{remark}[theorem]{Remark}
\newtheorem{assumption}[theorem]{Assumption}
\newtheorem{definition}[theorem]{Definition}
\newtheorem{example}[theorem]{Example}

\newcommand{\ba}{\begin{array}}
\newcommand{\ea}{\end{array}}
\newcommand{\bea}{\begin{eqnarray}}
\newcommand{\eea}{\end{eqnarray}}
\newcommand{\bead}{\begin{eqnarray*}}
\newcommand{\eead}{\end{eqnarray*}}
\newcommand{\be}{\begin{equation}}
\newcommand{\ee}{\end{equation}}
\newcommand{\bed}{\begin{displaymath}}
\newcommand{\eed}{\end{displaymath}}
\newcommand{\bal}{\begin{align}}
\newcommand{\eal}{\end{align}}

\newcommand{\bl}{\begin{lemma}}
\newcommand{\el}{\end{lemma}}
\newcommand{\bt}{\begin{theorem}}
\newcommand{\et}{\end{theorem}}
\newcommand{\bd}{\begin{definition}}
\newcommand{\ed}{\end{definition}}
\newcommand{\bc}{\begin{corollary}}
\newcommand{\ec}{\end{corollary}}
\newcommand{\bass}{\begin{assumption}}
\newcommand{\eass}{\end{assumption}}
\newcommand{\bexam}{\begin{example}}
\newcommand{\eexam}{\end{example}}

\newcommand{\Label}{\label}
\newcommand{\la}{\Label}

\newcommand{\de}{\textnormal{d}}

\def\wt#1{{{\widetilde #1} }}
\def\wh#1{{{\,\widehat #1\,} }}

      \def\sF{{\mathfrak F}}
   \def\sH{{\mathfrak H}}   
   \def\sK{{\mathfrak K}}   \def\sL{{\mathfrak L}}
   \def\sN{{\mathfrak N}}

   \def\se{{\mathfrak e}}   
   \def\sh{{\mathfrak h}}   
   \def\sk{{\mathfrak k}}   
\def\sn{{\mathfrak m}}   \def\sn{{\mathfrak n}}

      \def\dC{{\mathbb C}}

   \def\dN{{\mathbb N}}   
      \def\dR{{\mathbb R}}

   \def\cB{{\mathcal B}}   
      \def\cF{{\mathcal F}}

\def\cM{{\mathcal M}}      
      
\def\cS{{\mathcal S}}

\def\gD{{\Delta}}			\def\ga{{\alpha}}
\def\gG{{\Gamma}}			\def\gb{{\beta}}
\def\gl{{\lambda}}
\def\gL{{\Lambda}}			\def\gga{{\gamma}}
			\def\gd{{\delta}}
\def\gs{{\sigma}}			\def\go{{\omega}}
\def\gY{{\Upsilon}}                     \def\gk{{\varkappa}}

\def\ess-sup{{\text{ess-sup\,}}}
\def\supp{{\text{\rm supp\,}}}

\def\ps{\mathpzc{s}}
\def\pS{\mathpzc{S}}

\providecommand*{\abs}[1]{\lvert#1\rvert}
\providecommand*{\norm}[1]{\lVert#1\rVert}

\newenvironment{proof}%
{\begin{sloppypar}\noindent{\bf Proof.}}%
{\hspace*{\fill}$\square$\end{sloppypar}\bigskip}

\title{\red{A new model for quantum dot light emitting-absorbing devices}
\\[5mm]
\large{Dedicated to the memory of Pierre Duclos}}

\author{
Hagen Neidhardt
\and
Lukas Wilhelm\\
WIAS Berlin, Mohrenstr. 39, 10117 Berlin,
Germany\\
E-mail: hagen.neidhardt@wias-berlin.de
\and
Valentin A. Zagrebnov\\
Laboratoire d'Analyse, Topologie, Probabilit\'{e}s - UMR 7353\\
CMI - Technop\^{o}le Ch\^{a}teau-Gombert \\
39, rue F. Joliot Curie, 13453 Marseille Cedex 13, France\\
E-mail: Valentin.Zagrebnov@latp.univ-mrs.fr}

\date{\today}

\begin{document}

\maketitle

\begin{abstract}
\noindent
Motivated by the Jaynes-Cummings (JC) model, we consider here a quantum dot coupled simultaneously
to a reservoir of photons and to two electric leads (free-fermion reservoirs).
This Jaynes-Cummings-Leads (JCL) model makes possible that the fermion current through the dot creates
a photon flux, which describes a light-emitting device. The same model is also describe a transformation of
the photon flux into current of fermions, i.e. a quantum dot light-absorbing device.
The key tool to obtain these results is an abstract
Landauer-B\"uttiker formula.
\end{abstract}

\vspace{2mm}


\newpage

\tableofcontents

\section{Introduction}

The Landauer-B\"uttiker formula is widely used for the analysis of the steady
state current flowing trough a quantum device. It goes back to
\cite{Landauer1957} and \cite{Buettiker1985} and was initially
developed based on phenomenological arguments for non-interacting electrons (free-fermions).
The essential idea was to describe a quantum system as an inner or sample system (dot) with left and right
leads attached to it, i.e. free-fermion reservoirs with two different electro-chemical potentials.
The goal was to calculate the steady electron current going from one lead through the dot to
another one.

It was Landauer and B\"uttiker who found that this current is directly related to the
\textit{transmission coefficients} of some natural scattering system related to this particle transport problem.
The phenomenological  approach of Landauer and B\"uttiker later has been justified in several papers by
deriving the formula from fundamental concepts of the Quantum Mechanics, see the series of papers
\cite{Pillet2007,Baro2004,Nenciu2008,Cornean2012b,Cornean2010c,Cornean2005,Cornean2006,CNWZ2012,CNZ2009} and \cite{Nenciu2007}.

Note that this quantum mechanical approach is possible since for the case of free-fermion reservoirs
the corresponding transport problem reduces to study the Hamiltonian dynamics of extended ``one-particle'' system.
During last decade there has been an important progress in rigorous development of the Quantum Statistical
Mechanics of Open Systems \cite{AJPI,AJPII,AJPIII}. This is a many-body approach adapted for interacting systems.
It also allows, besides the Hamiltonian \cite{AJPI}, to develop a  Markovian description of effective microscopic
dynamics of the sample system (dot) connected to environment of external reservoirs \cite{AJPII}.
Then evolution the sample system is governed by a quantum  Master Equation. Although powerful and useful the
Markovian approach needs a microscopic Hamiltonian justification, which is a nontrivial problem \cite{AJPII}.

In the present paper we follow the one-particle quantum mechanical Hamiltonian approach.
Motivated by the quantum optics Jaynes-Cummings (JC) model, we consider here a \textit{two-level} quantum dot coupled
\textit{simultaneously} to environment of \textit{three} external reservoirs. The first is the standard JC
one-mode photon resonator, which makes the JC quantum dot an \textit{open system} \cite{GeKn2005}.
Two others are \textit{free-fermion} reservoirs coupled to the quantum dot. They mimic two electric leads.
This new Jaynes-Cummings-Leads ($JCL$-) model makes possible that the fermion current through the dot creates
a photon flux into the resonator, i.e. it describes a \textit{light-emitting} device.
The same model is also able to describe a transformation of the external photon flux into a current of fermions,
which corresponds to a quantum dot \textit{light-absorbing} device.

The aim of the paper is to analyze the fermion current going through the dot as a function of
electro-chemical potentials on leads and the contact with the  photon reservoir. Although the latter is the
canonical $JC$-interaction, the coupling of the JC model with leads needs certain precautions, if we like to stay
in the framework of one-particle quantum mechanical Hamiltonian approach and the scattering theory.

We discuss the construction of our $JCL$-model in Sections \ref{II.1}-\ref{SecII.6}. For simplicity,
we choose  for the leads Hamiltonians the one-particle discrete Schr\"{o}dinger operators with constant
one-site (electric) potentials on each of leads. Notice that these Hamiltonians are one band bounded self-adjoint
operators. The advantage is that one can easily adjust the leads band spectra positions (and consequently the dot-leads
transmission coefficients) shifting them with respect to the two-point quantum dot spectrum by varying the
one-site electric potentials (voltage). In Section \ref{SecII.4} we show that
the our model fits into framework of \textit{trace-class scattering} and in Section \ref{SecII.6} we verify
the important property that the coupled Hamiltonian has no singular continuous spectrum.

Our main tool is an abstract Landauer-B\"uttiker-type formula applied in Sections \ref{SecIII.1} and \ref{SecIII.2}
to the case of the $JCL$-model. Note that this abstract formula allows to calculate not only the electron current but also fluxes for
other quantities, such as photon or energy/entropy currents. In particular, we calculate the outgoing flux of
photons induced by electric current via leads. This corresponds to a \textit{light-emitting} device. We also found
that pumping the JCL quantum dot by photon flux from resonator may induce current of fermions into leads. This
reversing imitates a quantum \textit{light-absorbing} cell device.
These are the main properties of our model and the main application of the Landauer-B\"uttiker-type formula of
Sections \ref{SecIII.1} and \ref{SecIII.2}. They are presented in Sections \ref{SecIV} and \ref{SecV},
where we distinguish \textit{contact-induced} and \textit{photon-induced} fermion currents.

To describe the results of Sections \ref{SecIV} and \ref{SecV} note that in our setup
the sample Hamiltonian is a \textit{two-level} quantum dot \textit{decoupled} from
the one-mode resonator. Then the unperturbed Hamiltonian $H_{0}$ describes is a collection of four totally decoupled
sub-systems: the sample, the resonator and the two leads. The perturbed Hamiltonian $H$ is a fully coupled
system and the feature of our model is that it is \textit{totally} (i.e. including the leads) \textit{embedded} into
the external electromagnetic field of resonator. This allows a systematic application of the abstract Landauer-B\"uttiker-type
formula, c.f. Sections \ref{SecIII.1} and \ref{SecIII.2}.

As we see there is a variety of possibilities to switch on interactions between sub-systems, i.e. to produce
intermediate Hamiltonians. We distinguish the following two of them:
\begin{enumerate}
\item[(a)] First to switch on the coupling between sample and resonator: the standard JC model $H_{JC}$, see e.g.
\cite{GeKn2005}. Then to connect it to leads, which gives the Hamiltonian $H_{JCL} := H$ of the fully coupled system.

\item[(b)] First to couple the sample to leads: the corresponding Hamiltonian $H_{SL}$ is a standard
``Black Box'' $SL$-model for free-fermion current, see \cite{Pillet2007}, \cite{AJPIII}. Then to embed it into resonator
and to couple the sample with electromagnetic field by the $JC$-interaction. This again produces our $JCL$-model with
$H_{JCL} = H$.
\end{enumerate}

Similar to the $SL$-model $\{H_{SL}, H_{0}\}$, it turns out that the $JCL$-model also fits into the framework of
the abstract Landauer-B\"uttiker formula, and in particular, is a trace-class scattering system $\{H_{JCL}=H, H_{SL}\}$.
The current in the $SL$-model is called the \textit{contact-induced} current $J^c_{el}$. It was a subject of
numerous papers, see e.g. \cite{Pillet2007,Baro2004}, or \cite{AJPIII} and references quoted there. Note that the
current $J_{el}$ is due to the difference of electro-chemical potentials between two leads, but it may be
\textit{zero} even if this difference is not null \cite{Cornean2006,CNWZ2012}.

The fermion current in the $JCL$-model, takes into account the effect of the electron-photon
interaction under the assumption that the leads are already coupled. It is called the \textit{photon-induced}
component $J_{el}^{ph}$ of the total current. Up to our knowledge the present paper is the first, where it is
studied rigorously. We show that the total free-fermion current $J$ in the $JCL$-model decomposes into a sum of
the contact- and the photon-induced currents: $J_{el} := J^c_{el} + J_{el}^{ph}$. An extremal case is, when the
contact-induced current is zero, but the photon-induced component is not, c.f. Section \ref{SecV.1}. In this case
the flux of photons $J_{ph}$ out of the quantum dot (sample) is also non-zero, i.e. the dot serves as the light
emitting device, c.f. Section \ref{sec:LandBuettApplicationPhotonCur}. In general the $J_{ph}\neq 0$ only when the
{photon-induced} component $J_{el}^{ph}\neq 0$.

In this paper we derive  explicit formulas for these currents in the following three cases which are important
for the understanding of the $JCL$-model:
\begin{enumerate}

\item[(i)] The electro-chemical potentials of fermions in the left and right leads are
  equal. Note that in this case the (contact-induced) current in the $JCL$-model is zero.

\item[(ii)] The spectrum of the left and right lead Hamiltonians
  do not overlap. Again, in this case the contact-induced electron current $J^c_{el}$ of the current in the
  $JCL$-model is zero, and only the \textit{photon-induced electron current} $J_{el}^{ph}$ of the total current
  is possible.

\item[(iii)] The leads are coupled to the Jaynes-Cummings model such that left and right leads interact
 only by virtue of the photon interaction in the Jaynes-Cummings model.
Then the contact-induced electron current $J^c_{el}$ is also zero.
\end{enumerate}
For these cases we find that the photon induced electron current
$J^{ph}_{\ga,el}$ entering the left ($\ga = l$) or right ($\ga = r$) lead is given by
\bed
\begin{split}
&J^{ph}_{\ga,el}= -\sum_{\substack{m,n\in\dN_0\\ \gk\in\{l,r\}}}
\frac{\se}{2\pi}\int_\dR d\gl \;\wh{\gs^{ph}_{n_\gk m_\ga}(\gl)}\times\\
& \left(\rho^{ph}(n)f_{FD}(\gl-\mu_\ga-n\go)- \rho^{ph}(m)f_{FD}(\gl-\mu_\gk-m\go)\right).
\end{split}
\eed
where $\wh{\gs^{ph}_{n_\gk m_\ga}(\gl)} \ge 0$ is a partial scattering cross-section
between the left channel with $m$-photons and the $\gk$-channel with
$n$-photons at energy $\gl\in\dR$. By $\se > 0$ the magnitude of the
electron charge is denoted. The photon current is given by
\bed
J_{ph} = \sum_{\substack{m,n \in \dN_0\\
\ga,\gk \in \{l,r\}}}(n-m)\rho^{ph}(m)
\frac{1}{2\pi}\int_\dR d\gl\;f_{FD}(\gl - \mu_\ga - m\go) \wh{\gs^{ph}_{n_\gk
    m_\ga}}(\gl)\, .
\eed
Both formulas become simpler if it is assumed that the $JCL$-model
is time reversible symmetric. In this case we get
\bed
\begin{split}
&J^{ph}_{l,el}= -\sum_{m,n\in\dN_0}
\frac{\se}{2\pi}\int_\dR d\gl \;\wh{\gs^{ph}_{n_{r} m_l}(\gl)}\times\\
& \left(\rho^{ph}(n)f_{FD}(\gl-\mu_l-n\go)- \rho^{ph}(m)f_{FD}(\gl-\mu_{r}-m\go)\right),
\end{split}
\eed
and
\bed
\begin{split}
&J_{ph} = \sum_{\substack{m,n \in \dN_0, n > m\\
\gk,\ga\in \{l,r\}}}
\frac{1}{2\pi}\int_\dR d\gl\;\wh{\gs^{ph}_{n_\gk m_\ga}}(\gl)\times\\
& (n-m)\left(\rho^{ph}(m)f_{FD}(\gl - \mu_\ga - m\go)
  - \rho^{ph}(n)f_{FD}(\gl - \mu_\gk - n\go)\right)\, .
\end{split}
\eed
It turns out that choosing the parameters of the model in an suitable manner one gets either a photon
emitting or a photon absorbing system. Hence $JCL$-model can be used
either as a light emission device or as a light-cell.
Proofs of explicit formulas for fermion and photon currents $J^{ph}_{l,el}$ , $J_{ph}$ is the contents of
Sections \ref{SecIV} and \ref{SecV}.

Note that the $JCL$-model is called \textit{mirror symmetric} if (roughly speaking) one can interchange
left and right leads and the $JCL$-model remains unchanged. In Section \ref{SecV} we discuss a surprising
example of a mirror symmetric $JCL$-model such that the free-fermion current is \textit{zero} but the model
is photon emitting. This peculiarity is due to a specific choice of the \textit{photon-fermion} interaction in
our model.

\section{Jaynes-Cummings quantum dot coupled to leads}	\label{SecII}

\subsection{Jaynes-Cummings model}\la{II.0}

The starting point for construction of our $JCL$-model is the quantum optics Jaynes-Cummings Hamiltonian $H^{JC}$.
Its simplest version is a \textit{two-level} system (quantum dot) with the energy spacing $\varepsilon$,
defined by Hamiltonian $h_{S}$ on the Hilbert space $\sh_{S} = \dC^2$, see e.g. \cite{GeKn2005}. It is assumed that
this system is ``open'' and interacts with the one-mode $ \ \omega \ $ photon resonator with Hamiltonian $h^{ph}$.

Since mathematically $h^{ph}$ coincides with quantum harmonic oscillator, the Hilbert space of the resonator
is the boson Fock space $\sh^{ph} = \sF_+(\dC)$ over $\dC$ and
\begin{equation}\label{phot}
h^{ph} = \omega \, b^\ast b \ .
\end{equation}
Here $b^\ast$ and $b$ are verifying the Canonical Commutation Relations ($CCR$) creation and annihilation
operators with domains in $ \ \sF_+(\dC) \simeq \ell^2(\dN_0)$. Operator (\ref{phot}) is self-adjoint on
its domain
\bed
\dom(h^{ph}) = \left\{ (k_0,k_1,k_2,\ldots) \in \ell^2(\dN_0): \sum_{n\in\dN_0} n^2 |k_n|^2 < \infty \right\}.
\eed
Note that canonical basis $\{\phi_n := (0,0,\ldots,k_n =1, 0, \ldots)\}_{n\in\dN_0}$ in
$\ell^2(\dN_0)$ consists of eigenvectors of operator (\ref{phot}): $ h^{ph}\phi_n = n \omega \, \phi_n$.

To model the \textit{two-level} system with the energy spacing $ \ \varepsilon $, one fixes in $\dC^2$ two
ortho-normal vectors $\{e^S_0,e^S_1\}$, for example
\be\la{states}
e^S_0 :=
\begin{pmatrix}
0\\
1
\end{pmatrix} \ \ \ \ {\rm{and}} \ \ \ \
e^S_1 :=
\begin{pmatrix}
1\\
0
\end{pmatrix} \ \ ,
\quad
\ee
which are eigenvectors of Hamiltonian $h_{S}$ with eigenvalues $\{\gl^S_0 = 0 , \, \gl^S_1 = \varepsilon \}$.
To this end we put
\be\la{h-S}
h_{S}: = \varepsilon
\begin{pmatrix}
1 & 0  \\
0 & 0 \\
\end{pmatrix} \ \ ,
\ee
and we introduce two \textit{ladder} operators:
\be\la{sigma+-}
\sigma^{+}: =
\begin{pmatrix}
0 & 1  \\
0 & 0 \\
\end{pmatrix} \ \ \ \ , \ \ \ \
\sigma^{-}: =
\begin{pmatrix}
0 & 0  \\
1 & 0 \\
\end{pmatrix} \ \ .
\ee
Then one gets $h_{S} = \varepsilon  \ \sigma^{+} \sigma^{-}$ as well as
\begin{equation}\label{ground-state}
e^S_1 = \sigma^{+}e^S_0 \ \ , \ \ e^S_0 = \sigma^{-}e^S_1 \ \ \ {\rm{and}} \ \ \
\sigma^{-}e^S_0 = \begin{pmatrix}
0\\
0
\end{pmatrix} \ \ .
\end{equation}
So, $e^S_0$ is the ground state of Hamiltonian $h_{S}$. Note that \textit{non-interacting}
Jaynes-Cummings Hamiltonian $H^{JC}_{0}$ lives in the space $\sH^{JC} = \sh_{S} \otimes \sh^{ph} =
\dC^2 \otimes \sF_+(\dC)$ and it is defined as the \textit{matrix} operator
\begin{equation}\label{JC-0}
H^{JC}_{0}: = h_{S}\otimes I_{\sh^{ph}} + I_{\sh_{S}} \otimes h^{ph} \ .
\end{equation}
Here $I_{\sh^{ph}}$ denotes the unit operator in the Fock space $\sh^{ph}$, whereas $I_{\sh_{S}}$ stays for
the unit matrix in the space $\sh_{S}$.

With operators (\ref{sigma+-}) the interaction $V_{S b}$  between quantum dot and photons (bosons) in the
resonator is defined (in the rotating-wave approximation \cite{GeKn2005}) by the operator
\begin{equation}\label{V-Sb}
V_{S b}:= g_{S b} \ (\sigma^{+} \otimes b + \sigma^{-} \otimes b^\ast) \ .
\end{equation}

Operators (\ref{JC-0}) and (\ref{V-Sb}) define the Jaynes-Cummings model Hamiltonian
\begin{equation}\label{JC}
H_{JC}:= H^{JC}_{0} + V_{S b} \ ,
\end{equation}
which is self-adjoint operator on the common domain $\dom(H^{JC}_{0})\cap\dom(V_{S b})$. The standard
interpretation of $H_{JC}$ is that (\ref{JC}) describes an ``open'' two-level system interacting with external
one-mode electromagnetic field \cite{GeKn2005}.

Since the one-mode resonator is able to absorb \textit{infinitely} many bosons this interpretation sounds
reasonable, but one can see that the spectrum $\sigma (H^{JC})$ of the Jaynes-Cummings model is \textit{discrete}.
To this end note that the so-called number operator
\begin{equation*}
\mathfrak{N}_{JC}:= \sigma^{+} \sigma^{-} \otimes I_{\sh^{ph}} + I_{\sh_{S}} \otimes b^\ast b
\end{equation*}
commutes with $H_{JC}$. Then, since for any $n \geq 0$
\begin{equation}\label{SubSp}
\sH_{n > 0}^{JC}:=\{\zeta_0 e^S_0 \otimes \phi_n +  \zeta_1 e^S_1 \otimes \phi_{n-1}\}_{\zeta_{0,1}\in\dC} \ , \
\sH_{n = 0}^{JC}:=\{\zeta_0 e^S_0 \otimes \phi_0\}_{\zeta_{0}\in\dC} \ ,
\end{equation}
are eigenspaces of operator $\mathfrak{N}_{JC}$, they reduce $H_{JC}$,
 i.e. $H_{JC}: \sH_{n}^{JC} \rightarrow \sH_{n}^{JC}$. Note that $\sH^{JC} = \bigoplus_{n \geq 0}\sH_{n}^{JC}$,
where each $\sH_{n}^{JC}$ is invariant subspace of operator (\ref{JC}). Therefore, it has the representation
\begin{equation}\label{JC-inv}
H_{JC} = \bigoplus_{n\in \dN_0} H_{JC}^{(n)} \ , \ \ n > 1 \ , \ H_{JC}^{(0)} = 0 \, .
\end{equation}
Here operators $H_{JC}^{(n)}$ are the restrictions of $H_{JC}$, which act in each $\sH_{n}^{JC}$ as
\begin{eqnarray}\label{JC-n}
&&H_{JC}^{(n)}(\zeta_0 \ e^S_0 \otimes \phi_n +  \zeta_1 \ e^S_1 \otimes \phi_{n-1}) = \\
&&[\zeta_0 n \omega + \zeta_1 g_{S b} \sqrt{n}] \ e^S_0 \otimes \phi_n +
[\zeta_1 (\varepsilon + (n-1)\omega) + \zeta_0 g_{S b} \sqrt{n}] \ e^S_1 \otimes \phi_{n-1} \ .
\nonumber
\end{eqnarray}
Hence, the spectrum $\sigma(H_{JC})= \bigcup_{n \geq 0} \sigma(H_{JC}^{(n)})$. By virtue of (\ref{JC-n})
the spectrum $\sigma(H_{JC}^{(n)})$ is defined for $n\geq 1$ by eigenvalues $E(n)$ of two-by-two matrix
$\widehat{H}_{JC}^{(n)}$ acting on the coefficient space $\{\zeta_0, \zeta_1\}$:
\begin{equation}\label{JC-En}
\widehat{H}_{JC}^{(n)} \begin{pmatrix}
\zeta_1 \\
\zeta_0 \\
\end{pmatrix} = \begin{pmatrix}
\varepsilon + (n-1)\, \omega & g_{S b} \sqrt{n} \\
g_{S b} \sqrt{n} & n \omega \\
\end{pmatrix} \begin{pmatrix}
\zeta_1 \\
\zeta_0 \\
\end{pmatrix} = E(n) \begin{pmatrix}
\zeta_1 \\
\zeta_0 \\
\end{pmatrix} \ .
\end{equation}
Then (\ref{JC-inv}) and (\ref{JC-En}) imply that the spectrum of the Jaynes-Cummings model Hamiltonian
$H_{JC}$ is \textit{pure point}:
\begin{eqnarray}\label{spect-JC}
&&\sigma(H_{JC})= \sigma_{p.p.}(H_{JC})= \\
&& \{0\} \cup \bigcup_{n \in \dN}\left\{n \omega + \frac{1}{2}(\varepsilon -\omega)
\pm \sqrt{(\varepsilon -\omega)^2/4 + g_{S b}^2 n}\right\}  \ .\nonumber
\end{eqnarray}

This property is evidently persists for any system Hamiltonian $h_{S}$ with discrete spectrum
and linear interaction (\ref{V-Sb}) with a finite mode photon resonator \cite{GeKn2005}.

\vspace{2mm}
We resume the above observations concerning the Jaynes-Cummings model, which is our starting point,
by following remarks:
\begin{enumerate}
\item[(a)] The standard Hamiltonian (\ref{JC}) describes instead of \textit{flux} only oscillations of photons between
resonator and quantum dot, i.e. the system $h_{S}$ is not ``open''
enough.

\item[(b)] Since one our aim is to model a \textit{light-emitting} device, the system $h_{S}$ needs an \textit{external}
source of energy to pump it into dot, which then be transformed by interaction (\ref{V-Sb}) into the outgoing
\textit{photon current} pumping the resonator.

\item[(c)] To reach this aim we extend the standard Jaynes-Cummings model to our $JCL$-model by attaching to the
quantum dot $h_{S}$ (\ref{h-S}) two \textit{leads}, which are (infinite) reservoirs of \textit{free} fermions.
Manipulating with \textit{electro-chemical} potentials of fermions in these reservoirs we can force one of them
to inject fermions in the quantum dot, whereas another one to absorb the fermions out the quantum dot with the same rate.
This current of fermions throughout the dot would pump it and produce
the photon current according scenario (b).

\item[(d)] The most subtle point is to invent a \textit{leads-dot} interaction $V_{l S}$, which ensures the above mechanism
and which is simple enough that one still be able to treat this $JCL$-model using our extension of the
Landauer-B\"uttiker formalism.
\end{enumerate}


\subsection{The JCL-model}\la{II.1}

First let us make some general remarks and formulate certain conditions indispensable when one
follows the modeling (d).
\begin{enumerate}
\item[(1)] Note that since the Landauer-B{\"u}ttiker formalism \cite{CNWZ2012} is essentially a scattering
theory on a contact between two subsystems, it is developed only on a
``one-particle'' level. This allows to study with this formalism only ideal (\textit{non-interacting})
many-body systems. This condition we impose on many-body fermion systems (electrons) in two leads.
Thus, only direct interaction between different components of the system: dot-photons $V_{S b}$ and
electron-dot $V_{l S}$ are allowed.

\item[(2)] It is well-known that fermion reservoirs are technically simpler to treat then boson ones \cite{CNWZ2012}.
Moreover, in the framework of our model it is also very natural since we study electric current although
produced by ``non-interacting electrons''. So, below we use fermions/electrons as synonymous.

\item[(3)] In spite of precautions formulated above, the first difficulty to consider an ideal many-body system
interacting with quantized electromagnetic field (photons) is induced \textit{indirect} interaction. If
electrons can emit and absorb photons, it is possible for one electron to emit a photon that another
electron absorbs, thus creating the indirect photon-mediated electron-electron interaction. This
interaction makes impossible to develop the Landauer-B\"uttiker formula, which requires non-interacting
framework.
\end{enumerate}

\bass\la{ass:2.1}
{\em
To solve this difficulty we forbid in our model the photon-mediated interaction.
To this end we suppose that every electron (in leads and in dot) interacts with its \textit{own} distinct
copy of the electromagnetic field. So, to consider electrons together with its photon fields as non-interacting
``composed particles'', which allows to apply the Landauer-B{\"u}ttiker approach.
Formally it corresponds to the ``one-electron'' Hilbert space $\sh^{el}\otimes\sh^{ph}$, where $\sh^{ph}$ is
the Hilbert space of the individual photon field. The fermion description of composed-particles
$\sh^{el}\otimes\sh^{ph}$ corresponds to the antisymmetric Fock space $\sF_-(\sh^{el}\otimes\sh^{ph})$.
}
\eass

The composed-particle assumption \ref{ass:2.1} allows us to use the Landauer-B{\"u}ttiker formalism developed
for ideal many-body fermion systems.
Now we come closer to the formal description of our JCL-model with two (infinite) leads and a
one-mode quantum resonator.

Recall that the Hilbert space of the Jaynes-Cummings Hamiltonian with two energy levels is
$\sH^{JC} = \dC^2 \otimes \sF_+(\dC)$. The boson Fock space is constructed from a one-dimensional Hilbert space
since we consider only photons of a single fixed frequency.
We model the electrons in the leads as free fermions living on a discrete semi-infinite lattices. Thus
\begin{equation}\label{space-Ld}
\sh^{el} = \ell^2(\dN) \oplus \dC^2 \oplus \ell^2(\dN) = \sh_l^{el} \oplus \sh_{S} \oplus \sh_r^{el}
\end{equation}
is the one-particle Hilbert space for electrons and for the dot. Here, $\sh_\ga^{el}$, $\ga \in \{l,r\}$,
are the Hilbert spaces of the \textit{left} respectively \textit{right} lead and
$\sh_{S} = \dC^2$ is the Hilbert space of the quantum dot. We denote by
\bed
\{\delta_n^\ga\}_{n\in\dN}, \qquad \{\delta_n^S\}_{j=0}^{1}
\eed
the canonical basis consisting of individual lattice sites of
$\sh_\ga^{el}$, $\ga \in \{l,r\}$, and of $\sh_{S}$, respectively. With the Hilbert space for photons,
$\sh^{ph} = \sF_+(\dC) \simeq \ell^2(\dN_0)$,
we define the Hilbert space of the \textit{full} system, i.e. quantum dot with leads and with the
photon field, as
\begin{equation}\label{full-space-sH}
\sH = \sh^{el} \otimes \sh^{ph} = \big(\ell^2(\dN) \oplus \dC^2 \oplus \ell^2(\dN)\big)\otimes \ell^2(\dN_0).
\end{equation}
\begin{remark}\la{rem:II.H}
{\em
Note that the structure of full space (\ref{full-space-sH}) takes into account the condition \ref{ass:2.1}
and produces composed fermions via the last tensor product. It also manifests that electrons \textit{in the dot}
as well as those \textit{in the leads} are composed with photons. This makes difference with the picture imposed
by the the Jaynes-Cummings model, when \textit{only dot} is composed with photons:
\begin{equation}\label{JC+L}
\sH = \ell^2(\dN) \oplus \dC^2 \otimes \ell^2(\dN_0)\oplus \ell^2(\dN) \ \ \ , \ \ \
\sH^{JC} = \dC^2 \otimes \ell^2(\dN_0) \ ,
\end{equation}
see (\ref{JC-0}), (\ref{V-Sb}) and (\ref{JC}), where $\sH^{JC} = \sh_{S} \otimes \sh^{ph}$.
The next step is a choice of interactions between subsystems:
dot-resonator-leads.
}
\end{remark}

According to (\ref{space-Ld}) the decoupled leads-dot Hamiltonian is the matrix operator
\bed
h^{el}_0 =
\begin{pmatrix}
h_l^{el} & 0 & 0 \\
0 & h_S & 0 \\
0 & 0 & h_r^{el}
\end{pmatrix} \ \ {\rm{on}} \ \  u = \begin{pmatrix}
u_l\\
u_S \\
u_r
\end{pmatrix} \ , \ \{u_{\alpha}\in \ell^2(\dN)\}_{\alpha \in \{l,r\}} \ , \  u_S \in \dC^2 \ ,
\eed
where $h_\ga^{el} = -\Delta^D + v_\ga$ with a constant potential
bias $v_\ga \in \dR$, $\ga\in\{l,r\}$, and $h_{S}$ can be any
self-adjoint two-by-two matrix with eigenvalues $\{\gl^S_0,\gl^S_1 := \gl^S_0 + \varepsilon\}$,
$\varepsilon > 0$, and eigenvectors $\{e^S_0,e^S_1\}$, cf (\ref{h-S}).
Here, $\Delta^D$ denotes the discrete Laplacian on $\ell^2(\dN)$ with
homogeneous Dirichlet boundary conditions given by
\bead
(\gD^D f)(x) & :=  & f(x+1)-2f(x)+f(x-1), \quad x \in \dN,\\
\dom(\gD^D)    & :=   & \{f\in \ell^2(\dN_0): f(0) := 0\},
\eead
which is obviously a bounded self-adjoint operator. Notice that
$\gs(\gD^D) = [0,4]$.

We define the \textit{lead-dot interaction} for coupling $g_{el} \in \dR$ by the matrix operator
acting in (\ref{space-Ld}) as
\begin{equation}\label{Int-D-L}
v_{el} = g_{el}
\begin{pmatrix}
0 & \langle \cdot, \delta_0^S \rangle \delta_1^l & 0 \\
\langle \cdot,\delta_1^l \rangle \delta_0^S & 0 & \langle \cdot, \delta_1^r \rangle \delta_1^S \\
0 & \langle \cdot,\delta_1^S \rangle \delta_1^r & 0
\end{pmatrix} \ \ ,
\end{equation}
where non-trivial off-diagonal entries are \textit{projection} operators in the Hilbert space
(\ref{space-Ld}) with the scalar product $u,v \mapsto \langle u, v \rangle$ for $u,v \in\sh^{el}$.
Here $\{\gd^S_0,\gd^S_1\}$ is ortho-normal basis in $\sh^{el}_S$, which in general may be different
from $\{e^S_0,e^S_1\}$. Hence, interaction (\ref{Int-D-L}) describes quantum \textit{tunneling} between
leads and the dot via contact sites of the leads, which are supports of $\delta_1^l$ and $\delta_1^r$.

Then Hamiltonian for the system of interacting leads and dot we define as  $h^{el} := h_0^{el} + v_{el}$.
Here both $h_0^{el}$ and $h^{el}$ are bounded self-adjoint operators on $\sh^{el}$.

Recall that photon Hamiltonian in the one-mode resonator is defined by operator $h^{ph} = \omega b^\ast b$
with domain in the Fock space $\sF_+(\dC) \simeq \ell^2(\dN_0)$, (\ref{phot}). We denote the canonical basis in
$\ell^2(\dN_0)$ by $\{\Upsilon_n\}_{n\in\dN_0}$. Then for the spectrum of $h^{ph}$ one obviously gets
\be\la{2.1}
\sigma(h^{ph}) = \sigma_{pp}(h^{ph}) = \bigcup_{n\in\dN_0}\{n \omega\}.
\ee

We introduce the following decoupled Hamiltonian $H_0$, which describes
the system when the leads are decoupled from the quantum dot and the
electron does not interact with the photon field.
\be\la{rr2.3}
H_0 := H_0^{el} + H^{ph},
\ee
where
\bed
H^{el}_0 := h_0^{el} \otimes I_{\sh^{ph}}
\quad \mbox{and} \quad
H^{ph} := I_{\sh^{el}} \otimes h^{ph}.
\eed
The operator $H_0$ is self-adjoint on
$\dom(H_0) = \dom(I_{\sh^{el}}\otimes h^{ph})$. Recall that $h_0^{el}$ and
$h^{ph}$ are bounded self-adjoint operators. Hence $H^{el}_0$ and $H^{el}$ are semi-bounded from
below which yields that $H_0$ is semi-bounded from below.

The interaction of the photons and the electrons in the
quantum dot is given by the coupling of the dipole moment of the
electrons to the electromagnetic field in the rotating wave approximation. Namely,
\be\la{rr2.4}
V_{ph} = g_{ph} \left((\cdot,e^S_0)e^S_1 \otimes b + (\cdot,e^S_1)e^S_0 \otimes b^\ast \right)
\ee
for some coupling constant $g_{ph} \in \dR$. The total Hamiltonian is given by
\be\la{rr2.5}
H := H^{el} + H^{ph} + V_{ph} = H_0 + V_{el} + V_{ph},
\ee
where $H^{el} := h^{el} \otimes I_{\sh^{ph}}$ and $V_{el} := v_{el}\otimes I_{\sh^{ph}}$.

In the following we call $\pS = \{H,H_0\}$ the Jaynes-Cummings-leads
system, in short $JCL$-model, which we are going to analyze. In particular, we are interested in
the electron and photon currents for that system. The analysis will be based on the
abstract Landauer-B\"uttiker formula, cf. \cite{Pillet2007,CNWZ2012}.
\bl\la{rII.1}
H is bounded from below self-adjoint such that $\dom(H) =
\dom(H_0)$.
\el
\begin{proof}
Let $c \geq 2$. Then
\begin{equation*}
\norm{b\Upsilon_n}^2 \leq \norm{b^\ast \Upsilon_n}^2 = n+1 \leq c^{-1}n^2 + c, \quad n \in \dN_0.
\end{equation*}
Consider elements $f \in \sh_{S} \otimes \sh^{ph} \cap \dom(I_{\sh^{el}}\otimes h^{ph})$ with
\bed
f = \sum_{j,l}\beta_{jl} e_j \otimes \Upsilon_l, \quad j\in \{0,1\}, \quad l \in \dN_0,
\eed
which are dense in $\sH^{JC} := \sh^{el}_S \otimes h^{ph}$. Then $\norm{f}^2 = \sum_{j,l}\abs{\beta_{jl}}^2$
and $\norm{(I_{\sh^{el}}\otimes b^\ast b) f}^2 = \sum_{j,l=1}\abs{\beta_{jl}}^2 l^2$. We obtain
\bead
\lefteqn{
\norm{((\cdot,e^S_1)e^S_0\otimes b) f}^2 \leq \sum_{j,l}\abs{\beta_{jl}}^2
\norm{b \Upsilon_l}^2 \leq}\\
& &
\sum_{j,l}\abs{\beta_{jl}}^2(c^{-1}l^2 + c)
= c^{-1}\norm{(I_{\sh^{el}}\otimes b^\ast b) f}^2  + c\norm{f}^2
\eead
Similarly,
\bed
\norm{ ((\cdot,e^S_1)e^S_0\otimes b^\ast) f}^2 \leq c^{-1} \norm{(I_{\sh^{el}}\otimes b^\ast b) f}^2 + c \, \norm{f}^2.
\eed
If $c \geq 2$ is large enough, then we obtain that $V_{ph}$ is
dominated by $H^{ph}$ with relative bound less than one. Hence $H$ is self-adjoint and $\dom(H_0) = \dom(H)$.
Since $H_0^{el}$ and $V_{el}$ are bounded and $H^{ph}$ is self-adjoint and
bounded from below, it follows that $H = H_0^{el} + H^{ph} + V_{el} + V_{ph}$
is bounded from below \cite[Thm. V.4.1]{Ka1995}.
\end{proof}

\subsection{Time reversible symmetric systems}\la{II.2}

A system described by the Hamiltonian $H$ is called time reversible
symmetric if there is a conjugation $\gG$ defined on $\sH$ such that
$\gG H = H\gG$. Recall that $\gG$ is a conjugation if the conditions
$\gG^2 = I$ and $(\gG f,\gG g) = \overline{(f,g)}$, $f,g \in \sH$.

Let $\sh^{ph}_n$, $n \in \dN_0$, the subspace spanned by the
eigenvector $\gY_n$ in $\sh^{ph}$. We set
\be\la{2.6aa}
\sH_{n_\ga} := \sh^{el}_\ga \otimes \sh^{ph}_n, \quad n \in \dN_0,
\quad \ga \in \{l,r\}.
\ee
Notice that
\bed
\sH = \bigoplus_{n\in \dN_0,\ga\in\{l,r\}}\sH_{n_\ga}
\eed
\bd\la{II.2a}
{\rm
The $JCL$-model is called time reversible symmetric if there is
a conjugation $\gG$ acting on $\sH$ such that $H$ and $H_0$ are time
reversible symmetric and the subspaces
$\sH_{n_\ga}$, $n \in \dN_0$, $\ga \in \{l,r\}$, reduces $\gG$.
}
\ed
\bexam\la{II.3x}
{\rm
Let $\gga^{el}_\ga$ and $\gga^{el}_S$ be conjugations defined by
\bed
\gga^{el}_\ga f_\ga := \overline{f_\ga} := \{\overline{f_\ga(k)}\}_{k\in\dN}, \quad f_\ga\in \sh^{el}_\ga,
\quad \ga \in \{l,r\},
\eed
and
\bed
\gga^{el}_Sf_S =
\gga^{el}_S\begin{pmatrix}
f_S(0)\\
f_S(1)
\end{pmatrix}
:=
\begin{pmatrix}
\overline{f_S(0)}\\
\overline{f_S(1)}
\end{pmatrix}
\eed
We set $\gga^{el} := \gga^{el}_l \oplus \gga^{el}_S \oplus \gga^{el}_r$. Further, we set
\bed
\gga^{ph}\psi := \overline{\psi} = \{\overline{\psi(n)}\}_{n\in\dN_0}, \quad \psi \in \sh^{ph}.
\eed
We set $\gG := \gga^{el}\otimes \gga^{ph}$. One easily checks that
$\gG$ is a conjugation on $\sH = \sh^{el} \otimes \sh^{ph}$.
}
\eexam
\bl\la{II.4y}
Let $\gga^{el}_\ga$, $\ga \in \{S,l,r\}$, and $\gga^{ph}$ be given by
Example \ref{II.3x}.

\begin{enumerate}
\item[\rm (i)] If the conditions $\gga^{el}_S e^S_0 = e^S_0$ and
$\gga^{el}_S e^S_1 = e^S_1$ are satisfied, then $H_0$ is time reversible
symmetric with respect to $\gG$ and, moreover, the subspaces
$\sH_{n_\ga}$, $n \in \dN_0$, $\ga \in \{l,r\}$, reduces $\gG$.

\item[\rm (ii)]
If in addition the conditions $\gga^{el}_S\gd^S_0 = \gd^S_0$ and
$\gga^{el}_S\gd^S_1 = \gd^S_1$ are satisfied, then $JCL$-model
is time reversible symmetric.
\end{enumerate}
\el
\begin{proof}
(i) Obviously we have
\bed
\gga^{el}_\ga h^{el}_\ga = h^{el}_\ga \gga^{el}_\ga, \quad \ga \in\{l,r\},
\quad \mbox{and} \quad
\gga^{ph} h^{ph} = h^{ph} \gga^{ph}.
\eed
If $\gga^{el}_S e^S_0 = e^S_0$ and $\gga^{el}_S e^S_1 = e^S_1$ is
satisfied, then $\gga^{el}_S h^{el}_S =  h^{el}_S\gga^{el}_S$ which yields
$\gga^{el} h^{el}_0 = h^{el}_0 \gga^{el}$ and, hence, $\gG H_0 = \gG
H_0$. Since $\gga^{el}\sh^{el}_\ga = \sh^{el}_\ga$ and $\gga^{ph}\sh^{ph} =
\sh^{ph}$ one gets $\gG \sH_{n_\ga} = \sH_{n_\ga}$ which shows that
$\sH_{n_\ga}$ reduces $\gG$.

(ii) Notice that $\gga^{el}_\ga\gd^\ga_1 = \gd^\ga_1$, $\ga \in
\{l,r\}$. If in addition the conditions $\gga^{el}_S\gd^S_0 = \gd^S_0$ and
$\gga^{el}_S\gd^S_1 = \gd^S_1$ are satisfied, then $\gga^{el}v_{el} = v_{el}\gga^{el}$ is valid which yields
$\gga^{el}h^{el} = h^{el}\gga^{el}$. Hence $\gG H = H\gG$.
Together with (i) this proves that the $JCL$-model is time reversible
symmetric.
\end{proof}

Choosing
\be\la{2.5a}
e^S_0 :=
\begin{pmatrix}
1\\
0
\end{pmatrix},
\quad
e^S_1 :=
\begin{pmatrix}
0\\
1
\end{pmatrix},
\quad
\gd^S_0 :=
\frac{1}{\sqrt{2}}\begin{pmatrix}
1\\
1
\end{pmatrix},
\quad
\gd^S_1 :=
\frac{1}{\sqrt{2}}\begin{pmatrix}
1\\
-1
\end{pmatrix}
\ee
one satisfies the condition $\gga^{el}_S e^S_0 = e^S_0$ and $\gga^{el}_S e^S_1 = e^S_1$ as well as
$\gga^{el}_S e^S_0 = e^S_0$ and $\gga^{el}_S e^S_1 = e^S_1$.

\subsection{Mirror symmetric systems}\la{II.3}

A unitary operator $U$ acting on $\sH$ is called a mirror symmetry if
the conditions
\bed
U\sH_{n_\ga} = \sH_{n_{\ga'}}, \quad \ga,\ga' \in \{l,r\}, \quad \ga \not=
\ga'
\eed
are satisfied. In particular, this yields $U\sH^{JC} = \sH^{JC}$, $\sH^{JC} :=
\sh^{el}_S \otimes \sh^{ph}$.
\bd\la{II.4a}
{\rm
The $JCL$-model is called mirror symmetric if there is a mirror
symmetry commuting with $H_0$ and $H$.
}
\ed

One easily verifies that if $H_0$ is mirror symmetric, then
\bed
H_{n_{\ga'}}U = UH_{n_\ga}, \quad \quad n \in \dN_0, \quad \ga,\ga' \in
\{l,r\}, \quad \ga\not=\ga',
\eed
where
\bed
H_{n_\ga} := h^{el}_\ga \otimes I_{\sh^{ph}_n} + I_{\sh^{el}_\ga} \otimes h^{ph}_n = h^{el}_\ga + n\go,
\quad n \in \dN_0, \quad \ga,\ga' \in \{l,r\}, \quad \ga\not=\ga'.
\eed
In particular, this yields that $v_\ga = v_{\ga'}$.
Moreover, one gets $UH_S = H_SU$ where $H_S := h^{el}_S \otimes I_{\sh^{ph}} +
I_{\sh^{el}}\otimes h^{ph}$.

Notice that if $H$ and $H_0$ commute with the
same mirror symmetry $U$, then also the operator $H_c := h^{el}
\otimes I_{\sh^{ph}} + I_{\sh^{el}}\otimes h^{ph}$ commutes with $U$,
i.e, is mirror symmetric.
\bexam\la{II.3A}
{\rm
Let $\pS = \{H,H_0\}$ be the $JCL$-model.
Let $v_l = v_r$ and let $e^S_0$ and $e^S_1$ as well as $\gd^S_0$ and
$\gd^S_1$ be given by \eqref{2.5a}. We set
\be\la{2.6a}
u^{el}_S e^S_0 := e^S_0 \quad \mbox{and} \quad u^{el}_S e^S_1 = -e^S_1
\ee
as well as
\be\la{2.7a}
u^{ph}\gY_n = e^{-i n \pi}\gY_n, \quad n \in \dN_0.
\ee
Obviously, $U_S := u^{el}_S \otimes u^{ph}$ defines a unitary operator
on $\sH^{JC}$. A straightforward computation shows that
\be\la{2.8b}
U_S H_S = H_SU_S
\quad  \mbox{and} \quad U_S V_{ph} = V_{ph}U_S.
\ee
Furthermore, we set
\be\la{2.8a}
u_{rl}^{el} \delta_n^l := \delta_n^r, \quad \mbox{and} \quad
u_{lr}^{el} \delta_n^r = \delta_n^{l}, \quad n \in \dN,
\ee
and
\bed
u^{el} :=
\begin{pmatrix}
0 & 0 & u^{el}_{lr}\\
0 & u^{el}_S & 0 \\
u^{el}_{lr} & 0 & 0
\end{pmatrix}.
\eed
We have
\be\la{2.9a}
v_{el}\,u^{el}
\begin{pmatrix}
f_l\\
f_S\\
f_r
\end{pmatrix} =
\begin{pmatrix}
<f_S,(u^{el}_S)^*\gd^S_0>\gd^l_1\\
<f_r,(u^{el}_{lr})^*\gd^l_1>\gd^S_0 + <f_l,(u^{el}_{rl})^*\gd^r_1>\gd^S_1\\
<f_S,(u^{el}_S)^*\gd^S_1>\gd^r_1
\end{pmatrix}
\ee
Since $\gd^S_0 := \frac{1}{\sqrt{2}}(e^S_0+e^S_1)$ and $\gd^S_1 :=
\frac{1}{\sqrt{2}}(e^S_0-e^S_1)$ we get from \eqref{2.6a}
\be\la{2.10a}
(u^{el}_S)^*\gd^S_0 = \gd^S_1 \quad \mbox{and} \quad
(u^{el}_S)^*\gd^S_1 = \gd^S_0.
\ee
Obviously we have
\be\la{2.11a}
(u^{el}_{lr})^*\gd^l_1 = \gd^r_1 \quad \mbox \quad
(u^{el}_{rl})^*\gd^r_1 = \gd^l_1.
\ee
Inserting \eqref{2.10a} and \eqref{2.11a} into \eqref{2.9a} we find
\be\la{2.13a}
v_{el}\,u^{el}
\begin{pmatrix}
f_l\\
f_S\\
f_r
\end{pmatrix} =
\begin{pmatrix}
<f_S,\gd^S_1>\gd^l_1\\
<f_r,\gd^r_1>\gd^S_0 + <f_l,\gd^l_1>\gd^S_1\\
<f_S,\gd^S_0>\gd^r_1
\end{pmatrix}
\ee
us Further we have
\be\la{2.14a}
u^{el}v_{el}\begin{pmatrix}
f_l\\
f_S\\
f_r
\end{pmatrix} =
\begin{pmatrix}
<f_S,\gd^S_1>\gd^l_1\\
<f_l,\gd^l_1>\gd^S_1 + <f_r,\gd^r_1>\gd^S_0\\
<f_s,\gd^S_0>\gd^r_1
\end{pmatrix}\;.
\ee
Comparing \eqref{2.13a} and \eqref{2.14a} we get $u^{el}v_{el} =
v_{el}u^{el}$. Setting $U := u^{el}\otimes u^{ph}$ one immediately
proves that $UH_0 = H_0U$ and $UH = HU$. Since $U\sH_{n_\ga} =
\sH_{n_{ga'}}$ it is satisfied $\pS$ is mirror symmetric.
}
\eexam

Notice that in addition the Example \ref{II.3A} $\pS$ is time reversible
symmetric.

\subsection{Spectral properties of $H$: first part}\la{SecII.4}

In the following our goal is to apply the Landauer-B\"uttiker formula to the
$JCL$-model. By $\sL_p(\sH)$, $1 \le p \le \infty$,  we denote in the following the
Schatten-v.Neumann ideals.
\begin{proposition}\la{II.2A}
If $\pS = \{H,H_0\}$ is the $JCL$-model, then
$(H+i)^{-1} - (H_0+i)^{-1} \in \sL_1(\sH)$. In particular, the absolutely continuous parts $H^{ac}$
and $H^{ac}_0$ are unitarily equivalent.
\end{proposition}
\begin{proof}
We have
\bead
\lefteqn{
(H+i)^{-1}-(H_0+i)^{-1} = (H_0+i)^{-1}V(H+i)^{-1} =}\\
& &
(H_0+i)^{-1}V (H_0+i)^{-1} - (H_0+i)^{-1}V(H_0+i)^{-1}V(H+i)^{-1}
\eead
where $V = H-H_0 = V_{el} + V_{ph}$.
Taking into account Lemma \ref{rII.1} it suffices to prove that
$(H_0+i)^{-1} V (H_0+i)^{-1} \in \sL_1(\sH)$. Using the spectral
decomposition of $h^{ph}$ with respect to $\sh^{ph} = \bigoplus_{n\in\dN_0} \sh^{ph}_n$,
where $\sh^{ph}_n$ are the subspaces spanned by $\gY_n$, we obtain
\be\la{2.5}
(H_0+i)^{-1} = \bigoplus\limits_{{n}\in\dN_0} (h_0^{el} + n\go +i)^{-1} \otimes I_{\sh^{ph}_n}.
\ee
We have $(H_0+i)^{-1} V (H_0+i)^{-1} = (H_0+i)^{-1}(V_{el} + V_{ph})
(H_0+i)^{-1}$. Since $v_{el}$ is a finite rank operator we have
$\lVert v_{el} \rVert_{{\sL_1}} < \infty$. Furthermore,
$\sh^{ph}_n$ is obviously one-dimensional for any $n\in \dN_0$.
Hence $\lVert I_{\sh^{ph}_n} \rVert_{{\sL_1}} = 1$.
From \eqref{2.5} and $V_{el} = v_{el} \otimes I_{\sh^{ph}}$ we obtain
\bed
\begin{aligned}
\lVert (H_0+i)^{-1} V_{el} (H_0+i)^{-1} \rVert_{{\sL_1}}
& = \sum\limits_{{n}\in\dN_0} \lVert(h_0^{el} +  n\omega + i)^{-1} v_{el} (h_0^{el} +  n\omega + i)^{-1}
\rVert_{{\sL_1}} \\
& \leq  \sum\limits_{{n}\in\dN_0} \lVert(h_0^{el} +n\omega + i)^{-2} \rVert \; \lVert v_{el} \rVert_{{\sL_1}}
\end{aligned}
\eed
Since $h_0^{el}$ is bounded we get
\be\la{2.6}
\lVert (h_0^{el} + n\omega + i)^{-1} \rVert = \sup_{\lambda \in \sigma(h_0^{el})}
\big(\sqrt{(\lambda + n\omega)^2 + 1} \big)^{-1} \leq c(n+1)^{-1}
\ee
for some $c>0$. This immediately implies
$\lVert (H_0+i)^{-1} V_{el} (H_0+i)^{-1} \rVert_{{\sL_1}} < \infty $.

We are going to handle $(H_0 + i)^{-1}V_{ph}(H_0 + i)^{-1}$.
Let $p_{n}^{ph}$ be the projection from $\sh^{ph}$ onto $\sh^{ph}_n$. We have
\bead
\lefteqn{
(H_0+i)^{-1}\,(\cdot,e^S_0)e^S_1\otimes b \,(H_0+i)^{-1} }\\
&=&
\sum\limits_{m,n\in\dN_0} (h_0^{el}+m\omega + i)^{-1}(\cdot,e^S_0)e^S_1\,(h_0^{el}+ n\omega + i)^{-1}
\otimes p_{m}^{ph} b p_{n}^{ph}\\
&=&
\sum\limits_{n\in\dN}
(h_0^{el}+ (n-1)\omega + i)^{-1}\,(\cdot,e^S_0)e^S_1\, (h_0^{el}+n\omega + i)^{-1} \otimes \sqrt{n} \gY_{{n-1}}
{\langle \cdot,\gY_{n}\rangle}\\
\eead
From \eqref{2.6} we get
\bed
\begin{split}
\left\|(h_0^{el}+ (n-1)\omega + i)^{-1}\,(\cdot,e^S_0)e^S_1\,
  (h_0^{el}+n\omega + i)^{-1}\big) \right.
&\left.
\otimes \sqrt{n} \gY_{{n}} {\langle
  \cdot,\gY_{n}\rangle}\right\|_{\sL_1}\\
& \le c^2\frac{\sqrt{n}}{n(n+1)},
\end{split}
\eed
$n \in \dN$, which yields
\bed
\lVert (H_0+i)^{-1}\,(\cdot,e^S_0)e^S_1\otimes b \,(H_0+i)^{-1}\rVert_{{\sL_1}} \leq
c^2\sum\limits_{{n}\in\dN}^\infty \frac{\sqrt{n}}{n(n+1)} < \infty.
\eed
Since
\bed
\lVert (H_0+i)^{-1}\,(\cdot,e^S_1)e^S_0\otimes b^\ast\,(H_0+i)^{-1} \rVert_{{\sL_1}}  =
\lVert (H_0+i)^{-1}\,(\cdot,e^S_0)e^S_1 \otimes b\, (H_0+i)^{-1}\rVert_{{\sL_1}}
\eed
one gets $(H_0+i)^{-1} V_{ph} (H_0+i)^{-1} \in \sL_1(\sH)$ which completes the proof.
\end{proof}

Thus, the $JCL$-model $\cS = \{H,H_0\}$ is a $\sL_1$-scattering
system. Let us recall that $h_\ga^{el} = -\Delta^D + v_\ga$,
$\ga \in \{l,r\}$, on $\sh_l^{el} = \sh_r^{el} = \ell^2(\dN)$.
\begin{lemma}\la{II.3Aa}
Let $\ga \in \{l,r\}$. We have
\begin{equation*}
\sigma(h_\ga^{el}) = \sigma_{ac}(h_\ga^{el}) = [v_\ga,4+v_\ga].
\end{equation*}
The normalized generalized eigenfunctions of $h_\ga^{el}$ are given by
\begin{equation*}
g_\ga(x,\lambda) = \pi^{-\frac{1}{2}} (1-(-\lambda+2+v_\ga)^2/4)^{-\frac{1}{4}}
\sin\big(\arccos((-\lambda+2+v_\ga)/2)x\big)
\end{equation*}
for $x \in \dN$, $\lambda \in (v_\ga,4+v_\ga)$.
\end{lemma}
\begin{proof}
We prove the absolute continuity of the spectrum by showing that
\begin{equation*}
\{g_\ga(x,\lambda) \,\vert\, \lambda \in (-2,2) \}
\end{equation*}
is a complete set of generalized eigenfunctions. Note that it suffices to prove the lemma for
\bed
((\Delta^D + 2) f )(x) = f(x+1)+f(x-1), \qquad f(0)=0.
\eed
The lemma then follows by replacing $\lambda$ with
$-\lambda+2+v_\ga$. Let $\lambda \in (-2,2)$ and
\bed
g_{\Delta^D}(x,\lambda) =  \pi^{-\frac{1}{2}} (1-\lambda^2/4)^{-\frac{1}{4}} \sin\big(\arccos(\lambda/2)x\big)
\eed
Note that $g_{\Delta^D}(0,\lambda) = 0$, whence the boundary condition
is satisfied. We substitute $\mu = \arccos(\lambda/2) \in (0,\pi)$, i.e. $\lambda = 2 \cos(\mu)$ and obtain
\bed
\sin(\mu (x+1)) + \sin(\mu(x-1)) = 2 \sin(\mu x) \cos(\mu),
\eed
whence $g_{\Delta^D}(x,\lambda)$ satisfies the eigenvalue
equation. It is obvious that
$g_{\Delta^D}(\cdot,\lambda) \notin \ell^2({\dN_0})$ for $\lambda \in (-2,2)$.
To complete the proof of the lemma, it remains to show the
ortho-normality and the completeness. For the ortho-normality, we have to show that
\bed
\sum_{x \in \dN} g_{\Delta^D}(x,\lambda)g_{\Delta^D}(x,\nu) = \delta(\lambda - \nu).
\eed
Let $\psi \in C_0^\infty\big((-2,2)\big)$. We use the substitution $\mu = \arccos({\nu}/{2})$ and the relation
\bed
\sin(\arccos(y)) = (1-y^2)^{-\frac{1}{2}}
\eed
to obtain
\bed
\begin{aligned}
&\int_{-2}^2 \de\nu \, \sum_{x \in \dN} g_{\Delta^D}(x,\lambda)g_{\Delta^D}(x,\nu)\psi(\nu)\\
&\quad = 2\pi^{-1} \int_{0}^\pi \de\mu \, \sum_{x\in\dN}
\frac{\sin(\mu)\sin\big(\arccos(\lambda/2) x\big)
\sin(\mu x)}{(\sin(\mu))^{\frac{1}{2}}(\sin(\arccos(\lambda/2)))^{\frac{1}{2}}} \psi(2 \cos(\mu)) \\
&\quad = (2\pi)^{-1} \int_{0}^\pi \de\mu \, \sum_{x\in\dN}
\frac{(\sin(\mu))^{\frac{1}{2}}}{(\sin(\arccos(\lambda/2)))^{\frac{1}{2}}} \Big(e^{i (\arccos(\lambda/2)-\mu)x} +\\
&
e^{-i (\arccos(\lambda/2)-\mu)x}- e^{i (\arccos(\lambda/2)+\mu)x} - e^{-i (\arccos(\lambda/2)+\mu)x}\Big) \psi(2 \cos(\mu))
\end{aligned}
\eed
Observe that for the Dirichlet kernel
\bed
\sum_{x \in \dN_0} (e^{i xy} + e^{-i xy}) - 1 = 2\pi \, \delta(y),
\eed
whence
\bed
\begin{aligned}
&\int_{-2}^2 \de\nu \, \sum_{x \in \dN} g_{\Delta^D}(x,\lambda)g_{\Delta^D}(x,\nu)\psi(\nu)\\
&\quad = \int_{0}^\pi \de\mu \,\frac{(\sin(\mu))^{\frac{1}{2}}}{(\sin(\arccos(\lambda/2)))^{\frac{1}{2}}}
\Big( \delta(\arccos(\lambda/2) - \mu) + \\
& \quad\quad\delta(\arccos(\lambda/2) + \mu)\Big) \psi(2 \cos(\mu)) = \psi(\lambda).
\end{aligned}
\eed
In the second equality we use that the summand containing
$\delta(\arccos(\lambda/2)+\mu)$ is zero since both
$\arccos(\lambda/2) > 0$ and $\mu > 0$. Thus, the generalized
eigenfunctions are orthonormal. Finally, using once more the substitution $\mu = \arccos({\nu}/{2})$, we get
\bed
\begin{aligned}
&\int_{-2}^2 \de\nu \, g_{\Delta^D}(x,\nu)g_{\Delta^D}(y,\nu)\\
&\quad = \int_{-2}^2 \de\nu \, \big(1-(\nu/2)^2\big)^{-\frac{1}{2}} \sin\big(\arccos(\nu/2)x\big)
\sin\big(\arccos(\nu/2)y\big)\\
&\quad = 2\pi^{-1} \int_0^\pi \de\mu \, (\sin(\mu))^{-1}\sin(\mu)sin(\mu x)\sin(\mu y) \\
&\quad = \delta_{xy}
\end{aligned}
\eed
for $x,y \in \dN$, whence the family of generalized eigenfunctions is
also complete.
\end{proof}
From these two lemmas we obtain the following corollary that gives us the spectral properties of $H_0$.
\begin{proposition}\la{II.5}
Let $\pS = \{H,H_0\}$ be the $JCL$-model. Then
$\gs(H_0) = \gs_{ac}(H_0) \cup \gs_{pp}(H_0)$, where
\begin{equation*}
\sigma_{ac}(H_0) = \bigcup_{n \in \dN_0} [v_l + n\omega,v_l + 4 + n \omega] \cup [v_r + n \omega,v_r + 4 + n \omega]
\end{equation*}
and
\bed
\sigma_{pp}(H_0) = \bigcup_{n \in \dN_0} \{\gl^S_j + n \omega: j = 0,1\}.
\eed
The eigenvectors are given by $\widetilde g(m, n) = e^S_m \otimes \Upsilon_n$, $m = 0,1$, $n\in \dN_0$.
The generalized eigenfunctions are given by
$\widetilde g_\ga(\cdot,\lambda,n) = g_\ga(\cdot,\lambda-n\omega)\otimes \Upsilon_n$ for $\lambda \in \sigma_{ac}(H_0)$,
$n\in \dN_0$, $\ga \in \{l,r\}$.
\end{proposition}
\begin{proof}
It is well known (see e.g. \cite{Damak2006}) that for two self-adjoint
operators $A$ and $B$ with $\sigma_{sc}(A) = \sigma_{sc}(B) = \emptyset$,
we have $\sigma_{sc}(A\otimes 1 + 1 \otimes B) = \emptyset$,
\bed
\sigma_{ac}(A \otimes 1 + 1 \otimes B) = \big( \sigma_{ac}(A) + \sigma(B) \big) \cup \big( \sigma(A) + \sigma_{ac}(B)\big)
\eed
and
\begin{equation*}
\sigma_{pp}(A \otimes 1 + 1 \otimes B) = \sigma_{pp}(A) + \sigma_{pp}(B).
\end{equation*}
Furthermore, if $\psi_A(\lambda_A)$ and $\psi_B(\lambda_B)$ are
(generalized) eigenfunctions of $A$ and $B$, respectively,
then $\psi_A(\lambda_A)\otimes\psi_B(\lambda_B)$ is a (generalized)
eigenfunction of $A\otimes I + I\otimes B$ for the (generalized)
eigenvalue $\lambda_A + \lambda_B$.

The lemma follows now with $A = h_0^{el}$ and $B = h^{ph}$ using
Lemmata \ref{II.3Aa} and \eqref{2.1}
and the fact that $h_{S}$ has eigenvectors $\{e^S_0,e^S_1\}$
with eigenvalues $\{\gl^S_0, \gl^S_1 = \gl^S_0 + \varepsilon\}$.
\end{proof}

\subsection{Spectral representation}\la{SecII.3}

For the convenience of the reader we define here what we mean under a spectral representation of the
absolutely continuous part
$K^{ac}_0$ of a self-adjoint operator $K_0$ on a separable Hilbert space
$\sK$. Let $\sk$ be an auxiliary separable Hilbert space. We consider the Hilbert space
$L^2(\dR,d\gl,\sk)$. By $\cM$ we define the multiplication operator induced by the independent variable
$\gl$ in $L^2(\dR,d\gl,\sk)$.
Let $\Phi: \sK^{ac}(K_0) \longrightarrow L^2(\dR,d\gl,\sk)$ be an isometry acting from $\sK^{ac}(K_0)$
into $L^2(\dR,d\gl,\sk)$ such that
$\Phi\dom(K^{ac}_0) \subseteq \dom(\cM)$ and
\bed
\cM \Phi f = \Phi K^{ac}_0 f, \quad f \in \dom(K^{ac}_0).
\eed
Obviously, the orthogonal projection $P := \Phi\Phi^*$ commutes with $\cM$ which yields the existence of
a measurable family
$\{P(\gl)\}_{\gl \in \dR}$ such that
\bed
(P\wh f)(\gl) = P(\gl) \wh f(\gl), \qquad \wh f \in L^2(\dR,\gl,\sk).
\eed
We set $L^2(\dR,d\gl,\sk(\gl)) := PL^2(\dR,\gl,\sk)$, $\sk(\gl) :=
P(\gl)\sk$,  and call the triplet
\bed
\Pi(K^{ac}_0) := \{L^2(\dR,d\gl,\sk(\gl)),\cM,\Phi\}
\eed
a spectral representation of $K^{ac}_0$. If $\{L^2(\dR,d\gl,\sk(\gl)),\cM,\Phi\}$ is a spectral
representation of $K^{ac}$, then
$K^{ac}$ is unitarily equivalent $\cM_0 := \cM\upharpoonright L^2(\dR,d\gl,\sk(\gl))$. Indeed, one has
$\Phi K^{ac}_0\Phi^* = \cM_0$. The function $\xi^{ac}_{K_0}(\gl) :=
\dom(\sk(\gl))$, $\gl \in \dR$, is called the spectral multiplicity
function of $K^{ac}_0$. Notice that $0 \le \xi^{ac}_{K_0}(\gl) \le
\infty$ for $\gl \in \dR$.

For $\ga\in\{l,r\}$ the generalized eigenfunctions of $h_\ga^{el}$
define generalized Fourier transforms by $\phi_\ga^{el}: \sh^{el}_\ga = \sh^{el,ac}_\ga(h^{el}_\ga)
\rightarrow L^2([v_\ga,v_\ga+4])$ and
\be\la{3.5}
(\phi_\ga^{el} f_\ga)(\lambda) = \sum_{x \in {\dN_0}}
g_\ga(x,\lambda) f_\ga(x), \quad f_\ga \in \sh^{el}_\ga.
\ee
Setting
\be\la{d2.10}
\sh^{el}_\ga(\gl) :=
\begin{cases}
\dC & \gl \in [v_\ga,v_\ga + 4]\\
0 & \gl \in \dR \setminus [v_\ga,v_\ga + 4].
\end{cases}
\ee
one easily verifies that $\Pi(h^{el}_\ga) = \{L^2(\dR,d\gl,\sh^{el}_\ga(\gl)),\cM,\phi^{el}_\ga\}$
is a spectral representation of $h^{el}_\ga = h^{el,ac}_\ga$, $\ga = l,r$, where we
always assumed implicitly that $(\phi^{el}_\ga f_\ga)(\gl) = 0$ for $\gl \in \dR \setminus
[v_\ga,v_\ga+4]$.  Setting
\be\la{d2.11}
\sh^{el}(\gl) :=
\begin{matrix}
\sh^{el}_l(\gl)\\
 \oplus \\
\sh^{el}_r(\gl)
\end{matrix} \;\;\subseteq \dC^2, \quad \gl \in \dR,
\ee
and introducing the map
\be\la{d2.12}
\phi^{el}: \sh^{el,ac}(h^{el}_0) =
\begin{matrix}
\sh^{el}_l\\
\oplus\\
\sh^{el}_r
\end{matrix} \longrightarrow L^2(\dR,d\gl,\sh^{el}(\gl))
\ee
defined by
\be\la{d2.13}
\phi^{el}f :=
\begin{pmatrix}
\phi^{el}_lf_l\\
\phi^{el}_rf_r
\end{pmatrix}, \quad \mbox{where} \quad
f :=
\begin{pmatrix}
f_l\\
f_r
\end{pmatrix}
\ee
we obtain a spectral representation
$\Pi(h^{el,ac}_0) = \{L^2(\dR,d\gl,\sh^{el}(\gl)),\cM,\phi^{el}\}$ of the absolutely
continuous part $h^{el,ac}_0 = h^{el}_l \oplus h^{el}_r$ of
$h^{el}_0$. One easily verifies that $0 \le \xi^{ac}_{h^{el}_0}(\gl)
\le 2$ for $\gl \in \dR$. Introducing
\be\la{eq:2.40}
\gl^{el}_{\rm min} := \min\{v_l,v_r\}
\quad \mbox{and} \quad
\gl^{el}_{\rm max} := \max\{v_l + 4,v_r + 4\}
\ee
one easily verifies that $\xi^{ac}_{h^{el}_0}(\gl) = 0$ for $\gl \in
\dR \setminus [\gl^{el}_{\rm min},\gl^{el}_{\rm max}]$.

Notice, if $v_r + 4 \le v_l$, then
\bed
\sh^{el}(\gl) =
\begin{cases}
\dC, & \gl \in [v_r,v_r+4] \cup [v_l,v_l+4],\\
\{0\}, & \mbox{otherwise}
\end{cases}
\eed
which shows that $h^{el}_0$ has simple spectrum. In particular, it
holds $\xi^{ac}_{h^{el}_0}(\gl) =1$ for $\gl \in [v_r,v_r+4] \cup
[v_l,v_l+4]$ and otherwise $\xi^{ac}_{h^{el}_0}(\gl) = 0$.

Let us introduce the Hilbert space
$\sh := l^2(\dN_0,\dC^2) = \bigoplus_{n\in\dN_0} \sh_n$, $\sh_n := \dC^2$, $n \in
\dN_0$. Regarding $\sh^{el}(\gl-n\go)$ as a subspace of
$\sh_n$ one regards
\be\la{r2.6}
\sh(\gl) :=
\bigoplus_{n\in\dN_0}\sh_n(\gl), \quad \sh_n(\gl) := \sh^{el}(\gl-n\go),
\quad \gl \in \dR,
\ee
as a measurable family of subspaces in $\sh$. Notice that
$0 \le \dim(\sh(\gl)) < \infty$, $\gl \in \dR$. We consider the Hilbert
space $L^2(\dR,d\gl,\sh(\gl))$.

Furthermore, we introduce the isometric map $\Phi: \sH(H^{ac}_0)
\longrightarrow L^2(\dR,d\gl,\sh(\gl))$ defined by
\be\la{3.6}
(\Phi f)(\lambda) =
\bigoplus_{n\in\dN_0}
\begin{pmatrix}
(\phi_l^{el} f_l(n))(\lambda- n \omega) \\
(\phi_r^{el} f_r(n))(\lambda-n \omega)
\end{pmatrix}, \quad \gl \in \dR
\ee
where
\bed
\bigoplus_{n\in\dN_0}
\begin{pmatrix}
f_l(n)\\
f_r(n)
\end{pmatrix} \in
\bigoplus_{n\in\dN_0}\sh^{el,ac}(h^{el}_0) \otimes \sh^{ph}_n
= \bigoplus_{n\in\dN}
\left(
\begin{array}{c}
\sh^{el}_l \otimes h^{ph}_n\\
\oplus\\
\sh^{el}_r \otimes h^{ph}_n
\end{array}
\right)
\eed
where $\sh_{ph} = \bigoplus_{n\in\dN_0}\sh^{ph}_n$ and $\sh^{ph}_n$ is
the subspace spanned by the eigenvectors $\gY_n$ of $h^{ph}$. One easily verifies
that $\Phi$ is an isometry acting from $\sH^{ac}(H^{ac}_0)$ onto $L^2(\dR,d\gl,\sh(\gl))$.
\bl\la{rII.4}
The triplet $\{L^2(\dR,d\gl,\sh(\gl)),\cM,\Phi\}$ forms a spectral
representation of $H^{ac}_0$, that is, $\Pi(H^{ac}_0) =
\{L^2(\dR,d\gl,\sh(\gl)),\cM,\Phi\}$ where there is a constant $d \in
\dN_0$ such that $0 \le \xi^{ac}_{H_0}(\gl) \le 2d_{\rm
  max}$\, for $\gl \in \dR$  where $d_{\rm max} := \frac{\gl^{el}_{\rm max} - \gl^{el}_{\rm min}}{\go}$
		and $\gl^{el}_{\rm max}$ and $\gl^{el}_{\rm min}$ are given by \eqref{eq:2.40}.
\el
\begin{proof}
It remains to show that $\Phi$ transform $H^{ac}_0$ into the
multiplication operator $\cM$. We have
\bed
H^{ac}_0 f = \bigoplus_{n\in\dN_0}
\begin{pmatrix}
(h^{el}_lf_l)(n) + n\go f_l(n)\\
(h^{el}_rf_r)(n) + n\go f_r(n)
\end{pmatrix}
\eed
which yields
\bead
\lefteqn{
(\Phi H^{ac}_0 f)(\gl)}\\
& &
=
\bigoplus_{n\in\dN_0}
\begin{pmatrix}
(\phi^{el}_l(h^{el}_lf_l)(n))(\gl-n\go) + n\go(\phi^{el}_lf_l(n))(\gl-n\go)\\
(\phi^{el}_r(h^{el}_rf_r)(n))(\gl-n\go) + n\go(\phi^{el}_r f_r(n))(\gl-n\go)
\end{pmatrix}\\
& &
=
\bigoplus_{n\in\dN_0}
\begin{pmatrix}
\gl(\phi^{el}_l f_l(n))(\gl-n\go)\\
\gl(\phi^{el}_r f_r(n))(\gl-n\go)
\end{pmatrix} = (\cM\Phi f)(\gl), \quad \gl \in \dR.
\eead
which proves the desired property.

One easily checks that $\sh(\gl)$ might be only non-trivial if
$\gl-n\go \in [\gl^{el}_{\rm min},\gl^{el}_{\rm max}]$. Hence we get
that $\sh(\gl)$ is non-trivial if the condition
\bed
\frac{\gl - \gl^{el}_{\rm max}}{\go} \le n \le \frac{\gl - \gl^{el}_{\rm
    min}}{\go}
\eed
is satisfied. Hence
\bed
0 \le \xi^{ac}_{H_0}(\gl) \le  2\;\card\left\{n \in \dN_0: \frac{\gl - \gl^{el}_{\rm max}}{\go} \le n \le
\frac{\gl - \gl^{el}_{\rm
    min}}{\go} \right\}, \quad \gl \in \dR.
\eed
or
\bed
0 \le \xi^{ac}_{H_0}(\gl) \le 2\card\left\{n \in \dN_0:
0 \le n \le \frac{\gl^{el}_{\rm max} - \gl^{el}_{\rm
    max}}{\go}\right\},
\quad \gl \in \dR.
\eed
Hence $0 \le \xi^{ac}_{H_0}(\gl) \le d_{\rm max}$ for $\gl \in \dR$.
\end{proof}

In the following we denote the orthogonal projection
from $\sh(\gl)$ onto $\sh_n(\gl)$ by $P_n(\gl)$, $\gl \in \dR$, cf
\eqref{r2.6}.  Since $\sh(\gl) = \bigoplus_{n\in\dN_0} \sh_n(\gl)$ we
have $I_{\sh(\gl)} = \sum_{n\in\dN_0}P_n(\gl)$, $\gl \in
\dR$. Further, we introduce the subspaces
\bed
\sh_{n_\ga}(\gl) := \sh^{el}_\ga(\gl - n\go), \quad \gl \in \dR, \quad
n\in\dN_0.
\eed
Notice that
\bed
\sh_n(\gl) = \bigoplus_{\ga \in \{l,r\}}\sh_{n_\ga}(\gl), \quad \gl
\in \dR, \quad n \in \dN_0.
\eed
By $P_{n_\ga}(\gl)$ we denote the orthogonal projection from $\sh(\gl)$
onto $\sh_{n_\ga}(\gl)$, $\gl \in \dR$. Obviously, we have
$P_n(\gl) = \sum_{\ga \in \{l,r\}}P_{n_\ga}(\gl)$, $\gl \in
\dR$.
\bexam\la{II.7}
{\rm
In general the direct integral $\Pi(H^{ac}_0)$ can be very complicated,
in particular, the structure of $\sh(\gl)$
given by \eqref{r2.6} is difficult to analyze. However, there are
interesting simple cases:
\item[\;\;(i)]
Let $v = v_l = v_r$ and $4 \le \go$. In this case we have
$\sh^{el}(\gl) = \dC^2$ for $[v,v+4]$ and
\bed
\sh(\gl) =
\begin{cases}
\dC^2, & \gl \in [v+n\go,v + n\go + 4], \quad n \in \dN_0,\\
\{0\}, & \mbox{otherwise}.
\end{cases}
\eed

\item[\;\;(ii)]
Let $v_r = 0$, $v_l = 4$, $\go_0 = 4$. Then
\bed
\sh(\gl) =
\begin{cases}
\sh^{el}_r(\gl) = \dC, & \gl \in[0,4),\\
\sh^{el}_{lr}(\gl) = \dC^2, & \gl \in [4,8),\\
\sh^{el}_{rl}(\gl) = \dC^2, & \gl \in [8,12),\\
\cdots
\end{cases}
\eed
where
\bed
\sh^{el}_{\ga\ga'}(\gl) =
\begin{matrix}
\sh^{el}_\ga(\gl)\\
\oplus\\
\sh^{el}_{\ga'}(\gl)
\end{matrix}, \qquad \ga,\ga' \in \{l,r\}, \quad \ga \not= \ga'.
\eed
Hence $\dim(\sh(\gl)) = 2$ for $\gl \ge 4$. \hfill$\Box$
}
\eexam

Let $Z$ be a bounded operator acting on $\sH^{ac}(H_0)$ and commuting  with $H^{ac}_0$.
Since $Z$ commutes with $H^{ac}_0$ there is a measurable family
$\{Z(\gl)\}_{\gl \in \dR}$ of bounded operators acting on $\sh(\gl)$
such that $Z$ is unitarily equivalent to the
multiplication operator induced by $\{Z(\gl)\}_{\gl \in \dR}$ in
$\Pi(H^{ac}_0)$. We set
\bed
Z_{m_\ga n_\gk}(\gl) := P_{m_\ga}(\gl)
Z(\gl)\upharpoonright\sh_{n_\gk}(\gl), \quad \gl \in \dR, \quad
\quad m,n\in \dN_0, \quad
\ga,\gk \in \{l,r\}.
\eed
Let $Z_{m_\ga n_\gk} :=
P_{m_\ga}ZP_{n_\gk}$ where $P_{m_\ga}$ is the
orthogonal projection from $\sH$ onto $\sH_{m_\ga} \subseteq \sH^{ac}(H_0)$,
cf. \eqref{2.6aa}. Obviously, the multiplication operator induced
$\{Z_{m_\ga n_\gk}(\gl)\}_{\gl\in\dR}$ in $\Pi(H^{ac}_0)$ is unitarily
equivalent to $Z_{m_\ga n_\gk}$.

Since by Lemma \ref{rII.4} $\sh(\gl)$ is a finite dimensional space,
the operators $Z(\gl)$ are finite dimensional ones and we can introduce
the quantity
\bed
\gs_{m_\ga n_\gk}(\gl) = \tr(Z_{m_\ga n_\gk}(\gl)^*Z_{m_\ga
  n_\gk}(\gl)), \quad \gl \in \dR, \quad
\quad m,n\in \dN_0, \quad
\ga,\gk \in \{l,r\}.
\eed
\bl\la{II.13}
Let $H_0$ be the self-adjoint operator defined by \eqref{rr2.3} on
$\sH$. Further let $Z$ be a bounded operator on $\sH^{ac}(H_0)$ commuting with $H^{ac}_0$

\item[\;\;\rm (i)] Let $\gG$ be a conjugation on $\sH$, cf. Section \ref{II.2}.
  If $\gG$ commutes with $H_0$ and $P_{n_\ga}$, $n\in \dN_0$, $\ga \in
  \{l,r\}$ and $\gG Z \gG = Z^*$ holds, then
  $\gs_{m_\ga n_\gk}(\gl) = \gs_{n_\gk m_\ga}(\gl)$, $\gl \in \dR$.

\item[\;\;\rm (ii)] Let $U$ be a mirror symmetry on $\sH$. If $U$
  commutes with $H_0$ and $Z$, then
  $\gs_{m_\ga n_\gk}(\gl) = \gs_{m_{\ga'} n_{\gk'}}(\gl)$, $\gl \in \dR$,
  $m,n\in \dN_0$, $\ga,\ga',\gk,\gk' \in \{l,r\}$, $\ga \not= \ga'$,
  $\gk \not= \gk'$.
\el
\begin{proof}
(i) Since $\gG$ commutes with $H_0$ the conjugation $\gG$ is reduce by
$\sH^{ac}(H_0)$. So without loss of generality we assume that $\gG$
acts on $\sH^{ac}(H_0)$. We set $\gG_{n_\ga} :=\gG\upharpoonright\sH_{n_\ga}$.
Notice that
\bed
\gG = \bigoplus_{n\in \dN_0,\ga \in \{l,r\}}\gG_{n_\ga}.
\eed
There is a measurable family $\{\gG(\gl)\}_{\gl \in \dR}$ of
conjugations such that the multiplication operator induced by $\{\gG(\gl)\}_{\gl \in \dR}$
in $\Pi(H^{ac}_0)$ is unitarily equivalent to $\gG$. Moreover, since
$\gG$ commutes with $P_{n_\ga}$ we get that the multiplication
operator induced by the measurable family
\bed
\gG_{n_\ga}(\gl) := \gG(\gl)\upharpoonright\sh_{n_\ga}(\gl), \quad \gl
\in \dR, \quad m\in \dN_0, \quad \ga \in \{l,r\},
\eed
is unitarily equivalent to $\gG_{n_\ga}$. Using $\gG Z \gG = Z^*$ we
get $\gG_{m_\ga}Z_{m_\ga n_\gk} \gG_{n_\gk} = Z_{n_\gk
  m_\ga}^*$. Hence
\be\la{2.25}
\gG_{m_\ga}(\gl) Z_{m_\ga n_\gk}(\gl)  \gG_{n_\gk}(\gl) = Z_{n_\gk
  m_\ga}(\gl)^*, \quad \gl \in \dR.
\ee
If $X$ is trace class operator, then $\tr(\gG X\gG) =
\overline{\tr(X)}$. Using that we find
\bead
\lefteqn{
\gs_{m_\ga n_\gk}(\gl) = \overline{\tr(\gG_{n_\gk}(\gl) Z_{m_\ga n_\gk}(\gl)^*Z_{m_\ga n_\gk}(\gl)\gG_{n_\gk}(\gl))} =}\\
& &
\overline{\tr(\gG_{n_\gk}(\gl) Z_{m_\ga n_\gk}(\gl)^*\gG_{m_\ga}\gG_{m_\ga}Z_{m_\ga n_\gk}(\gl)\gG_{n_\gk}(\gl))}
\eead
From \eqref{2.25} we obtain
\bed
\gs_{m_\ga n_\gk}(\gl) = \overline{\tr(Z_{n_\gk m_\ga}(\gl)Z_{n_\gk
    m_\ga}(\gl)^*)} = \gs_{n_\gk m_\ga}(\gl), \quad \gl \in \dR,
\eed
which proves (i).

(ii) Again without loss of generality we can assume that $U$ acts only
$\sH^{ac}(H_0)$.
Since $U$ commutes with $H_0$ there is a measurable family
$\{U(\gl)\}_{\gl \in \dR}$ of unitary operators acting on $\sh(\gl)$ such that the multiplication operator
induced by $\{U(\gl)\}_{\gl \in \dR}$ is unitarily equivalent to $U$.
Since $U\sH_{n_\ga} = \sH_{n_{\ga'}}$ we have $U(\gl)\sh_{n_\ga}(\gl) = \sh_{n_{\ga'}}(\gl)$,
$\gl \in \dR$. Hence
\bead
\lefteqn{
\gs_{m_\ga n_\gk}(\gl) = \tr(U(\gl) Z_{m_\ga n_\gk}(\gl)^*Z_{m_\ga n_\gk}(\gl)U(\gl)^*)=}\\
& &
\tr(U(\gl) Z_{m_\ga, n_\gk}(\gl)^*U(\gl)^*U(\gl)Z_{m_\ga,
  n_\gk}(\gl)U(\gl)^*).
\eead
Hence
\bed
\gs_{m_\ga n_\gk}(\gl) =
\tr(P_{n_{\gk'}}U(\gl)Z(\gl)^*U(\gl)^*P_{m_{\ga'}}(\gl)U(\gl)Z(\gl)U(\gl)^*P_{n_{\gk'}}(\gl)).
\eed
Since $U$ commutes with $Z$ we find
\bed
\gs_{m_\ga n_\gk}(\gl) =
\tr(P_{n_{\gk'}}Z(\gl)^*P_{m_{\ga'}}(\gl)Z(\gl)P_{n_{\gk'}}(\gl)) =
\gs_{m_{\ga'} n_{\gk'}}(\gl), \quad \gl \in \dR.
\eed
which proves (ii).
\end{proof}

\subsection{Spectral properties of $H$: second part}\la{SecII.6}

Since we have full information on the spectral properties of $H_0$
we can use this to show that $H$ has no singular continuous spectrum.
Crucial for that is the following lemma:
with the help of \cite[Cor. IV.15.19]{Baumgaertel1983}, which
establishes existence and completeness of wave operators and absence
of singular continuous spectrum through a time-falloff method.
We cite it as a Lemma for convenience, with slight simplifications that suffice for our purpose.
\bl[{\cite[Corollary IV.15.19]{Baumgaertel1983}}]\la{II.8}
Let $\{H_0,H\}$ be a scattering system and let $\gL$ be a
closed countable set. Let $F_+$ and $F_-$ be two self-adjoint
operators such that $F_+ + F_- = P_{H_0}^{ac}$ and
\begin{equation*}
s-\lim_{t\rightarrow\infty} e^{\mp itH_0}F_\pm e^{\pm itH_0} = 0.
\end{equation*}
If $(H-i)^{-1}-(H_0-i)^{-1} \in \sL_\infty(\sH)$, $(1-P_{H_0}^{ac})\gamma(H_0) \in \sL_\infty(\sH)$, and
\begin{equation*}
\Big\vert \int_0^{\pm\infty} \de t \big\Vert \big((H_0-i)^{-1}-(H-i)^{-1}\big)e^{-itH_0}\gamma(H_0)
F_\pm \big\Vert \Big\vert < \infty
\end{equation*}
for all $\gamma \in C_0^\infty(\dR \setminus \gL)$,
then $W_\pm(H,H_0)$ exist and are complete and
$\sigma_{sc}(H) = \sigma_{sc}(H_0) = \emptyset$.
Furthermore, each eigenvalue of $H$ and $H_0$ in $\dR \setminus\gL$
is of finite multiplicity and these eigenvalues accumulate at most at points of $\gL$ or at $\pm \infty$.
\el

We already know that the wave operators exist and are
complete since the resolvent difference is trace class. Hence, we need Lemma \ref{II.8}
only to prove the following proposition.
\begin{proposition}	\label{II.9}
The Hamiltonian $H$ defined by \eqref{rr2.5} has no singular continuous
spectrum, that is, $\sigma_{sc}(H) = \emptyset$.
\end{proposition}
\begin{proof}
At first we have to construct the operators $F_\pm$. To this end,
let $\cF : L^2(\dR) \rightarrow L^2(\dR)$ be the usual Fourier
transform, i.e
\bed
{
(\cF f)(\mu) := \wh f(\mu) := \frac{1}{\sqrt{2\pi}}\int_\dR e^{-i\mu x}f(x) dx, \quad
f \in L^2(\dR,dx), \quad \mu \in \dR.
}
\eed
Further, let $\Pi_\pm$ be the orthogonal projection onto
$L^2(\dR_\pm)$ in $L^2(\dR)$. We set
\begin{equation*}
F_\pm = \Phi^\ast \cF \Pi_\pm \cF^\ast \Phi
\end{equation*}
where $\Phi$ is given by \eqref{3.6}. We immediately obtain $F_- + F_+ = P_{ac}(H_0)$. We still have to show that
\begin{equation*}
s-\lim_{t\rightarrow\infty} \norm{e^{\mp itH_0} \Phi^\ast \cF \Pi_\pm \cF^\ast \Phi e^{\pm itH_0}f} = 0
\end{equation*}
for $f \in \sH^{ac}(H_0)$. We prove the relation only for $F_+$ since the proof for $F_-$ is essentially identical.
We have
\bed
\big(\Pi_+ \cF^\ast \Phi e^{itH_0} f\big)(x)
= (2\pi)^{-\frac{1}{2}} \chi_{{\dR_+}}(x) \int_\dR \de\mu \, e^{i(x+t)\mu} \widehat{f}(\mu)
= \chi_{{\dR_+}}(x) \psi(x+t)
\eed
with $\psi = \cF \widehat{f}$. Now
\bead
\lefteqn{
\norm{e^{-itH_0} \Phi^\ast \cF \Pi_+ \cF^\ast \Phi e^{itH_0}f}^2 =
}\\
& &
\norm{\Pi_+ \cF^\ast \Phi e^{itH_0}f}^2 = {\int_{\dR_+} \de x
  \big\vert \psi(x+t) \big\vert^2
= \int_t^\infty \de x \big\vert \psi(x) \big\vert^2
\stackrel{t\rightarrow\infty}{\longrightarrow} 0.}
\eead
Concerning the compactness condition, we already know that
$(H-i)^{-1} - (H_0-i)^{-1} \in \sL_1(\sH) \subset \sL_\infty(\sH)$ from Proposition \ref{II.2A}. Let
\bed
\gL = \bigcup_{n \in \dN_0} \{v_l + n \omega, v_r + n \omega, v_l + 4 + n \omega, v_r + 4 + n \omega\},
\eed
which is closed and countable. We know from Corollary \ref{II.5}
that $H_0$ has no singular continuous spectrum and the eigenvalues
are of finite multiplicity. It follows that
$(1-P_{ac}(H_0))\gamma(H_0)$ is compact for every
$\gamma \in C_0^\infty(\dR \setminus \gL)$. The remaining assumption of Lemma \ref{II.8} is
\bed
\Big\vert \int_0^{\pm\infty} \de t \, \big\lVert \big((H-i)^{-1} -
(H_0-i)^{-1} \big) \gamma(H_0) e^{-i tH_0} F_\pm \big\rVert \Big\vert
< \infty.
\eed
If we can prove this, then we immediately obtain that $H$ has no
singular continuous spectrum. Now
$(H-i)^{-1} - (H_0-i)^{-1} = (H-i)^{-1}(V_{el} + V_{ph})(H_0-i)^{-1}$. But $(H-i)^{-1}$ is bounded,
\bed
\ran(F_\pm) \subset \sH^{ac}(H_0) = (\sh_l^{el} \oplus \sh_r^{el}) \otimes \sh^{ph},
\eed
and $V_{ph} P^{ac}(H_0) = 0$. Also, $V_{el} = v_{el} \otimes I_{\sh^{ph}}$ and
\begin{equation*}
\ker(v_{el})^\bot \subset \dC \delta_1^l \oplus \sh_{S} \oplus \dC\delta_1^r.
\end{equation*}
Hence, it suffices to prove
\bed
\Big\vert \int_0^{\pm\infty} \de t \, \big\lVert P_1^\ga (H_0-i)^{-1} \gamma(H_0) e^{-i tH_0}
F_\pm\big\rVert \Big\vert < \infty,
\eed
$\ga \in \{l,r\}$, where $P_1^\ga = p_1^\ga \otimes I_{\sh^{ph}}$ and
$p_1^\ga$ is the orthogonal projection onto $\sh^{el}_\ga$. In
the following we treat only the case $F_+$. The calculations
for $F_-$ are completely analogous. We use that $\Phi$ maps
$H_0^{ac}$ into the multiplication operator $\cM$ induced
by $\gl$. Hence we get
\bead
\lefteqn{
\big\lVert P_1^\ga \wt\gamma(H_0) e^{-i tH_0} \Phi^\ast \cF f \big\rVert =
\big\lVert P_1^\ga \Phi^\ast \Phi \wt \gga(H_0) e^{-i tH_0} \Phi^\ast \cF f \big\rVert =}\\
& &
= (2\pi)^{-\frac{1}{2}}\Big(\sum_{n \in \dN_0} \, \Big\vert
\int_{\gd_{\ga,n}} \de \lambda \; g_\ga(1,\lambda - n\omega)\wt\gga(\gl)
\int_{\dR_+} \de x \, e^{-i\lambda(x+t)}} f(x)\Big\vert^2\Big)^{\frac{1}{2}
\eead
where $\supp(f) \subseteq \dR_+$, $\wt \gga(\gl) := (\gl-i)^{-1}\gga(\gl)$, $\gl \in \dR$, and
$\gd_{\ga,n} := [v_\ga+ n\go_0,v_\ga+n\go	+ 4]$. Notice that $\wt\gga(\gl) \in C^\infty_0(\dR \setminus \gL)$.
We find
\bead
\lefteqn{
\int_{\gd_{j,n}} \de \lambda \; g_\ga(1,\lambda - n\omega)\wt\gga(\gl)
\int_{\dR_+} \de x \, e^{-i \lambda(x + t)} f(x) =}\\
& &
\int^{v_\ga+4}_{v_\ga} \de \lambda \; g_\ga(1,\lambda)\wt\gga(\gl + n\go)
\int_{\dR_+} \de x \, e^{-i (\lambda + n\go)(x+t)} f(x)
\eead
which yields
\bead
\lefteqn{
\big\lVert P_1^\ga \Phi^\ast \Phi \wt \gga(H_0) e^{-i tH_0} \Phi^\ast \cF f \big\rVert =}\\
& &
= (2\pi)^{-\frac{1}{2}} \Big(\sum_{n \in \dN_0} \, \Big\vert
\int^{v_\ga+4}_{v_\ga} \de \lambda \; g_\ga(1,\lambda)\wt\gga(\gl +
n\go_0)\times\\
& &
\int_{\dR_+} \de x \, e^{-i (\lambda + n\go_0)(x+ t)} f(x)\Big\vert^2\Big)^{\frac{1}{2}}.
\eead
Since the support of $\gga(\gl)$ is compact we get that the sum $\sum_{n \in \dN_0}$ is finite.
Changing the integrals we get
\bead
\lefteqn{
\int_{\gd_{\ga,n}} \de \lambda \; g_\ga(1,\lambda - n\omega)\wt\gga(\gl)
\int_{\dR_+} \de x \, e^{-i\lambda(x + t)} f(x) =}\\
& &
\int_{\dR_+} \de x \, f(x)e^{-in\go_0(x+t)}
\int^{v_\ga+4}_{v_\ga} \de \lambda \; g_\ga(1,\lambda)\wt\gga(\gl + n\go)
 e^{-i\lambda(x+ t)}
\eead
Integrating by parts $m$-times we obtain
\bead
\lefteqn{
\int_{\gd_{\ga,n}} \de \lambda \; g_\ga(1,\lambda - n\omega)\wt\gga(\gl)
\int_{\dR_+} \de x \, e^{-i\lambda(x + t)} f(x) =}\\
& &
(-i)^m\int_{\dR_+} \de x \, f(x)\frac{e^{-in\go(x+t)}}{(x + t)^m}
\int^{v_\ga+4}_{v_\ga} \de \lambda \; e^{-i\lambda(x + t)}\frac{d^m}{d\gl^m}\left(g_\ga(1,\lambda)\wt
\gga(\gl + n\go)\right)
\eead
Hence
\bed
\begin{split}
\big\lvert \int_{\gd_{\ga,n}} \de \lambda \; g_\ga(1,\lambda - n\omega)\wt\gga(\gl)
\int_{\dR_+} &\de x \, e^{-i\lambda(x + t)} f(x)\big\rvert^2 \\
& \le
C^2_n\left(\int_{\dR_+} \de x \, |f(x)|\frac{1}{(x + t)^m}\right)^2
\end{split}
\eed
which yields
\bed
\begin{split}
\big\lvert \int_{\gd_{\ga,n}} \de \lambda \; g_\ga(1,\lambda - n\omega)\wt\gga(\gl)
\int_{\dR_+} &\de x \, e^{-i\lambda(x + t)} f(x)\big\rvert^2 \\
&\le C^2_n\frac{1}{t^{(2m-1)}}\|f\|^2
\end{split}
\eed
for $m \in \dN$ where
\bed
C_n := \int^{v_\ga+4}_{v_\ga} \de \lambda \; \Big\vert\frac{d^m}{d\gl^m}\left(g_\ga(1,\lambda)\wt\gga(\gl + n\go)
\Big\vert\right).
\eed
Notice that $C_n = 0$ for sufficiently large $n \in \dN$. Therefore
\bed
\big\lVert P_1^\ga \wt\gamma(H_0) e^{-i tH_0} \Phi^\ast \cF f \big\rVert \le
\left(\sum_{n\in\dN_0}C^2_n\right)^{1/2}\frac{1}{t^{m-1/2}}\|f\|,
\quad f \in L^2(\dR_+,dx),
\eed
which shows that $\big\lVert P_1^\ga \wt\gamma(H_0) e^{-i tH_0} F_+\big\rVert \in L^1(\dR_+,dt)$ for $m \ge 2$.
\end{proof}
\section{Landauer-B\"uttiker formula and applications}\la{SecIII}

\subsection{Landauer-B\"uttiker formula}\la{SecIII.1}

The abstract Landauer-B\"uttiker formula can be used to calculate
flows through devices. Usually one considers  a pair
$\pS = \{K,K_0\}$ be of self-adjoint operators where the
unperturbed Hamiltonian $K_0$ describes a totally decoupled system,
that means, the inner system is closed and
the leads are decoupled from it, while the perturbed
Hamiltonian $K$ describes the system where the leads are coupled to the inner
system. An important ingredient is system $\pS = \{K,K_0\}$ is
represents a complete scattering or even a trace class scattering
system.

In \cite{Pillet2007} an abstract Landauer-B\"uttiker formula was
derived in the framework of a trace class scattering theory for
semi-bounded self-adjoint operators which allows to reproduce the
results of \cite{Landauer1957} and \cite{Buettiker1985} rigorously.
In \cite{CNWZ2012} the results of \cite{Pillet2007} were generalized
to non-semi-bounded operators.
Following \cite{Pillet2007} we consider a trace class scattering
system $\pS = \{K,K_0\}$. We recall that $\pS = \{K,K_0\}$ is called a trace class
scattering system if the resolvent difference of $K$ and $K_0$ belongs
to the trace class. If $\pS = \{K,K_0\}$ is a trace class scattering
system, then the wave operators $W_\pm(K,K_0)$ exists and are
complete. The scattering operator is defined by
$S(K,K_0) := W_+(K,K_0)^*W_-(K,K_0)$. The main
ingredients besides the trace class scattering system $\pS = \{K,K_0\}$ are the
density and the charge operators $\rho$ and $Q$, respectively.

The density operator
$\rho$ is a non-negative bounded self-adjoint operator commuting with
$K_0$. The charge $Q$ is a bounded self-adjoint operator commuting
also with $K_0$. If $K$ has no singular continuous spectrum, then the current related to the density
operator $\rho$ and the charge $Q$ is defined by
\be\la{1.0a}
J^\pS_{\rho,Q} = -i\,\tr\left(W_-(K,K_0)\rho W_-(K,K_0)^*[K,Q]\right)
\ee
where $[K,Q]$ is the commutator of $K$ and $Q$. In fact, the
commutator $[K,Q]$ might be not defined. In this case the regularized definition
\be\la{3.2}
J^\pS_{\rho,Q} = -i\,\tr\left(W_-(K,K_0)(I+ K^2_0)\rho W_-(K,K_0)^*\frac{1}{K-i}[K,Q]\frac{1}{K+i}\right)
\ee
is used where it is assumed that $(I + K^2_0)\rho$ is a bounded
operator. Since the condition $(H-i)^{-1}[H,Q](H+ i)^{-1} \in \sL_1(\sH)$
is satisfied the definition \eqref{3.2} makes sense. By $\sL_1(\sH)$ is
the ideal of trace class operators is denoted.

Let $K_0$ be self-adjoint operator on the separable Hilbert space
$\sK$. We call $\rho$ be a density operator for $K_0$ if $\rho$ is a
bounded non-negative self-adjoint operator
commuting with $K_0$. Since $\rho$ commutes with $K_0$ one gets that $\rho$ leave invariant the
subspace $\sK^{ac}(K_0)$. We set
\bed
\rho_{ac} := \rho\upharpoonright\sK^{ac}(K_0).
\eed
call $\rho_{ac}$ the $ac$-density part of $\rho$.

A bounded self-adjoint operator $Q$ commuting with
$K_0$ is called a charge.
If $Q$ is a charge, then
\bed
Q_{ac} := Q\upharpoonright\sK^{ac}(K_0).
\eed
is called its $ac$-charge part.

Let $\Pi(K^{ac}_0) = \{L^2(\dR,d\gl,\sk(\gl)),\cM,\Phi\}$
be a spectral representation of $K^{ac}_0$. If $\rho$ is a density
operator, then there is a measurable family
$\{\rho_{ac}(\gl)\}_{\gl\in\dR}$ of bounded self-adjoint operators
such that the multiplication operator
\bed
(\cM_{\rho_{ac}}\wh f)(\gl) := \rho_{ac}(\gl)\wh f(\gl), \quad \wh f \in \dom(M_{\rho_{ac}}) := L^2(\dR,d\gl,\sk(\gl)),
\eed
is unitarily equivalent to $ac$-part $\rho_{ac}$, that is,
$\cM_{\rho_{ac}} = \Phi \rho_{ac} \Phi^*$. In particular this yields that
$\ess-sup_{\gl\in\dR}\|\rho_{ac}(\gl)\|_{\cB(\sk(\gl)} = \|\rho_{ac}\|_{\cB(\sK^{ac}(K_0))}$.
In the following we call $\{\rho_{ac}(\gl)\}_{\gl\in\dR}$ the density
matrix of $\rho_{ac}$.

Similarly, one gets that if $Q$, then
there is a measurable family $\{Q_{ac}(\gl)\}_{\gl\in\dR}$  of bounded self-adjoint operators
such that the multiplication operator
\bead
(\cM_{Q_{ac}}\wh f)(\gl) & := & Q_{ac}(\gl)\wh f(\gl),\\
\wh f\in \dom(Q_{ac}) & := & \{ f \in L^2(\dR,d\gl,\sk(\gl)): Q_{ac}(\gl)\wh f(\gl) \in L^2(\dR,d\gl,\sk(\gl))\},
\eead
is unitarily equivalent to $Q_{ac}$, i.e. $\cM_{Q_{ac}}
= \Phi Q_{ac}\Phi^*$. In particular, one has
\be\la{3.1}
\ess-sup_{\gl \in\dR}\|Q_{ac}(\gl)\|_{\cB(\sk(\gl))} = \|Q_{ac}\|_{\cB(\sK^{ac}(K_0))}.
\ee
If $Q$ is a charge, then
the family $\{Q_{ac}(\gl)\}_{\gl\in\dR}$ is called the charge matrix
of the $ac$-part of $Q$.

Let $\pS = \{K,K_0\}$ be a trace scattering system.
By $\{S(\gl)\}_{\gl\in\dR}$ we denote the scattering matrix which
corresponds to the scattering operator $S(K,K_0)$ with respect to the
spectral representation $\Pi(K^{ac}_0)$.
The operator $T := S(K,K_0) - P^{ac}(K_0)$ is called the transmission
operator. By $\{T(\gl)\}_{\gl \in \dR}$ we denote the transmission
which is related to the transmission operator. Scattering and
transmission matrix are related by $S(\gl) = T_{\sk(\gl)} + T(\gl)$
for a.e. $\gl \in \dR$. Notice
that $T(\gl)$ belongs for to the trace class  a.e. $\gl \in \dR$.
\bt[{\cite[Corollary 2.14]{CNWZ2012}}]\la{rrII.1}
Let $\pS := \{K,K_0\}$ be a trace class scattering system and let
$\{S(\gl)\}_{\gl \in \dR}$ be the scattering matrix of $\pS$ with respect to the spectral
representation $\Pi(K^{ac}_0)$. Further let $\rho$ and $Q$ be density
and charge operators and let
$\{\rho_{ac}(\gl)\}_{\gl\in\dR}$ and $\{Q_{ac}(\gl)\}_{\gl\in\dR}$
be the density and charge matrices of the $ac$-parts
$\rho_{ac}$ and charge $Q_{ac}$ with respect to $\Pi(K^{ac}_0)$, respectively.
If $(I + K^2_0)\rho$ is bounded, then the current $J^\pS_{\rho,Q}$
defined by \eqref{3.2} admits the representation
\be\la{3.7}
J^\pS_{\rho,Q} = \frac{1}{2\pi} \int_\dR
\tr\big(\rho_{ac}(\lambda)(Q_{ac}(\lambda) -
S^\ast(\lambda)Q_{ac}(\lambda)S(\lambda))\big) \de\lambda
\ee
where the integrand on the right hand side and  the current
$J^\pS_{\rho,Q}$ satisfy the estimate
\bea\la{3.7c}
\lefteqn{
\left |\tr\left(\rho_{ac}(\lambda)(Q_{ac}(\lambda)
-S^\ast(\lambda)Q_{ac}(\lambda)S(\lambda))\right)\right| \le}\\
& &
4\|\rho(\gl)\|_{\sL(\sk(\gl))}\|T(\gl)\|_{\sL_1(\sk(\gl))}\|Q(\gl)\|_{\sL(\sk(\gl))}
\nonumber
\eea
for a.e. $\gl \in \dR$ and
\be\la{3.7b}
|J^\pS_{\rho,Q}| \le C_0\|(H+ i)^{-1} - (H_0 + i)^{-1}\|_{\sL_1(\sK)}
\ee
where $C_0 := \frac{2}{\pi}\|(1 +H^2_0)\rho\|_{\sL(\sK)}$.
\et

In applications not every charge $Q$ is a bounded operator. We say the
self-adjoint operator $Q$ commuting with $K_0$ is a $p$-tempered charge if $Q(H_0-i)^{-p}$
is a bounded operator for $p \in \dN_0$. As above we can introduce
$Q_{ac} := Q \upharpoonright\dom(Q) \cap \sK^{ac}(K_0)$. It turns out
that $QE_{K_0}(\gD)$ is a bounded  operator for any bounded Borel set $\gD$.
This yields that the corresponding charge matrix
$\{Q_{ac}(\gl)\}_{\gl \in \dR}$ is a measurable family of bounded
self-adjoint operators such that
\bed
\ess-sup_{\gl \in \dR}(1 +\gl^2)^{p/2}\|Q_{ac}(\gl)\|_{\sL(\sk(\gl))} < \infty.
\eed
To generalize the current $J^\pS_{\rho,Q}$ to tempered charges $Q$ one
uses the fact that $Q(\gD) := QE_{K_0}(\gD)$ is a charge for any
bounded Borel set $\gD$. Hence the current $J^\pS_{\rho,Q(\gD)}$ is
well-defined by \eqref{3.2} for any bounded Borel set $\gD$. Using
Theorem \ref{rrII.1} one gets that for $p$-tempered charges the
limit
\be\la{3.4xx}
J^\pS_{\rho,Q} := \lim_{\gD \to \dR}J^\pS_{\rho,Q(\gD)}
\ee
exists provided $(H_0 - i)^{p+2}\rho$ is a bounded operator. This
gives rise for the following corollary.
\bc\la{III.2}
Let the assumptions of the Theorem \ref{rrII.1} be satisfied. If
for some $p \in \dN_0$ the operator
$(H_0 - i)^{p+2}\rho$ is bounded and $Q$ is a
$p$-tempered charge for $K_0$, then the current defined by
\eqref{3.4xx} admits the representation \eqref{3.7} where the right
hand side of \eqref{3.7} satisfies the estimate
\eqref{3.7c}. Moreover, the current $J^{\pS}_{\rho,Q}$ can be
estimated by
\be
|J^\pS_{\rho,Q}| \le C_p \|(H+ i)^{-1} - (H_0 +i)^{-1}\|_{\sL_1(\sK)}
\ee
where $C_p := \frac{2}{\pi}\|(1+H^2_0)^{p+2/2}\rho\|_{\sL(\sK)}\|Q(I+
H^2_0)^{-p/2}\|_{\sL(\sK)}$.
\ec

At first glance the formula \eqref{3.7} is not very similar to the original
Landauer-B\"uttiker formula of \cite{Buettiker1985,Landauer1957}.
To make the formula more convenient  we recall that
a standard application example for the Landauer-B\"uttiker formula is the
so-called black-box model, cf. \cite{Pillet2007}. In this case the
Hilbert space $\sK$ is given by
\be\la{1.1a}
\sK = \sK_S \oplus \bigoplus^N_{j=1} \sK_j, \quad 2 \le N < \infty.
\ee
and $K_0$ by
\be\la{1.1b}
K_0 = K_S \oplus \bigoplus^N_{j=1} K_j, \quad 2 \le N < \infty.
\ee
The Hilbert space $\sK_S$ is called the sample or dot and $K_S$ is the
sample or dot Hamiltonian. The Hilbert spaces $\sK_j$ are called reservoirs or leads
and $K_j$ are the reservoir or lead Hamiltonians. For simplicity we assume that
the reservoir Hamiltonians $K_j$ are absolutely continuous and the
sample Hamiltonian $K_S$ has point spectrum.
A typical choice for the density operator is
\be\la{1.2}
\rho = f_S(K_S) \oplus \bigoplus^N_{j=1} f_j(K_j),
\ee
where $f_S(\cdot)$ and $f_j(\cdot)$ are non-negative bounded Borel
functions, and for the charge
\be\la{1.3}
Q = g_S(H_s) \oplus \bigoplus^N_{j=1} g_j(H_j),
\ee
where $g_S(\cdot)$ and $g_j(\cdot)$ a bounded Borel
functions. Making this
choice the Landauer-B\"uttiker formula \eqref{3.7} takes the form
\be\la{1.3a}
J^{\pS}_{\rho,Q} = \frac{1}{2\pi} \sum^N_{j,k=1} \int_\dR (f_j(\gl) -
f_k(\gl))g_j(\gl)\gs_{jk}(\gl)d\gl
\ee
where
\be\la{1.6}
\gs_{jk}(\gl) := \tr(T_{jk}(\gl)^*T_{jk}(\gl)), \quad j,k =1,\ldots,N,
\quad \gl \in \dR,
\ee
are called the {\it total transmission probability} from reservoir $k$
to reservoir $j$, cf. \cite{Pillet2007}. We call it the
{\it cross-section} of the scattering process going from channel $k$
to channel $j$ at energy $\gl \in \dR$. $\{T_{jk}(\gl)\}_{\gl \in\dR}$ is
called the transmission matrix from channel $k$ to channel $j$ at energy $\gl \in
\dR$ with respect to the spectral representation $\Pi(K^{ac}_0)$.
We note  that $\{T_{jk}(\gl)\}_{\gl \in \dR}$ corresponds to
the transmission operator
\be\la{1.5}
T_{jk} :=  P_jT(K,K_0)P_k, \quad T(K,K_0) := S(K,K_0) - P^{ac}(K_0),
\ee
acting from the reservoir $k$ to reservoir $j$ where $T(K,K_0)$ is
called the transmission operator. Let $\{T(\gl)\}_{\gl\in\dR}$ be the
transmission matrix.  Following \cite{Pillet2007} the current
$J^\pS_{\rho,Q}$ given by \eqref{1.3a} is directed from the reservoirs into
the sample.

The quantity $\|T(\gl)\|_{\sL_2} = \tr(T(\gl)^*T(\gl))$ is
well-defined and is called the cross-section of the scattering system
$\pS$ at energy $\gl \in \dR$. Notice that
\bed
\gs(\gl) = \|T(\gl)\|_{\sL_2} = \tr(T(\gl)^*T(\gl)) = \sum^N_{j,k=1}\gs_{jk}(\gl).
\quad \gl \in \dR,
\eed
We point out that the channel cross-sections $\gs_{jk}(\gl)$
admit the property
\be\la{3.9}
\sum^N_{j=1}\gs_{jk}(\gl) = \sum^N_{j=1}\gs_{kj}(\gl), \quad \gl \in \dR,
\ee
which is a consequence of the unitarity of the scattering matrix.
Moreover, if there is a conjugation $J$ such that $KJ = JK$ and $K_0J = JK_0$
holds, that is, if the scattering system $\pS$ is time reversible
symmetric, then we have even more, namely, it holds
\be\la{3.10}
\gs_{jk}(\gl) = \gs_{kj}(\gl), \quad \gl \in \dR.
\ee

Usually the Landauer-B\"uttiker formula
\eqref{1.3a} is used to calculated the electron current entering the
reservoir $j$ from the sample. In this case one has to choose $Q := Q^{el}_j :=
-\se P_j$ where $P_j$ is the orthogonal projection form $\sK$ onto
$\sK_j$ and $\se > 0$ is the magnitude of the elementary charge. This
is equivalent to choose $g_j(\gl) = -\se$ and
$g_k(\gl) = 0$ for $k \not= j$, $\gl \in \dR$. Doing so we get the
Landauer-B\"uttiker formula simplifies to
\be\la{1.7}
J^{\pS}_{\rho,Q^{el}_j} = -\frac{\se}{2\pi} \sum^N_{k=1} \int_\dR (f_j(\gl) -
f_k(\gl))\gs_{jk}(\gl)d\gl.
\ee
To restore the original Landauer-B\"uttiker formula one sets
\be\la{1.9}
f_j(\gl) = f(\gl- \mu_j), \quad \gl \in \dR,
\ee
where $\mu_j$ is the chemical potential of the reservoir $\sK_j$ and
$f(\cdot)$ is a bounded non-negative Borel function called the distribution function. This
gives to the formula
\be\la{1.10}
J^{\pS}_{\rho,Q^{el}_j} = -\frac{\se}{2\pi} \sum^N_{k=1} \int_\dR (f(\gl - \mu_j) -
f(\gl-\mu_k))\gs_{jk}(\gl)d\gl.
\ee
In particular, if we choose one
\be\la{5.2}
f(\gl) := f_{FD}(\gl) := \frac{1}{1 + e^{\gb\gl}}, \quad \gb > 0,
\quad \gl \in \dR,
\ee
where $f_{FD}(\cdot)$ is the Fermi-Dirac distribution function, and
inserting \eqref{5.2} into \eqref{1.10} we obtain
\be\la{1.12}
J^{\pS}_{\rho,Q^{el}_j} = -\frac{\se}{2\pi} \sum^N_{k=1} \int_\dR (f_{FD}(\gl - \mu_j) -
f_{FD}(\gl-\mu_k))\gs_{jk}(\gl)d\gl.
\ee
If we have only two reservoirs, then they are usually denoted by $l$ (left)
and $r$ (right). Let $j = l$ and $k = r$.  Then
\be\la{1.13}
J^{\pS}_{\rho,Q^{el}_l} = -\frac{\se}{2\pi}\int_\dR (f_{FD}(\gl - \mu_l) -
f_{FD}(\gl-\mu_r))\gs_{lr}(\gl)d\gl.
\ee
One easily checks that $J^{\pS}_{\rho,Q_l} \le 0$ if $\mu_l \ge
\mu_r$. That means, the current is leaving the left reservoir and is
entering the right one which is accordance with physical intuition.
\bexam\la{III.3}
{\rm
Notice that $\ps_c := \{h^{el},h^{el}_0\}$ is a $\sL_1$
scattering system. The Hamiltonian $h^{el}$   takes into
account the effect of coupling of reservoirs or leads $\sh_l :=
l^2(\dN)$ and $\sh_r := l^2(\dN)$ to the sample $\sh_S = \dC^2$
which is also called the quantum dot.
The leads Hamiltonian are given by $h^{el}_\ga = -\gD^D + v_\ga$, $\ga
= l,r$. The sample or quantum dot Hamiltonian is given by $h^{el}_S$.
The wave operators are given by
\be\la{3.6x}
w_\pm(h^{el},h^{el}_0) :=
\slim_{t\to\infty}e^{ith^{el}}e^{-ith^{el}_0}P^{ac}(h^{el}_0)
\ee
The scattering operator is given by $s_c :=
w_+(h^{el},h^{el}_0)^*w_-(h^{el},h^{el}_0)$.
Let $\Pi(h^{el,ac}_0)$ the spectral representation of $h^{el,ac}_0$
introduced in Section \ref{SecII.3}. If $\rho^{el}$ and $q^{el}$ are
density and charge operators for $h^{el}_0$, then the
Landauer-B\"uttiker formula takes the form
\be\la{eq:3.18}
J^{\ps_c}_{\rho^{el},q^{el}} = \frac{1}{2\pi}\int_\dR
\tr\left(\rho^{el}_{ac}(\gl)\left(q^{el}_{ac} -  s_c(\gl)^*q^{el}_{ac}(\gl)s_c(\gl)\right)\right)
\ee
where $\{s_c(\gl)\}_{\gl\in\dR}$, $\{q^{el}(\gl)\}_{\gl\in\dR}$ and
$\{\rho^{el}(\gl)\}_{\gl\in\dR}$ are the scattering, charge and
density matrices with respect to $\Pi(h^{el,ac}_0)$, respectively.
The condition that $((h^{el}_0)^2 + I_{\sh^{el}})\rho^{el}$ is a
bounded operator is superfluous because $h^{el}_0$ is a bounded
operator. For the same reason we have that every $p$-tempered charge
$q^{el}$ is in fact a charge, that means, $q^{el}$ is a bounded
self-adjoint operator.

The scattering system $\ps_c$ is a black-box model with reservoirs
$\sh^{el}_l$ and $\sh^{el}_r$. Choosing
\bed
\rho^{el} = f_l(h^{el}) \oplus f_S(h^{el}_S) \oplus f_r(h^{el}_r)
\eed
where $f_\ga(\cdot)$, $\ga = l,r$, are bounded Borel functions, and
\bed
q^{el} = g_l(h^{el}_l) \oplus g_S(h^{el}_S) \oplus g_r(h^{el}_r),
\eed
where $g_\ga(\cdot)$, $\ga \in \{l,r\}$, are locally bounded Borel functions, then
from \eqref{1.3a} it follows that
\bed
J^{\ps_c}_{\rho^{el},q^{el}} = \frac{1}{2\pi}
\sum_{\substack{
\ga,\gk\in\{l,r\}\\
\ga\not=\gk
}}
\int_\dR (f_\ga(\gl)-f_\gk(\gl))g_\ga(\gl)\gs_c(\gl)d\gl
\eed
where $\{\gs_c(\gl)\}_{\gl \in \dR}$ is the channel cross-section from
left to right and vice versa. Indeed, let  $\{t_c(\gl)\}_{\gl \in\dR}$
the transition matrix which corresponds to the transition operator
$t_c := s_c - I_{\sh^{el}}$. Obviously, one has
$t_c(\gl) = I_{\sh(\gl)} - s_c(\gl)$, $\gl \in \dR$.
Let $\{p^{el}_\ga(\gl)\}_{\gl\in\dR}$ be the matrix which corresponds
to the orthogonal projection $p^{el}_\ga$ from $\sh^{el}$ onto
$\sh^{el}_\ga$. Further, let $t^c_{rl}(\gl) := p^{el}_r(\gl)t_c(\gl)p^{el}_l$ and
$t^c_{lr} := p^{el}_l(\gl)t_c(\gl)p^{el}_r$. Notice that both
quantities are in fact scalar functions.
Obviously, the channel cross-sections $\gs^c_{lr}(\gl)$ and
$\gs^c_{rl}(\gl)$ at energy $\gl \in \dR$ are given by
$\gs_c(\gl) := \gs^c_{lr}(\gl)  = |t^c_{lr}(\gl)|^2 = |t^c_{rl}(\gl)|^2 = \gs^c_{rl}(\gl)$, $\gl \in
\dR$.

In particular,  if $g_l(\gl) = 1$ and $g_r = 0$, then
\be\la{3.6a}
J^{\ps_c}_{\rho^{el},q^{el}_l} = \frac{1}{2\pi} \int_\dR (f_l(\gl) - f_r(\gl))\gs_c(\gl)d\gl,
\ee
and $q^{el}_l := p^{el}_l$. Following \cite{Pillet2007}
$J^{\ps_c}_{\rho^{el},q^{el}_l}$ denotes the current entering
the quantum dot from the left lead.
}
\eexam

\subsection{Application to the $JCL$-model}\la{SecIII.2}

Let $\pS = \{H,H_0\}$ be now the $JCL$-model.
Further, let $\rho$ and $Q$ be a density operator and a charge
for $H_0$, respectively.  Under these assumptions the current
$J^\pS_{\rho,Q}$ is defined by
\be\la{3.7a}
J^\pS_{\rho,Q} := -i\tr\left(W_-(H,H_0)(I + H^2_0)\rho W_-(H,H_0)^*\frac{1}{H-i}[H,Q]\frac{1}{H+i}\right),
\ee
and admits representation \eqref{3.7}. If $Q$ is a $p$-tempered charge
and $(H_0 - i)^{p+2}\rho$ is a bounded operator, then
the current $J^\pS_{\rho,Q}$ is defined in accordance with \eqref{3.4xx}
and the Landauer-B\"uttiker formula \eqref{3.7} is valid, too.

We introduce the intermediate
scattering system $\pS_c := \{H,H_c\}$ where
\bed
H_c := h^{el} \otimes I_{\sh^{ph}} + I_{\sh^{el}} \otimes h^{ph} = H_0 + V_{el}.
\eed
The Hamiltonian $H_c$  describes the coupling
of the leads to quantum dot but under the assumption that
the photon interaction is not switched on.

Obviously, $\pS_{ph} := \{H,H_c\}$ and $\pS_c := \{H_c,H_0\}$ are $\sL_1$-scattering systems. The corresponding
scattering operators are denote by $S_{ph}$ and $S_c$,
respectively. Let $\Pi(H^{ac}_c) =\{L^2(\dR,d\gl,\sh_c(\gl)),\cM,\Phi_c\}$ of $H^{ac}_c$
be a spectral representation of $H_c$.
The scattering matrix of the scattering system $\{H,H_c\}$ with respect to $\Pi(H^{ac}_c)$ is denoted by
$\{S_{ph}(\gl)\}_{\gl \in \dR}$.
The scattering matrix of the scattering system $\{H_c,H_0\}$ with
respect to
$\Pi(H^{ac}_0) = \{L^2(\dR,d\gl,\sh_0(\gl)),\cM,\Phi_0\}$ is denoted by
$\{S_c(\gl)\}_{\gl\in\dR}$.

Since $\pS_c$ is a $\sL_1$-scattering system the wave operators $W_\pm(H_c,H_0)$ exists and are complete and
since $\Phi_cW_\pm(H_c,H_0)\Phi^*_0$ commute with $\cM$, there is a
measurable families $\{W_\pm(\gl)\}_{\gl \in \dR}$ of isometries acting
from $\sh_0(\gl)$ onto $\sh_c(\gl)$ for a.e. $\gl \in \dR$ such that
\bed
(\Phi_cW_\pm(H_c,H_0)\Phi^*_0\wh f)(\gl) = W_\pm(\gl)\wh f(\gl), \quad
\gl \in \dR, \quad \wh f \in L^2(\dR,d\gl,\sh_0(\gl)).
\eed
The families $\{W_\pm(\gl)\}_{\gl\in\dR}$ are called wave matrices.

A straightforward computation shows that
$\wh{S_{ph}} := W_+(H_c,H_0)^*S_{ph}W_+(H_c,H_0)$ commutes with
$H_0$. Hence, with respect to the spectral representation
$\Pi(H^{ac}_0)$ the operator $\wh{S_{ph}}$ is unitarily equivalent
to a multiplication induced by a measurable family
$\{\wh{S_{ph}}(\gl)\}_{\gl\in\dR}$ of unitary operators in
$\sh_0(\gl)$. A straightforward computation shows that
\be\la{r3.1}
\wh{S_{ph}(\gl)} = W_+(\gl)^*S_{ph}(\gl)W_+(\gl)
\ee
for a.e. $\gl \in \dR$. Roughly speaking,
$\{\wh{S_{ph}}(\gl)\}_{\gl\in\dR}$
is the scattering matrix of $S_{ph}$ with respect to the spectral
representation $\Pi(H^{ac}_0)$.

Furthermore, let
\be\la{rr3.2}
\rho^c := W_-(H_c,H_0)\rho W_-(H_c,H_0)^*
\ee
and
\be\la{rr3.3}
Q^c := W_+(H_c,H_0)Q W_+(H_c,H_0)^*.
\ee
The operators $\rho^c$ and $Q^c$ are density and tempered charge operators for
the scattering system $\pS_{ph}$. Indeed, one easily verifies that
$\rho^c$ and $Q^c$ are commute with $H_c$. Moreover, $\rho^c$ is
non-negative. Furthermore, if $Q$ is a charge, then $Q^c$ is
a charge, too.
This gives rise to introduce the currents $J^c_{\rho,Q} :=  J^{\pS_c}_{\rho,Q}$,
\be\la{3.11x}
J^c_{\rho,Q} := -i\tr\left(W_-(H_c,H_0)\rho W_-(H_c,H_0)^*\frac{1}{H_c-i}[H_c,Q]\frac{1}{H_c+i}\right),
\ee
and $J^{ph}_{\rho,Q} :=  J^{\pS_{ph}}_{\rho^c,Q^c}$
\be\la{3.12x}
J^{ph}_{\rho,Q} := -i\tr\left(W_-(H,H_c)\rho^cW_-(H,H_c)^*\frac{1}{H-i}[H,Q^c]\frac{1}{H+i}\right)
\ee
which are well defined. If $Q$ is $p$-tempered charge and $(H_0 - i)^{p+2}\rho$ is a bounded
operator, then one easily checks that $Q^c$ is a $p$-tempered charge and $(H_c -
i)^{p+2}\rho^c$ is a bounded operator. Hence the definition of the
currents $J^{\pS_c}_{\rho^c,Q^c}$ can be extended to this case and the
Landauer-B\"uttiker formula \eqref{3.7} holds.

Finally we note that the corresponding matrices $\{\rho^c_{ac}(\gl)\}_{\gl\in\dR}$
and $\{Q^c_{ac}(\gl)\}_{\gl\in\dR}$ are related to the matrices
$\{\rho_{ac}(\gl)\}_{\gl\in\dR}$ and $\{Q_{ac}(\gl)\}_{\gl\in\dR}$ by
\be\la{rr3.4}
\rho^c_{ac}(\gl) = W_-(\gl)\rho_{ac}(\gl)W_-(\gl)^*
\quad \mbox{and} \quad
Q^c_{ac}(\gl) = W_+(\gl)Q_{ac}(\gl)W_+(\gl)^*
\ee
for a.e. $\gl \in \dR$.
\begin{proposition}[Current decomposition]\la{dIII.1}
Let $\pS = \{H,H_0\}$ be the $JCL$-model.
Further, let $\rho$ and $Q$ be a density operator and a $p$-tempered
charge, $p \in \dN_0$, for $H_0$, respectively.
If $(H_0 - i)^{p+2}\rho$ is a bounded operator, then the decomposition
\be\la{cc3.2}
J^\pS_{\rho,Q} = J^c_{\rho,Q} + J^{ph}_{\rho,Q}
\ee
holds where $J^c_{\rho,Q}$ and $J^{ph}_{\rho,Q}$ are given
by \eqref{3.11x} and \eqref{3.12x}.

In particular, let  $\{S_c(\gl)\}_{\gl\in\dR}$,
$\{\rho_{ac}(\gl)\}_{\gl\in\dR}$ and
$\{Q_{ac}(\gl)\}_{\gl\in\dR}$ be scattering, density and charge
matrices of $S_c$, $\rho$ and $Q$ with respect to $\Pi(H^{ac}_0)$ and let
$\{S_{ph}(\gl)\}_{\gl\in\dR}$, $\{\rho^c_{ac}(\gl)\}_{\gl\in\dR}$ and
$\{Q^c_{ac}(\gl)\}_{\gl\in\dR}$ be the scattering, density and charge
matrices of the scattering operator $S_{ph}$, density operator
$\rho^c$, cf. \eqref{rr3.2}, and charge operator $Q^c$,
cf. \eqref{rr3.3}, with respect to the spectral representation
$\Pi(H^{ac}_c\}$. Then the representations
\bea
J^c_{\rho,Q} & := &\frac{1}{2\pi}\int_\dR\tr(\rho_{ac}(\gl)(Q_{ac}(\gl) - S_c(\gl)^*Q_{ac}(\gl)S_c(\gl))d\gl,\la{cc3.3}\\
J^{ph}_{\rho,Q} & := & \frac{1}{2\pi}\int_\dR \tr(\rho^c_{ac}(\gl)(Q^c_{ac}(\gl) - S_{ph}(\gl)^*Q^c_{ac}
(\gl)S_{ph}(\gl)))d\gl,\la{cc3.4}
\eea
take place.
\end{proposition}
\begin{proof}
Since $\pS_c$ and $\pS_{ph}$ are $\sL_1$-scattering systems from
Theorem \ref{rrII.1} the representations \eqref{cc3.3} and
\eqref{cc3.4} are easily follow. Taking into account \eqref{rr3.4} we get
\bead
\lefteqn{
\tr(\rho^c_{ac}(\gl)(Q^c_{ac}(\gl) - S_{ph}(\gl)^*Q^c_{ac}(\gl)S_{ph}(\gl))) =}\\
& &
\tr(W_-(\gl)\rho_{ac}W_-(\gl)^*(W_+(\gl)Q_{ac}(\gl)W_+(\gl) - S_{ph}(\gl)^*Q^c_{ac}(\gl)S_{ph}(\gl))).
\eead
Using $S_c(\gl) = W_+(\gl)^*W_-(\gl)$ we find
\bea\la{d3.8}
\lefteqn{
\tr(\rho^c_{ac}(\gl)(Q^c_{ac}(\gl) - S_{ph}(\gl)^*Q^c_{ac}(\gl)S_{ph}(\gl))) =
\tr\left(\rho_{ac}(\gl)\times\right.
}\\
& &
\hspace{-3mm}\left.\left(S_c(\gl)^*Q_{ac}(\gl)S_c(\gl) -
W_-(\gl)^*S_{ph}(\gl)^*W_+(\gl)Q_{ac}(\gl)W_+(\gl)^*S_{ph}(\gl)W_-(\gl)\right)\right).
\nonumber
\eea
Since $\{H_c,H_0\}$ and $\{H,H_c\}$ are $\sL_1$-scattering systems the
existence of the wave operators
$W_\pm(H,H_c)$ and $W_\pm(H_c,H_0)$ follows. Using the chain rule we find
$W_\pm(H,H_0) = W_\pm(H,H_c)W_\pm(H_c,H_0)$ which yields
\bead
S & = & W_+(H,H_0)^*W_+(H,H_0)\\
  & = & W_+(H_c,H_0)^*W_+(H,H_c)W_-(H,H_c)W_-(H_c,H_0)\nonumber\\
  & = & W_+(H_c,H_0)^*S_{ph}W_-(H_c,H_0). \nonumber
\eead
Hence the scattering matrix $\{S(\gl)\}_{\gl\in\dR}$ of $\{H,H_0\}$ admits the
representation
\be\la{c3.4}
S(\gl) = W_+(\gl)^*S_{ph}(\gl)W_-(\gl), \quad \gl \in \dR.
\ee
Inserting \eqref{c3.4} into \eqref{d3.8} we get
\be\la{d3.10}
J^{ph}_{\rho,Q} =
\frac{1}{2\pi}\int_\dR \tr(\rho_{ac}(\gl)(S_c(\gl)^*Q_{ac}(\gl)S_c(\gl) - S(\gl)^*Q_{ac}(\gl)S(\gl)))d\gl
\ee
Using \eqref{d3.10} we obtain
\bed
J^c_{\rho,Q} + J^{ph}_{\rho,Q} =
\frac{1}{2\pi}\int_\dR \tr(\rho_{ac}(\gl)(Q_{ac}(\gl) -
S(\gl)^*Q_{ac}(\gl)S(\gl)))d\gl.
\eed
Finally, taking into account \eqref{3.7} we obtain \eqref{cc3.2}.
\end{proof}
{
\samepage
\begin{remark}\la{rem:III.1}
{\rm
\item[\;\;(i)]
The current $J^c_{\rho,Q}$ is due to the coupling of the leads to the quantum dot
and is therefore called the {\it contact induced current}.
\item[\;\; (ii)]
The current $J^{ph}_{\rho,Q}$ is due to the interaction of photons with electrons
and is therefore called the {\it photon induced current}. Notice the this current is calculated under the assumption
that the leads already contacted to the dot.
}
\end{remark}
}
\bc\la{cor:III.1}
Let the assumptions of Proposition \ref{dIII.1} be satisfied. With
respect to the spectral representation $\Pi(H^{ac}_0)$ of $H^{ac}_0$
the photon induced current $J^{ph}_{\rho,Q}$ can be represented by
\be\la{d3.11}
J^{ph}_{\rho,Q} :=  \frac{1}{2\pi}\int_\dR \tr(\wh{\rho_{ac}(\gl)}(Q_{ac}(\gl) -
\wh{S_{ph}(\gl)^*}Q_{ac}(\gl)\wh{S_{ph}}(\gl)))d\gl
\ee
where the measurable families $\{\wh{S_{ph}(\gl)}\}_{\gl\in\dR}$ and
$\{\wh{\rho_{ac}(\gl)}\}_{\gl\in\dR}$ are given by \eqref{r3.1} and
\be\la{d3.12}
\wh{\rho_{ac}(\gl)} := S_c(\gl)\rho_{ac}(\gl)S_c(\gl)^* \quad \gl \in \dR,
\ee
respectively.
\ec
\begin{proof}
Using \eqref{rr3.4} and $S_c(\gl) = W_+(\gl)^*W_-(\gl)$ we find
\bead
\lefteqn{
\tr(\rho^c_{ac}(\gl)(Q^c_{ac}(\gl) -
S_{ph}(\gl)^*Q^c_{ac}(\gl)S_{ph}(\gl))) =
\tr\left(S_c(\gl)\rho_{ac}(\gl)S_c(\gl)^*\right.\times}\\
& &
\left.\left(Q_{ac}(\gl) - W_+(\gl)^*S_{ph}(\gl)^*W_+(\gl)Q_{ac}(\gl)W_+(\gl)^*S_{ph}(\gl)W_+(\gl)\right)\right).
\eead
Taking into account the representations \eqref{r3.1} and \eqref{d3.12}
we get
\bead
\lefteqn{
\tr(\rho^c_{ac}(\gl)(Q^c_{ac}(\gl) -
S_{ph}(\gl)^*Q^c_{ac}(\gl)S_{ph}(\gl))) =}\\
& &
\tr(S_c(\gl)\rho_{ac}(\gl)S_c(\gl)^*(Q_{ac}(\gl) - \wh{S_{ph}(\gl)^*}Q_{ac}(\gl)\wh{S_{ph}(\gl)}))
\eead
which immediately yields \eqref{d3.11}.
\end{proof}
\begin{remark}
{\rm
In the following we call $\{\wh{\rho_{ac}(\gl)}\}_{\gl\in\dR}$,
cf. \eqref{d3.12}, the photon modified electron density matrix. Notice that
$\{\wh{\rho_{ac}(\gl)}\}_{\gl\in\dR}$ might be non-diagonal
even if the electron density matrix
$\{\rho_{ac}(\gl)\}_{\gl\in\dR}$ is diagonal.
}
\end{remark}
\section{Analysis of currents}\la{SecIV}

In the following we analyze currents $J^c_{\rho,Q}$ and $J^{ph}_{\rho,Q}$ under the assumption that
$\rho$ and $Q$ have the tensor product structure
\be\la{4.1}
\rho = \rho^{el} \otimes \rho^{ph} \quad \mbox{and} \quad Q = q^{el} \otimes q^{ph}
\ee
where $\rho^{el}$ and $\rho^{ph}$ as well as $q^{el}$ and $q^{ph}$
are density operators and (tempered) charges for $h^{el}_0$ and $h^{ph}$, respectively.
Since $\rho^{ph}$ commutes with $h^{ph}$, which is discrete, the operator $\rho^{ph}$has the form
\be\la{4.4}
\rho^{ph} = \rho^{ph}(n)(\cdot,\gY_n)\gY_n, \quad n \in \dN_0,
\ee
where $\rho^{ph}(n)$ are non-negative numbers. Similarly, $q^{ph}$ can be represented by
\be
q^{ph} = q^{ph}(n)(\cdot,\gY_n)\gY_n, \quad n \in \dN_0,
\ee
where $q^{ph}(n)$ are real numbers.
\bl\la{IV.1}
Let $\pS = \{H,H_0\}$ be the $JCL$-model.
Assume that $\rho \not= 0$ and $Q$ have the structure
\eqref{4.1} where $\rho^{el}$ is a density operator
and $q^{el}$ is a charge for $h^{el}_0$.

\item[\rm (i)]
The operator $(H_0 - i)^{p+2}\rho$, $p \in \dN_0$, is bounded if and
only if  the condition
\be\la{4.5}
\sup_{n\in\dN_0} \rho^{ph}(n) n^{p+2} < \infty
\ee
is satisfied.

\item[\rm (ii)]
The charge $Q$ is $p$-tempered if and only if
\be\la{4.6}
\sup_{n\in\dN}|q^{ph}(n)|n^{-p} < \infty.
\ee
is valid
\el
\begin{proof}
(i) The operator $(H_0 - i)^{p+2}\rho$ admits the representation
\bed
(H_0 - i)^{p+2}\rho = \bigoplus_{p\in\dN_0} \rho^{ph}(n)(h^{el}_0 + n\go - i)^{p+2}\rho^{el}.
\eed
We have
\bea\la{4.7}
\|(H_0 - i)^{p+2}\rho\|_{\sL(\sH)} & = & \sup_{p\in\dN_0}\rho^{ph}(n)\|(h^{el}_0 + n\go - i)^{p+2}
\rho^{el}\|_{\sL(\sh^{el})}\\
& = & \sup_{p\in\dN_0}\rho^{ph}(n)n^{p+2}n^{-(p+2)}\left\|(h^{el}_0 + n\go - i)^{p+2}\rho^{el}\right\|_{\sL(\sh^{el})}.
\nonumber
\eea
Since $\lim_{n\to\infty}n^{-(p+2)}\left\|(h^{el}_0 + n\go - i)^{p+2}\rho^{el}\right\|_{\sL(\sh^{el})} =
\go^{p+2}\|\rho^{el}\|_{\sL(\sh^{el})}$
we get for sufficiently large $n \in \dN_0$ that
\bed
\frac{\go^{p+2}}{2}\|\rho^{el}\|_{\sL(\sh^{el})} \le n^{-(p+2)}\|(h^{el}_0 + n\go -i)^{p+2}\rho^{el}\|_{\sL(\sh^{el})}\,.
\eed
Using that and \eqref{4.7} we immediately obtain \eqref{4.5}.
Conversely, from \eqref{4.7} and \eqref{4.5} we obtain that $(H_0 -
i)^{p+2}\rho$ is a bounded operator.

(ii) As above we have
\bed
Q(H_0 - i)^{-p} = \bigoplus_{n\in\dN_0} q^{ph}(n) q^{el}
\eed
Hence
\bed
\|Q(H_0 - i)^{-p}\|_{\sL(\sH)} = \sup_{n\in\dN_0}|q^{ph}(n)|\|q^{el}(h^{el}_0 +n\go - i)^{-p}\|_{\sL(\sh^{el})}.
\eed
Since $\lim_{n\to\infty}n^{p}\|(h^{el}_0 +n\go - i)^{-p}\|_{\sL(\sh^{el})} = \go^{-p}\|q^{el}\|_{\sL(\sh^{el})}$
we get similarly as above that \eqref{4.6} holds. The converse is obvious.
\end{proof}

\subsection{Contact induced current}

Let us recall that $\pS_c = \{H_c,H_0\}$ is
a $\sL_1$-scattering system. An obvious computations shows that
\bed
W_\pm(H_c,H_0) = w_\pm(h^{el},h^{el}_0)\otimes I_{\sh^{ph}}
\eed
where $w_\pm(h^{el},h^{el}_0)$ is given by \eqref{3.6x}. Hence
\bed
S_c = s_c \otimes I_{\sh^{ph}}, \quad \mbox{where} \quad s_c :=
w_+(h^{el}_c,h^{el}_0)^*w_-(h^{el}_c,h^{el}_0).
\eed
\begin{proposition}\la{dIII.2}
Let $\pS = \{H,H_0\}$ be the $JCL$-model. Assume
that $\rho$ and $Q$ are given by \eqref{4.1} where $\rho^{el}$ and
$q^{el}$ are density  and charge operators for $h^{el}_0$ and
$\rho^{ph}$ and $q^{ph}$ for $h^{ph}$, respectively.
If for  some $p\in \dN_0$ the conditions \eqref{4.5} and \eqref{4.6}
are satisfied, then the current $J^c_{\rho,Q}$ is well defined and
admits the representation
\be\la{eq:3.17}
J^c_{\rho,Q}  = \gga J^{\ps_c}_{\rho^{el},q^{el}}, \qquad \gga
:= \sum_{n\in\dN_0} q^{ph}(n)\rho^{ph}(n)
\ee
where $J^{\ps_c}_{\rho^{el},q^{el}}$ is defined by \eqref{3.2}. In particular, if
$\tr(\rho^{ph}) =1$ and $q^{ph} = I_{\sh^{ph}}$, then $J^c_{\rho,Q} = J^{\ps_c}_{\rho^{el},q^{el}}$.
\end{proposition}
\begin{proof}
First of all we note that by lemma \ref{IV.1} the operator $(H_0 -i)^{p+2}\rho$
is bounded and $Q$ is $p$-tempered. Hence the current
$J^{\pS_c}_{\rho,Q}$ is correctly defined and the Landauer-B\"uttiker
formula \eqref{3.7} is valid.

With respect to the spectral representation $\Pi(H^{ac}_0)$ of Lemma
\ref{rII.4} the charge matrix $\{Q_{ac}(\gl)\}_{\gl\in\dR}$ of
$Q_{ac} = q^{el}_{ac}\otimes q^{ph}$ admits the representation
\be\la{r3.4}
Q_{ac}(\gl) = \bigoplus_{n\in\dN_0} q^{el}_{ac}(\gl- n\go)q^{ph}(n), \quad \gl \in \dR.
\ee
Since $S_c = s_c \otimes I_{\sh^{ph}}$ the scattering matrix
$\{S_c(\gl)\}_{\gl \in \dR}$ admits the representation
\bed
S_c(\gl) = \bigoplus_{n\in\dN_0} s_c(\gl-n\go), \quad \gl \in \dR.
\eed
Hence
\bea\la{4.9x}
\lefteqn{
Q_{ac}(\gl) - S_c(\gl)^*Q_{ac}(\gl)S_c(\gl) = }\\
& &
\bigoplus_{n\in\dN_0} q^{ph}(n)\left(q^{el}_{ac}(\gl- n\go) -s_c(\gl - n\go)^*q^{el}_{ac}(\gl-\go n)
s_c(\gl - n\go)\right).
\nonumber
\eea
Moreover, the density matrix $\{\rho_{ac}(\gl)\}_{\gl\in\dR}$ admits
the representation
\be\la{r3.5}
\rho_{ac}(\gl) = \bigoplus_{n\in\dN_0} \rho^{ph}(n)\rho^{el}_{ac}(\gl-n\go)
\ee
Inserting \eqref{r3.5} into \eqref{4.9x} we find
\bead
\lefteqn{
\rho^{ac}(\gl)\left(Q_{ac}(\gl) - S_c(\gl)^*Q_{ac}(\gl)S_c(\gl)\right)
=
\bigoplus_{n\in\dN_0} q^{ph}(n)\rho^{ph}(n)\times}\\
& &
\rho^{el}_{ac}(\gl-n\go)\left(q^{el}_{ac}(\gl-\go n) -s_c(\gl -
  n\go)^*q^{el}_{ac}(\gl-\go n)s_c(\gl - n\go)\right)
\eead
Since $\gga = \sum_{n\in\dN_0}q^{ph}(n)\rho^{ph}(n)$ is absolutely
convergent by \eqref{4.5} and \eqref{4.6} we obtain that
\bea\la{4.11}
\lefteqn{
\tr\left(\rho^{ac}(\gl)\left(Q_{ac}(\gl) - S_c(\gl)^*Q_{ac}(\gl)S_c(\gl)\right)\right)
=
\sum_{n\in\dN_0} q^{ph}(n)\rho^{ph}(n)\times}\\
& &
\tr\left(\rho^{el}_{ac}(\gl-n\go)\left(q^{el}_{ac}(\gl-\go n) -s_c(\gl -
  n\go)^*q^{el}_{ac}(\gl-\go n)s_c(\gl - n\go)\right)\right)
\nonumber
\eea
Obviously, we have
\bed
\begin{split}
&\left|\tr\left(\rho^{el}_{ac}(\gl-n\go)\left(q^{el}_{ac}(\gl-\go n) -s_c(\gl - n\go)^*
q^{el}_{ac}(\gl-\go n)s_c(\gl - n\go)\right)\right)\right| \le\\
& \quad \quad 4\|\rho^{el}_{ac}(\gl-n\go)\|_{\sL(\sh_n(\gl))}\|q^{el}_{ac}(\gl-n\go)\|_{\sL(\sh_n(\gl))},
\quad \gl \in \dR.
\end{split}
\eed
We insert \eqref{4.11} into the Landauer-B\"uttiker formula \eqref{cc3.3}.
Using \eqref{4.5} and \eqref{4.6} as well as
\bed
\int_\dR
\|\rho^{el}_{ac}(\gl)\|_{\sL(\sh_n(\gl))}\|q^{el}_{ac}(\gl)\|_{\sL(\sh_n(\gl))}d\gl
< \infty
\eed
we see that we can interchange the integral and the sum. Doing
so we get
\bead
\lefteqn{
J^c_{\rho,Q} =
\sum_{n\in\dN_0}q^{ph}(n)\rho^{ph}(n)\frac{1}{2\pi}\int_\dR \tr\left(\rho^{el}_{ac}(\gl-n\go)\times\right.}\\
& &
\left.\left(q^{el}_{ac}(\gl-\go n) -s_c(\gl - n\go)^*q^{el}_{ac}(\gl-\go n)s_c(\gl - n\go)\right)\right)d\gl.
\eead
Using \eqref{eq:3.18} we prove \eqref{eq:3.17}.

If $\tr(\rho^{ph}) = 1$, then
$\sum_{\dN_0}\rho^{ph}(n) =1$. Further, if $\rho^{ph} = I_{\sh^{ph}}$, then
$q^{ph}(n) = 1$. Hence $\gga = 1$.
\end{proof}

\subsection{Photon induced current}

To calculate the current $J^{ph}_{\rho,Q}$ we used the representation
\eqref{d3.11}. We set
\bed
\wh{S^{ph}_{mn}}(\gl) :=
P_m(\gl)\wh{S_{ph}}(\gl)\upharpoonright\sh_n(\gl), \quad \gl \in
\dR.
\eed
where $\{\wh{S_{ph}}(\gl)\}_{\gl\in\dR}$ is defined by \eqref{r3.1}
and $P_m(\gl)$ is the orthogonal projection from $\sh(\gl)$,
cf. \eqref{r2.6}, onto $\sh_m(\gl) := \sh^{el}(\gl - m\go)$, $\gl \in \dR$.
\begin{proposition}\la{IV.3}
Let $\pS = \{H,H_0\}$ be the $JCL$-model.
Assume that $\rho$ and $Q$ are given by \eqref{4.1} where $\rho^{el}$ and
$q^{el}$ are density  and charge operators for $h^{el}_0$ and
$\rho^{ph}$ and $q^{ph}$ for $h^{ph}$, respectively. If for some
$p \in \dN_0$ the conditions \eqref{4.5} and \eqref{4.6} are
satisfied, then the current $J^{ph}_{\rho,Q}$ is well-defined and
admits the representation
\begin{align}\la{r3.7}
J^{ph}_{\rho,Q} =&\sum_{m\in\dN_0}\rho^{ph}(m)\sum_{n\in\dN_0}q^{ph}(n)
\frac{1}{2\pi}\int_\dR d\gl\;\tr\left(\wh{\rho^{el}_{ac}(\gl-m\go)}\times\right.\\
&\left.\left(q^{el}_{ac}(\gl - n\go)\gd_{mn} - \wh{S^{ph}_{nm}}(\gl)^*q^{el}_{ac}(\gl-n\go)\wh{S^{ph}_{nm}}
(\gl)\right)\right).
\nonumber
\end{align}
where $\{\wh{\rho^{el}_{ac}(\gl)}\}_{\gl\in\dR}$ is the photon
modified electron density defined, cf. \eqref{d3.12}, which takes the
form
\be\la{4.13}
\wh{\rho^{el}_{ac}(\gl)} = s_c(\gl)\rho^{el}(\gl)s_c(\gl)^*, \quad \gl \in \dR.
\ee
\end{proposition}
\begin{proof}
By Lemma \ref{IV.1} we get that that the charge $Q$ is $p$-tempered
and $(H_0 - i)^p\rho$ is a bounded operator. By Corollary \ref{III.2}
the current $J^{ph}_{\rho,Q} := J^{\pS_{ph}}_{\rho^c,Q^c}$ is
well-defined.

Since $\left(Q_{ac}(\gl) - \wh{S_{ph}(\gl)^*}Q_{ac}(\gl)\wh{S_{ph}}(\gl)\right)$ is a trace
class operator for $\gl \in \dR$ we get from
\eqref{d3.11} and \eqref{r3.5} that
\bed
\begin{split}
\tr&\left(\wh{\rho_{ac}(\gl)}\left(Q_{ac}(\gl) - \wh{S_{ph}(\gl)^*}Q_{ac}(\gl)\wh{S_{ph}}(\gl)\right)\right) =
\sum_{m\in\dN_0} \rho^{ph}(m) \times\\
&\tr\left(\wh{\rho^{el}(\gl - m\go)}
P_m(\gl)\left(Q_{ac}(\gl) - \wh{S_{ph}(\gl)^*}Q_{ac}(\gl)\wh{S_{ph}}(\gl)\right)P_m(\gl)\right)
\end{split}
\eed
Further we have
\bed
\begin{split}
P_m(\gl)&\left(Q_{ac}(\gl) - \wh{S_{ph}(\gl)^*}Q_{ac}(\gl)\wh{S_{ph}}(\gl)\right)P_m(\gl)\\
        & = q^{ph}(m)\left(q^{el}(\gl - m\go) -P_m(\gl)\wh{S_{ph}(\gl)^*}Q_{ac}(\gl)\wh{S_{ph}}(\gl)\right)P_m(\gl)\\
        & = q^{ph}(m)q^{el}(\gl - m\go) - \sum_{n\in\dN_0}q^{ph}(n)\wh{S^{ph}_{nm}(\gl)^*}q^{el}(\gl -n\go)
        \wh{S^{ph}_{nm}(\gl)}
\end{split}
\eed
for $\gl \in \dR$ where $\wh{S^{ph}_{nm}(\gl)^*} := P_n(\gl)\wh{S_{ph}}(\gl)P_m(\gl)$, $\gl \in \dR$. Notice that $\sum_{n\in\dN_0}$ is a sum with a
finite number of summands.  Hence
\bed
\begin{split}
\tr&\left(\wh{\rho_{ac}(\gl)}\left(Q_{ac}(\gl) - \wh{S_{ph}(\gl)^*}Q_{ac}(\gl)\wh{S_{ph}}(\gl)\right)\right) = \sum_{m\in\dN_0} \rho^{ph}(m)\sum_{n\in\dN_0}q^{ph}(n)\times\\
& \tr\left(\wh{\rho^{el}(\gl-m\go)}\left(q^{el}(\gl - m\go)\gd_{mn} - \wh{S^{ph}_{nm}(\gl)^*}q^{el}(\gl -n\go)
\wh{S^{ph}_{nm}(\gl)}\right)\right)\nonumber
\end{split}
\eed
We are going to show that
\bed
\begin{split}
\sum_{m\in\dN_0} \rho^{ph}(m) &\sum_{n\in\dN_0} |q^{ph}(n)|
\int_\dR\left|\tr\left(\wh{\rho^{el}(\gl-m\go)}\times\right.\right.\\
& \left.\left.\left(q^{el}(\gl - m\go)\gd_{mn} - \wh{S^{ph}_{nm}(\gl)^*}q^{el}(\gl-n\go)\wh{S^{ph}_{nm}
(\gl)}\right)\right)\right|d\gl < \infty.
\end{split}
\eed
Obviously one has the estimate
\bed
\begin{split}
&\left|\tr\right.\left.\left(\wh{\rho^{el}(\gl-m\go)}\left(q^{el}(\gl - m\go)\gd_{mn} - \wh{S^{ph}_{nm}(\gl)^*}q^{el}(\gl-n\go)\wh{S^{ph}_{nm}(\gl)}\right)\right)\right| \le\\
&2\|\wh{\rho^{el}(\gl-m\go)}\|_{\sL(\sh_m(\gl))}\left(\|q^{el}(\gl - m\go)\|_{\sL(\sh_m(\gl))}\gd_{nm}+\|q^{el}(\gl -n\go)\|_{\sL(\sh_n(\gl))}\right).
\end{split}
\eed
Further, we get
\bead
\lefteqn{
\int_{\gl\in\dR}\|\wh{\rho^{el}(\gl-m\go)}\|_{\sL(\sh_m(\gl))}\|q^{el}(\gl - m\go)\|_{\sL(\sh_m(\gl))}\gd_{nm}
\le }\\
& &
\int_{\gl\in\dR}\|\wh{\rho^{el}(\gl)}\|_{\sL(\sh_m(\gl))}\|q^{el}(\gl)\|_{\sL(\sh_m(\gl))}d\gl
\eead
and
\bead
\lefteqn{
\int_\dR\|\wh{\rho^{el}(\gl-m\go)}\|_{\sL(\sh_m(\gl))}\|q^{el}(\gl  -n\go)\|_{\sL(\sh_n(\gl))}d\gl \le}\\
& &
\|q^{el}_{ac}\|_{\sL(\sh^{el})}\int_{\gl\in\dR}\|\wh{\rho^{el}(\gl-(m-n)\go)}\|_{\sL(\sh_{m-n}(\gl))}d\gl
\eead
If the conditions \eqref{4.5} and \eqref{4.6} are satisfied, then
\bed
\sum_{m\in\dN_0} \rho^{ph}(m) |q^{ph}(m)|\int_\dR
\|\wh{\rho^{el}(\gl)}\|_{\sL(\sh_m(\gl))}\|q^{el}(\gl)\|_{\sL(\sh_m(\gl))}d\gl
< \infty
\eed
Further, we have
\bed
\begin{split}
&\sum_{m\in\dN_0}\rho^{ph}(m)\sum_{n\in\dN_0}|q^{ph}(n)|\int_{\gl\in\dR}\|\wh{\rho^{el}(\gl-(m-n)\go)}\|_{\sL(\sh_{m-n}(\gl))}d\gl \le \\
&(v_{\rm max} - v_{\rm min} + 4)\|\rho^{el}_{ac}\|_{\sL(\sh^{el})} \sum_{m\in\dN_0}\rho^{ph}(m)\sum_{|m - n|\le d_{\rm max}}| q^{ph}(n)| < \infty
\end{split}
\eed
where $d_{\rm max}$ is introduced by Lemma \ref{rII.4}. To prove
\bed
\sum_{m\in\dN_0} \rho^{ph}(m)\sum_{|m - n|\le d_{\rm max}}| q^{ph}(n)|
< \infty
\eed
we use again \eqref{4.5} and \eqref{4.6}. The last step admits to interchange
the integral and the sums which immediately proves \eqref{r3.7}
\end{proof}
\bc\la{IV.4}
Let $\pS = \{H,H_0\}$ be the  $JCL$-model.
Assume that $\rho$ and $Q$ are given by \eqref{4.1} where $\rho^{el}$ and
$q^{el}$ are density  and charge operators for $h^{el}_0$ and
$\rho^{ph}$ and $q^{ph}$ for $h^{ph}$, respectively.
If $\rho^{el}$ is an equilibrium state, i.e. $\rho^{el} =
f^{el}(h^{el}_0)$, then
\begin{align}
J^{ph}_{\rho,Q} = &\sum_{m,n\in\dN_0}q^{ph}(n)
\frac{1}{2\pi}\int_\dR \left(\rho^{ph}(n)f^{el}(\gl-n\go)-\rho^{ph}(m)f^{el}(\gl-m\go)\right)\times
\nonumber\\
&\tr\left(\wh{S^{ph}_{nm}}(\gl)^*q^{el}_{ac}(\gl-n\go)\wh{S^{ph}_{nm}}(\gl)\right)d\gl.
\la{r3.9}
\end{align}
\ec
\begin{proof}
From \eqref{r3.7} we get
\begin{align}\nonumber
J^{ph}_{\rho,Q} =&\sum_{n\in\dN_0}q^{ph}(n)\sum_{m\in\dN_0}\rho^{ph}(m)
\frac{1}{2\pi}\int_\dR d\gl\;f^{el}(\gl - m\go)\times\\
&\tr\left(q^{el}_{ac}(\gl - n\go)\gd_{mn} - \wh{S^{ph}_{nm}}(\gl)^*q^{el}_{ac}(\gl-n\go)\wh{S^{ph}_{nm}}(\gl)\right).
\nonumber
\end{align}
Hence
\bed
\begin{split}
J^{ph}_{\rho,Q}& = \sum_{n\in\dN_0}q^{ph}(n)\frac{1}{2\pi}\int_\dR d\gl\;\sum_{m\in\dN_0}\rho^{ph}(m)f^{el}
(\gl - m\go)\times\\
&\tr\left(q^{el}_{ac}(\gl - n\go)\gd_{mn} -
  \wh{S^{ph}_{nm}}(\gl)^*q^{el}_{ac}(\gl-n\go)\wh{S^{ph}_{nm}}(\gl)\right).
\end{split}
\eed
This gives
\begin{align}\la{4.15}
J^{ph}_{\rho,Q}= & \sum_{n\in\dN_0}q^{ph}(n)\frac{1}{2\pi}\int_\dR d\gl\;\bigl(\rho^{ph}(n)f^{el}(\gl - n\go)
\tr\left(q^{el}_{ac}(\gl - n\go)\right) -\\
& \sum_{m\in\dN_0}\rho^{ph}(m)f^{el}(\gl - m\go)\tr\left(\wh{S^{ph}_{nm}}(\gl)^*q^{el}_{ac}(\gl-n\go)
\wh{S^{ph}_{nm}}(\gl)\right)\bigr).
\nonumber
\end{align}
Since
\bed
\begin{split}
\sum_{m\in\dN_0}\rho^{ph}(m)&f^{el}(\gl
-m\go)\tr\left(\wh{S^{ph}_{nm}}(\gl)^*q^{el}_{ac}(\gl-n\go)\wh{S^{ph}_{nm}}(\gl)\right)
=\\
\sum_{m\in\dN_0}&\left(\rho^{ph}(m)f^{el}(\gl -m\go) -
\rho^{ph}(n)f^{el}(\gl -n\go)\right)\times\\
&
\tr\left(\wh{S^{ph}_{nm}}(\gl)^*q^{el}_{ac}(\gl-n\go)\wh{S^{ph}_{nm}}(\gl)\right)
+\\
\rho^{ph}(n)&f^{el}(\gl
-n\go)\sum_{m\in\dN_0}\tr\left(\wh{S^{ph}_{nm}}(\gl)^*q^{el}_{ac}(\gl-n\go)\wh{S^{ph}_{nm}}(\gl)\right)
\end{split}
\eed
Inserting this into \eqref{4.15} we obtain \eqref{r3.9}.
\end{proof}

\section{Electron and photon currents}\label{SecV}

\subsection{Electron current}\label{SecV.1}

To calculate the electron current induced by contacts and photons
contact we make the following choice throughout this section. We set
\be\la{5.0}
Q^{el}_\ga := q^{el}_\ga \otimes q^{ph}, \quad q^{el}_\ga := -\se
p^{el}_\ga \quad \mbox{and} \quad q^{ph} := I_{\sh^{ph}},
\quad\ga \in
\{l,r\},
\ee
where $p^{el}_\ga$ denotes the orthogonal projection from $\sh^{el}$ onto $\sh_\ga^{el}$.
By $\se > 0$ we denote the magnitude of the elementary charge. Since $p^{el}_\ga$ commutes with
$h^{el}_\ga$ one easily verifies that $Q^{el}_\ga$ commutes with $H_0$ which
shows that $Q^{el}_\ga$ is a charge. Following \cite{Pillet2007} the flux
related to $Q^{el}_\ga$ gives us the electron current $J^{\pS}_{\rho,Q^{el}_\ga}$ entering the lead
$\ga$ from the sample. Notice $Q^{el}_\ga = -\se P_\ga$ where $P_\ga$
is the orthogonal projection from $\sH$ onto $\sH_\ga :=
\sh^{el}_\ga  \otimes \sh^{ph}$. Since $q^{ph} = I_{\sh^{ph}}$ the condition
\eqref{4.6} is immediately satisfied for any $p \ge 0$.

Let $f(\cdot): \dR \longrightarrow \dR$ be a non-negative bounded measurable
function. We set
\be\la{5.1}
\rho^{el} = \rho^{el}_l \oplus \rho^{el}_S \oplus \rho^{el}_r, \quad \rho^{el}_\ga := f(h^{el}_\ga - \mu_\ga),
\quad \ga \in \{l,r\}.
\ee
and $\rho = \rho^{el} \otimes \rho^{ph}$. By $\mu_\ga$ the chemical potential of the lead $\ga$ is denoted. In
applications one sets $f(\gl) := f_{FD}(\gl)$, $\gl \in \dR$, where $f_{FD}(\gl)$ is the so-called
Fermi-Dirac distribution given by \eqref{5.2}.
If $\gb = \infty$, then $f_{FD}(\gl) := \chi_{\dR_-}(\gl)$, $\gl \in \dR$.
Notice that $[\rho^{el},p^{el}] = 0$. For $\rho^{ph}$ we choose the Gibbs state
\be\la{5.3}
\rho^{ph} := \frac{1}{Z}e^{-\gb h^{ph}}, \quad Z = \tr(e^{-\gb h^{ph}}) = \frac{1}{1 - e^{-\gb\go}},
\ee
Hence $\rho^{ph} = (1 - e^{-\gb\go})e^{-\gb h^{ph}}$.
If $\gb = \infty$, then $\rho^{ph} := (\cdot,\gY_0)\gY_0$.
Obviously, $\tr(\rho^{ph}) = 1$. We note that $\rho^{ph}(n) =
(1-e^{-\gb\go})e^{-n\gb\go}$, $n \in \dN_0$, satisfies the condition
\eqref{4.5} for any $p \ge 0$. Obviously, $\rho_0 = \rho^{el} \otimes \rho^{ph}$ is a density
operator for $H_0$.
\bd
{\rm
Let $\pS = \{H,H_0\}$ be the $JCL$-model. If $Q := Q^{el}_\ga$, where $Q^{el}_\ga$ is given by
\eqref{5.0},  and $\rho := \rho_0 := \rho^{el}\otimes \rho^{ph}$,
where $ \rho^{el}$  and $\rho^{ph}$ are given by \eqref{5.1} and
\eqref{5.3}, then $J^{el}_{\rho_0,Q^{el}_\ga} :=
J^\pS_{\rho_0,Q^{el}_\ga}$ is called the electron current entering the
lead $\ga$. The currents
$J^c_{\rho_0,Q^{el}_\ga}$ and $J^{ph}_{\rho_0,Q^{el}_\ga}$ are called the
contact induced  and photon induced electron currents.
}
\ed

\subsubsection{Contact induced electron current}

The following proposition immediately follows from Proposition \ref{dIII.2}.
\begin{proposition}
Let $\pS = \{H,H_0\}$ be the $JCL$-model.
Then the contact induced electron current $J^c_{\rho_0,Q^{el}_\ga}$,
$\ga \in \{l,r\}$, is given by
$J^c_{\rho_0,Q^{el}_\ga} = J^{\ps_c}_{\rho^{el},q^{el}_\ga}$.
In particular, one has
\be\la{5.5}
J^c_{\rho_0,Q^{el}_\ga} = -\frac{\se}{2\pi} \int_\dR
(f(\gl-\mu_\ga) - f(\gl - \mu_\gk)\gs_c(\gl)d\gl, \;\; \ga,\gk \in
\{l,r\},\;\; \ga \not= \gk,
\ee
where $\{\gs_c(\gl)\}_{\gl\in\dR}$ is the channel cross-section
from left to the right of the scattering system $\ps_c =
\{h^{el},h^{el}_0\}$, cf. Example \ref{III.3}.
\end{proposition}
\begin{proof}
Since $\tr(\rho^{ph}) =1$ it follows from
Proposition \ref{dIII.2} that $J^c_{\rho_0,Q^{el}_\ga} =
J^{\ps_c}_{\rho^{el},q^{el}_\ga}$. From
\eqref{3.6a}, cf. Example \ref{III.3}, we find \eqref{5.5}.
\end{proof}

If $\mu_l > \mu_r$ and $f(\cdot)$ is decreasing,
then $J^c_{\rho_0,Q^{el}_l}  < 0$. Hence the
electron contact current is going from the
left lead to the right which is in accordance with the physical
intuition. In particular, this is valid for the Fermi-Dirac distribution.
\begin{proposition}\la{prop:III.4}
Let $\pS = \{H,H_0\}$ be the $JCL$-model.
Further, let $\rho^{el}$ and $\rho^{ph}$ be given by \eqref{5.1} and
\eqref{5.3}, respectively. If the charge $Q^{el}_\ga$ is given by \eqref{5.0}, then
the following holds:
\begin{enumerate}

\item[\em{(E)}] If $\mu_l = \mu_r$, then
  $J^c_{\rho_0,Q^{el}_\ga} = 0$, $\ga \in \{l,r\}$.

\item[\em{(S)}] If $v_l \ge v_r + 4$, then
  $J^c_{\rho_0,Q^{el}_\ga} = 0$, $\ga \in \{l,r\}$, even if $\mu_l \not= \mu_l$.

\item[\em{(C)}] If $e^S_0 = \gd^S_0$ and $e^S_1 = \gd^S_1$, then
$J^c_{\rho_0,Q^{el}_\ga}
= 0$, $\ga \in \{l,r\}$, even if $\mu_l \not= \mu_l$.
\end{enumerate}
\end{proposition}
\begin{proof}
(E) If $\mu_l = \mu_r$, then $f(\gl -\mu_l)  = f(\gl-\mu_r)$.
Applying formula \eqref{5.5} we obtain
$J^c_{\rho_0,Q^{el}_\ga} = 0$.

(S) If $v_l \ge v_r + 4$, then $h^{el,ac}_0$ has simple spectrum.
Hence the scattering matrix $\{s_c(\gl)\}_{\gl \in \dR}$ of the
scattering system $\ps_c = \{h^{el},h^{el}_0\}$ is a scalar
function which immediately yields $\gs_c(\gl) = 0$, $\gl \in \dR$,
which yields $J^c_{\rho_0,Q^{el}_\ga} = 0$.

(C) In this case the Hamiltonian $h^{el}$ decomposes into
  a direct sum of two Hamiltonians which do not interact. Hence the
  scattering matrix of $\{s_c(\gl)\}_{\gl \in \dR}$ of the scattering
  system $\ps_c = \{h^{el},h^{el}_0\}$ is diagonal which immediately yields
$J^c_{\rho_0,Q^{el}_\ga} = 0$.
\end{proof}

\subsubsection{Photon induced electron current }

To analyze \eqref{r3.7} is hopeless if we make no assumptions
concerning $\rho^{el}$ and the scattering operator $s_c$. The simplest
assumptions is that $\rho^{el}$ and $s_c$ commute. In this
case we get $\wh{\rho^{el}}(\gl) = \rho^{el}(\gl)$, $\gl \in \dR$.
\bl\la{V.2a}
Let $\pS = \{H,H_0\}$ be the $JCL$-model.
Further let $\rho^{el}$ be given by \eqref{5.1}.
If one of the cases $(E)$, $(S)$ or $(C)$ of Proposition
\ref{prop:III.4} is realized, then the $\rho^{el}$ and $s_c$ commute.
\el
\begin{proof}
If $(E)$ holds, then $\rho^{el}= f(h^{el}_0)$ which yields
$[\rho^{el},s_c] = 0$. If $(S)$ is valid, then the scattering matrix
$\{s_c(\gl)\}_{\gl \in \dR}$ is a scalar function which shows $[\rho^{el},s_c] = 0$.
Finally, if $(C)$ is realized, then the scattering matrix
$\{s_c(\gl)\}_{\gl \in \dR}$ diagonal. Since the $\rho^{el}$ is given
by \eqref{5.1} we get $[\rho^{el},s_c] = 0$.
\end{proof}

We are going to calculate the current $J^{ph}_{\rho_0,Q^{el}_\ga}$,
see \eqref{r3.7}. Obviously, we have $P_\ga(\gl) = \sum_{n\in\dN_0}p^{el}_\ga(\gl-n\go)$
and $I_{\sh(\gl)} = P_l(\gl) + P_r(\gl)$, $\gl \in \dR$. We set
\bed
P_{n_\ga}(\gl) := P_\ga(\gl)P_n(\gl) = P_n(\gl)P_\ga(\gl) =
p^{el}_\ga(\gl-n\go), \; \ga\in \{l,r\},
\eed
$n \in \dN_0$, $\gl \in \dR$. In the following we use the notation
$\wh{T_{ph}}(\gl) = \wh{S_{ph}}(\gl) - I_{\sh(\gl)}$,
$\gl \in \dR$, where $\{\wh{T_{ph}(\gl)}\}_{\gl\in\dR}$ is called the
transition matrix and $\{\wh{S_{ph}(\gl)}\}_{\gl \in \dR}$ is given by \eqref{r3.1}.
We set
\bed
\wh{T^{ph}_{k_\ga m_\gk}(\gl)} :=
P_{k_\ga}(\gl)\wh{T_{ph}}(\gl)P_{m_\gk}(\gl), \quad \gl \in \dR, \quad \ga,\gk
\in \{l,r\},\quad k,m \in \dN_0.
\eed
and
\be\la{eq:3.31}
\wh{\gs^{ph}_{k_\ga m_\gk}} (\gl) = \tr(\wh{T^{ph}_{k_\ga m_\gk}(\gl)^*}\wh{T^{ph}_{k_\ga m_\gk}(\gl)}),
\quad \gl \in \dR,
\ee
which is the cross-section between the channels $k_\ga$ and $m_\gk$.
\begin{proposition}\la{V.4}
Let $\pS = \{H,H_0\}$ be the $JCL$-model.

\item[\;\;\rm (i)] If $\rho^{el}$ commutes with the scattering operator
$s_c$ and $q^{el}$, then
\begin{align}\la{r3.18b}
&J^{ph}_{\rho_0,Q^{el}_\ga}= -\sum_{\substack{m,n\in\dN_0\\ \gk\in\{l,r\}}}
\frac{\se}{2\pi}\int_\dR \times\\
& \left(\rho^{ph}(n)f(\gl-\mu_\ga-n\go)- \rho^{ph}(m)f(\gl-\mu_\gk-m\go)\right)\wh{\gs^{ph}_{n_\ga m_\gk}(\gl)}d\gl.
\nonumber
\end{align}

\item[\;\;\rm (ii)] If in addition $\pS = \{H,H_0\}$ is time reversible
  symmetric, then
\begin{align}\la{5.8}
&J^{ph}_{\rho_0,Q^{el}_\ga}= -\sum_{m,n\in\dN_0}
\frac{\se}{2\pi}\int_\dR \times\\
& \left(\rho^{ph}(n)f(\gl-\mu_\ga-n\go)- \rho^{ph}(m)f(\gl-\mu_{\ga'}-m\go)\right)\wh{\gs^{ph}_{n_{\ga} m_{\ga'}}(\gl)}d\gl,
\nonumber
\end{align}
$\ga,\ga' \in \{l,r\}$, $\ga \not= \ga'$.
\end{proposition}
\begin{proof}
(i) Let us assume that
\bed
q^{el} = \sum_{\gk\in\{l,r\}}g_\gk(h^{el}_\gk),
\eed
Notice that
\be\la{5.12}
q^{el}_{ac}(\gl) = \sum_{\gk\in\{l,r\}}g_{\gk}(\gl)p^{el}_\gk(\gl),
\quad \gl\in\dR.
\ee
Inserting \eqref{5.12} into \eqref{r3.7} and
using $q^{ph} = I_{\sh^{ph}}$ we get
\bed
\begin{split}
J^{ph}_{\rho_0,Q} =&\sum_{\substack{m\in\dN_0\\\ga \in \{l,r\}}}\rho^{ph}(m)\sum_{\substack{n\in\dN_0\\\gk\in\{l,r\}}}
\frac{1}{2\pi}\int_\dR d\gl\;\phi_\ga(\gl-m\go)g_\gk(\gl-n\go)\times\\
&\tr\left(p^{el}_\ga(\gl-m\go)\left(p^{el}_\gk(\gl-n\go)\gd_{mn} - \wh{S^{ph}_{nm}}(\gl)^*p^{el}_\gk(\gl-n\go)
\wh{S^{ph}_{nm}}(\gl)\right)\right)
\end{split}
\eed
where for simplicity we have set
\be\la{5.12a}
\phi_\ga(\gl) := f(\gl-\mu_\ga), \quad \gl \in
\dR, \quad n \in \dN_0, \quad \ga \in \{l,r\}.
\ee
Obviously, we have
\begin{align}
J^{ph}_{\rho_0,Q} =&\sum_{\substack{n\in\dN_0\\\gk\in\{l,r\}}}\rho^{ph}(n)
\frac{1}{2\pi}\int_\dR d\gl\;\phi_\gk(\gl-n\go)g_\gk(\gl-n\go)\tr\left(p^{el}_\gk(\gl-n\go)\right)-\nonumber\\
&\sum_{\substack{n\in\dN_0\\\gk\in\{l,r\}}}\sum_{\substack{m\in\dN_0\\\ga \in \{l,r\}}}\rho^{ph}(m)
\frac{1}{2\pi}\int_\dR d\gl\;\phi_\ga(\gl-m\go)g_\gk(\gl-n\go)\times\la{eq:5.10}\\
&\tr\left(p^{el}_\ga(\gl-m\go)\wh{S^{ph}_{nm}}(\gl)^*p^{el}_\gk(\gl-n\go)\wh{S^{ph}_{nm}}(\gl)
p^{el}_\ga(\gl-m\go)\right).\nonumber
\end{align}
Since the scattering matrix $\{\wh{S^{ph}}(\gl)\}_{\gl \in \dR}$ is unitary we have
\be\la{eq:5.11}
p^{el}_\gk(\gl - n\go) = \sum_{\substack{m\in\dN_0 \\ \ga \in \{l,r\}}} p^{el}_\gk(\gl - n\go)\wh{S^{ph}_{mn}}(\gl)^* p^{el}_\ga(\gl - m\go) \wh{S^{ph}_{mn}}(\gl)p^{el}_\gk(\gl - n\go)
\ee
for $n\in \dN_0$ and $\gk \in \{l,r\}$.
Inserting \eqref{eq:5.11} into \eqref{eq:5.10} we find
\bed
\begin{split}
J^{ph}_{\rho_0,Q}
=&\sum_{\substack{n\in\dN_0\\\gk\in\{l,r\}}}\sum_{\substack{m\in\dN_0\\\ga\in\{l,r\}}}\rho^{ph}(n)
\frac{1}{2\pi}\int_\dR d\gl\;\phi_\gk(\gl-n\go)g_\gk(\gl-n\go)\times\\
&\tr\left(p^{el}_\gk(\gl-n\go)\wh{S^{ph}_{nm}}(\gl)^*p^{el}_\ga(\gl-m\go)\wh{S^{ph}_{mn}}(\gl)
p^{el}_\gk(\gl-n\go)\right)-\\
&\sum_{\substack{n\in\dN_0\\\gk\in\{l,r\}}}\sum_{\substack{m\in\dN_0\\\ga \in \{l,r\}}}\rho^{ph}(m)
\frac{1}{2\pi}\int_\dR d\gl\;\phi_\ga(\gl-m\go)g_\gk(\gl-n\go)\times\\
&\tr\left(p^{el}_\ga(\gl-m\go)\wh{S^{ph}_{nm}}(\gl)^*p^{el}_\gk(\gl-n\go)\wh{S^{ph}_{nm}}
(\gl)p^{el}_\ga(\gl-m\go)\right).
\end{split}
\eed
Using the notation \eqref{eq:3.31} we find
\bed
\begin{split}
J^{ph}_{\rho_0,Q}
=&\sum_{\substack{n\in\dN_0\\\gk\in\{l,r\}}}\sum_{\substack{m\in\dN_0\\\ga\in\{l,r\}}}\rho^{ph}(n)
\frac{1}{2\pi}\int_\dR d\gl\;\phi_\gk(\gl-n\go)g_\gk(\gl-n\go)\wh{\gs^{ph}_{m_\ga n_\gk}(\gl)}-\\
&\sum_{\substack{n\in\dN_0\\\gk\in\{l,r\}}}\sum_{\substack{m\in\dN_0\\\ga \in \{l,r\}}}\rho^{ph}(m)
\frac{1}{2\pi}\int_\dR d\gl\;\phi_\ga(\gl-m\go)g_\gk(\gl-n\go)\wh{\gs^{ph}_{n_\gk m_\ga}(\gl)}:
\end{split}
\eed
By \eqref{3.9} we find
\bed
\sum_{\substack{m\in\dN_0\\\ga\in\{l,r\}}}\wh{\gs^{ph}_{m_\ga n_\gk}(\gl)} =
\sum_{\substack{m\in\dN_0\\\ga\in\{l,r\}}}\wh{\gs^{ph}_{n_\gk m_\ga}(\gl)}
\quad \gl \in \dR.
\eed
Using that we get
\begin{align}\la{eq:3.37}
&J^{ph}_{\rho_0,Q} =
\sum_{\substack{m,n\in\dN_0\\\ga,\gk\in\{l,r\}}}\
\frac{1}{2\pi}\int_\dR \times\\
&
\left(\rho^{ph}(n)\phi_\gk(\gl-n\go) - \rho^{ph}(m)\phi_\ga(\gl-m\go)\right)g_\gk(\gl-n\go)
\wh{\gs^{ph}_{n_\gk m_\ga}(\gl)}d\gl.
\nonumber
\end{align}
Setting $g_\ga(\gl) = -\se$ and $g_\gk(\gl) \equiv 0$, $\gk \not=
\ga$, we obtain \eqref{r3.18b}.

(ii) A straightforward computation shows that
\bead
\lefteqn{
\sum_{n,m\in\dN_0}\int_\dR \left(\rho^{ph}(n)f(\gl -\mu_a - n\go) -\rho^{ph}(m)f(\gl -\mu_a - m\go)\right)
\wh{\gs^{ph}_{n_\ga m_\ga}(\gl)}d\gl=}\\
& &
\sum_{n,m\in\dN_0}\int_\dR \left(\rho^{ph}(m)f(\gl -\mu_a - m\go) -\rho^{ph}(n)f(\gl -\mu_a - n\go)\right)
\wh{\gs^{ph}_{m_\ga n_\ga}(\gl)}d\gl
\eead
Since $\gs^{ph}_{m_\ga n_\ga}(\gl) = \gs^{ph}_{n_\ga m_\ga}(\gl)$,
$\gl \in\dR$, we get
\bed
\begin{split}
&
\sum_{n,m\in\dN_0}\int_\dR \left(\rho^{ph}(n)f(\gl -\mu_a - n\go) -\rho^{ph}(m)f(\gl -\mu_a - m\go)\right)
\wh{\gs^{ph}_{n_\ga m_\ga}(\gl)}d\gl=\\
&
-\sum_{n,m\in\dN_0}\int_\dR \left(\rho^{ph}(n)f(\gl -\mu_a - n\go) -\rho^{ph}(m)f(\gl -\mu_a - m\go)\right)
\wh{\gs^{ph}_{n_\ga m_\ga}(\gl)}d\gl
\end{split}
\eed
which yields
\bed
\sum_{n,m\in\dN_0}\int_\dR \left(\rho^{ph}(n)f(\gl -\mu_a - n\go)
  -\rho^{ph}(m)f(\gl -\mu_a - m\go)\right)\wh{\gs^{ph}_{n_\ga
    m_\ga}(\gl)}d\gl = 0.
\eed
Using that we get immediately the representation \eqref{5.8} from
\eqref{r3.18b}.
\end{proof}
\bc\la{V.6}
Let $\pS = \{H,H_0\}$ be the$JCL$-model.

\item[\;\;\rm (i)]
If the cases cases $(E)$, $(S)$ or $(C)$ of Proposition
\ref{prop:III.4} are realized, then the representation
 \eqref{r3.18b} holds.

\item[\;\;\rm (ii)] If the case (E) of Proposition
\ref{prop:III.4} is realized and the system $\pS = \{H,H_0\}$ is time
reversible symmetric, then
\begin{align}\la{5.16}
&J^{ph}_{\rho_0,Q^{el}_\ga} = \\
& -\sum_{m,n\in\dN_0}\frac{\se}{2\pi}\int_\dR (\rho^{ph}(n)f(\gl - \mu- n\go) - \rho^{ph}(m)f(\gl - \mu - m\go))
\wh{\gs^{ph}_{n_\ga m_{\ga'}}}(\gl)d\gl
\nonumber
\end{align}
$n \in \dN_0$, $\ga \in \{l,r\}$ where $\mu := \mu_l = \mu_r$ and $\ga
\not= \ga'$.

\item[\;\;\rm (iii)] If the case (E) of Proposition
\ref{prop:III.4} is realized and the system $\pS = \{H,H_0\}$ is time
reversible and mirror symmetric, then $J^{ph}_{\rho_0,Q^{el}_\ga} =
0$.
\ec
\begin{proof}
(i) The statement follows from Proposition \ref{V.4}(i) and Lemma \ref{V.2a}.

(ii) Setting $\mu_\ga = \mu_{\ga'}$ formula \eqref{5.16} follows
\eqref{5.8}.

(iii) If $\pS = \{H,H_0\}$ is time reversible and mirror symmetric we
get from Lemma \ref{II.13}\,(ii) that $\wh{\gs^{ph}_{n_\ga m_{\ga'}}}(\gl) =
\wh{\gs^{ph}_{n_{\ga'} m_\ga}}(\gl)$, $\gl \in \dR$, $n,m\in\dN_0$,
$\ga,\ga' \in \{l,r\}$, $\ga \not= \ga'$. Using that we get from \eqref{5.16} that
\bed
\begin{split}
&J^{ph}_{\rho_0,Q^{el}_\ga} = \\
& -\sum_{m,n\in\dN_0}\frac{\se}{2\pi}\int_\dR (\rho^{ph}(n)f(\gl - \mu- n\go) - \rho^{ph}(m)f(\gl - \mu - m\go))
\wh{\gs^{ph}_{n_{\ga'} m_\ga}}(\gl)d\gl.
\end{split}
\eed
Interchanging $m$ and $n$ we get
\bed
\begin{split}
&J^{ph}_{\rho_0,Q^{el}_\ga} = \\
& -\sum_{m,n\in\dN_0}\frac{\se}{2\pi}\int_\dR (\rho^{ph}(m)f(\gl - \mu- m\go) - \rho^{ph}(n)f(\gl - \mu - n\go))
\wh{\gs^{ph}_{m_{\ga'} n_\ga}}(\gl)d\gl.
\end{split}
\eed
Using that $\pS$ is time reversible symmetric we get from Lemma
\ref{II.13}\,(i) that
\bed
\begin{split}
&J^{ph}_{\rho_0,Q^{el}_\ga} = \\
& -\sum_{m,n\in\dN_0}\frac{\se}{2\pi}\int_\dR (\rho^{ph}(m)f(\gl -
\mu- m\go) - \rho^{ph}(n)f(\gl - \mu - n\go))\wh{\gs^{ph}_{n_\ga m_{\ga'}}}(\gl)d\gl.
\end{split}
\eed
which shows that $J^{ph}_{\rho_0,Q^{el}_\ga} =
-J^{ph}_{\rho_0,Q^{el}_\ga}$. Hence $J^{ph}_{\rho_0,Q^{el}_\ga} =0$.
\end{proof}
We note that by Proposition \ref{prop:III.4} the contact induced
current is zero, i.e. $J^c_{\rho_0,Q^{el}_\ga} = 0$. Hence, if the
$\pS$ is time reversible and mirror symmetric, then the total current is
zero, i.e. $J^\pS_{\rho_0,Q^{el}_\ga} = 0$.
\begin{remark}\la{rem:III.2}
{\rm
Let the case $(E)$ of Proposition \ref{prop:III.4} be realized, that is, $\mu_l = \mu_r$.
Moreover, we assume for simplicity that $0 =: v_r \le
v := v_l$.
\begin{enumerate}

\item [(i)] If $\gb = \infty$, then  $\rho^{ph}(n) = \gd_{0n}$, $n
  \in \dN_0$. From \eqref{r3.18b} we immediately get that
  $J^{ph}_{\rho^{el},Q^{el}_\ga} = 0$. That means, if the temperature
  is zero, then the photon induced electron current is zero.

\item[(ii)] The photon induced electron current might be zero even if $\gb < \infty$. Indeed,
let $\pS = \{H,H_0\}$ be time reversible symmetric and let the case (E) be realized.
If $\go \ge v + 4$ and , then $\sh^{el}(\gl) := \sh^{el}_n(\gl) = \sh^{el}(\gl-n\go)$,
$n \in \dN_0$. Hence one always has $n=m$ in formula
\eqref{5.16} which immediately yields $J^{ph}_{\rho_0,Q^{el}_\ga} = 0$.

\item[(iii)] The photon induced electron
current might be different from zero.  Indeed, let $\pS = \{H,H_0\}$ be time
reversible symmetric and let $v=2$ and $\go = 4$,
then one gets that to calculate the $J^{ph}_{\rho_0,Q^{el}_l}$ one has to take into
account $m = n+1$ in formula \eqref{5.16}. Therefore we find
\bed
\begin{split}
&J^{ph}_{\rho_0,Q^{el}_l} = -\sum_{n\in\dN_0}\frac{\se}{2\pi}\int_\dR d\gl\;\times\\
&
\left(\rho^{ph}(n)f(\gl-\mu - n\go) -
  \rho^{ph}(n+1)f(\gl - \mu - (n+1)\go)\right)
\wh{\gs^{ph}_{n_l\,(n+1)_r}}(\gl).
\end{split}
\eed
If $\rho^{ph}$ is given by \eqref{5.3} and
$f(\gl) = f_{FD}(\gl)$, cf. \eqref{5.2}, then one easily verifies that
\bed
\frac{\partial}{\partial x} \rho^{ph}(x)f_{FD}(\gl - \mu - x\go) <
0, \quad x,\mu,\gl \in \dR.
\eed
Hence $\rho^{ph}(n)f_{FD}(\gl-\mu -n\go)$ is decreasing in $n \in
\dN_0$  for $\gl,\mu \in \dR$ which yields
$\left(\rho^{ph}(n)f(\gl-\mu - n\go) - \rho^{ph}(n+1)f(\gl - \mu -
  (n+1)\go)\right) \ge 0$. Therefore $J^{ph}_{\rho_0,Q^{el}_l} \le
0$ which means that the photon induced current leaves the
left-hand side and enters the right-hand side. In fact
$J^{ph}_{\rho_0,Q^{el}_l} = 0$ implies that
$\wh{\gs^{ph}_{n_l\,(n+1)_r}}(\gl) = 0$ for $n \in \dN_0$ and $\gl \in
\dR$ which means that there is no scattering from the left-hand side
to the right one and vice versa which can be excluded generically.
\end{enumerate}
}
\end{remark}

\subsection{Photon current}	\label{sec:LandBuettApplicationPhotonCur}

The photon current is related to the charge
\begin{equation*}
Q := Q^{ph} = - I_{\sh^{el}}\otimes \sn,
\end{equation*}
where $\sn = \de\Gamma(1) = b^\ast b$ is the photon number operator
on $\sh^{ph} = \sF_+(\dC)$, which is self-adjoint and commutes with
$h^{ph}$. It follows that $Q^{ph}$ is also self-adjoint and commutes with
$H_0$.  It is not bounded, but since $\dom(\sn) = \dom(h^{ph})$,
it is immediately obvious that $Q^{ph}(H_0 + \theta)^{-1}$ is
bounded, whence $\sN$ is a tempered charge. Its charge matrix with respect to the
spectral representation $\Pi(H^{ac}_0)$ of Lemma \ref{rII.4} is given by
\bed
Q^{ph}_{ac}(\lambda) = - \bigoplus_{n \in \dN_0} n P_n(\gl).
\eed
We recall that $P_n(\gl)$ is the orthogonal projection form
$\sh(\gl)$ onto $\sh_n(\gl) = \sh^{el}(\gl-n\go)$, $\gl\in\dR$.
We are going to calculate the photon current or, how it is also
called, the photon production rate.

\subsubsection{Contact induced photon current}

The following proposition is in accordance with the physical intuition.
\begin{proposition}\la{V.10}
Let $\pS = \{H,H_0\}$ be the $JCL$-model.
Then $J^c_{\rho_0,Q^{ph}} = 0$.
\end{proposition}
\begin{proof}
We note that $q^{el}_{ac}(\gl) = I_{\sh^{el}(\gl)}$, $\gl \in \dR$.
Inserting this into \eqref{eq:3.18} we get
$J^{\ps_c}_{\rho^{el},q^{el}} = 0$. Applying Proposition \ref{dIII.2} we prove $J^c_{\rho_0,Q^{ph}} = 0$.
\end{proof}

The result reflects the fact that the lead contact does not contributed to
the photon current which is plausible.

\subsubsection{Photon current}

From the Proposition \ref{V.10} we get that only the photon induced
photon current $J^{ph}_{\rho_0,Q^{ph}}$ contributes to the photon
current $J^\pS_{\rho_0,Q^{ph}}$. Since $J^\pS_{\rho_0,Q^{ph}} = J^{ph}_{\rho_0,Q^{ph}}$
we call $J^{ph}_{\rho_0,Q^{ph}}$ simply the {\em photon current}.

Using the notation $\wh{T^{ph}_{nm}}(\gl) :=
P_n(\gl)\wh{T_{ph}}(\gl)\upharpoonright\sh^{el}(\gl-m\go)$, $\gl \in \dR$, $m,n\in\dN_0$. We set
\be\la{5.18}
\wt{T^{ph}_{nm}}(\gl) = \wh{T^{ph}_{nm}}(\gl)s_c(\gl - m\go), \quad
\gl \in \dR, \quad m,n \in \dN_0
\ee
and
\be\la{5.19}
\wt{T^{ph}_{n_\gk m_\ga}}(\gl) :=
P_{n_\gk}(\gl)\wt{T^{ph}_{nm}}(\gl)\upharpoonright\sh^{el}_\ga(\gl - m\go), \quad
\gl \in \dR,
\ee
$m,n \in \dN_0$, $\ga,\gk \in \{l,r\}$, as well as
$\wt{\gs^{ph}_{n_\gk\,m_\ga}}(\gl) := \tr(\wt{T^{ph}_{n_\gk m_\ga}}(\gl)^*\wt{T^{ph}_{n_\gk m_\ga}}(\gl))$,
$\gl \in \dR$.
\begin{proposition}\la{V.11}
Let $\pS = \{H,H_0\}$ be the $JCL$-model.

\item[\;\;\rm (i)]
Then
\be\la{5.20}
J^{ph}_{\rho_0,Q^{ph}} = \sum_{\substack{m,n \in \dN_0\\
\ga,\gk \in \{l,r\}}}(n-m)\rho^{ph}(m)
\frac{1}{2\pi}\int_\dR f(\gl - \mu_\ga - m\go) \wt{\gs^{ph}_{n_\gk
    m_\ga}}(\gl)d\gl
\ee

\item[\;\;\rm (ii)] If $\rho^{el}$ commutes with $s_c$, then
\be\la{5.15}
J^{ph}_{\rho_0,Q^{ph}} = \sum_{\substack{m,n \in \dN_0\\
\ga,\gk \in \{l,r\}}}(n-m)\rho^{ph}(m)
\frac{1}{2\pi}\int_\dR f(\gl - \mu_\ga - m\go) \wh{\gs^{ph}_{n_\gk
    m_\ga}}(\gl)d\gl
\ee

\item[\;\;\rm (iii)]
If $\rho^{el}$ commutes with $s_c$ and $\pS = \{H,H_0\}$ is
time reversible symmetric, then
\begin{align}\la{5.16x}
&J^{ph}_{\rho_0,Q^{ph}} = \sum_{\substack{m,n \in \dN_0, n > m\\
\gk,\ga\in \{l,r\}}}
\frac{1}{2\pi}\int_\dR d\gl\times\\
& (n-m)\left(\rho^{ph}(m)f(\gl - \mu_\ga - m\go)
  - \rho^{ph}(n)f(\gl - \mu_\gk - n\go)\right)\wh{\gs^{ph}_{n_\gk m_\ga}}(\gl)
\nonumber
\end{align}
where $\ga' \in\{l,r\}$ and $\ga' \not= \ga$.
\end{proposition}
\begin{proof}
(i) From \eqref{r3.7} we get
\bed
\begin{split}
J^{ph}_{\rho_0,Q^{ph}} =&-\sum_{m,n\in\dN_0}n\rho^{ph}(m)
\frac{1}{2\pi}\int_\dR d\gl\;\tr\left(\wh{\rho^{el}_{ac}(\gl-m\go)}\times\right.\\
&\left.\left(P_n(\gl)\gd_{mn} - \wh{S^{ph}_{nm}}(\gl)^*q^{el}_{ac}(\gl-n\go)\wh{S^{ph}_{nm}}(\gl)\right)\right).
\end{split}
\eed
Hence
\bead
\lefteqn{
J^{ph}_{\rho_0,Q^{ph}} =
-\sum_{m\in \dN_0}m\rho^{ph}(m)\times}\\
& &
\frac{1}{2\pi}\int_\dR \tr\left(\wh{\rho^{el}_{ac}(\gl-m\go)}
\left(P_m(\gl) -\wh{S^{ph}_{mm}(\gl)^*}P_m(\gl)\wh{S^{ph}_{mm}(\gl)}\right)\right)d\gl
+\\
& &
\sum_{\substack{m,n \in \dN_0\\ m \not= n}}
n\rho^{ph}(m)\frac{1}{2\pi}\int_\dR \tr\left(\wh{\rho^{el}_{ac}(\gl-m\go)}\wh{S^{ph}_{nm}(\gl)^*}P_n(\gl)
\wh{S^{ph}_{nm}(\gl)}\right)d\gl.
\eead
Using the relation $P_m(\gl) = I_{\sh(\gl)} - \sum_{n \in \dN_0, m \not= n}P_n(\gl)$, $\gl \in \dR$,  we get
\bed
\begin{split}
&J^{ph}_{\rho_0,Q^{ph}} =\\
&-\sum_{\substack{m,n\in \dN_0\\m \not= n}}
m\rho^{ph}(m)\frac{1}{2\pi}\int_\dR \tr\left(\wh{\rho^{el}_{ac}(\gl-m\go)}
\left(\wh{S^{ph}_{nm}(\gl)^*}P_n(\gl)\wh{S^{ph}_{nm}(\gl)}\right)\right)d\gl+\\
&
\sum_{\substack{m,n \in \dN_0\\m \not= n}}
n\rho^{ph}(m)\frac{1}{2\pi}\int_\dR \tr\left(\wh{\rho^{el}_{ac}(\gl-m\go)}\wh{S^{ph}_{nm}(\gl)^*}P_n(\gl)
\wh{S^{ph}_{nm}(\gl)}\right)d\gl.
\end{split}
\eed
Since $\wh{T_{ph}}(\gl) = \wh{S_{ph}}(\gl) - I_{\sh(\gl)}$, $\gl \in \dR$, we find
\bed
\begin{split}
&J^{ph}_{\rho_0,Q^{ph}} =\\
&
-\sum_{m,n\in \dN_0}
(m - n)\rho^{ph}(m)\frac{1}{2\pi}\int_\dR
\tr\left(\wh{\rho^{el}_{ac}(\gl-m\go)}\wh{T^{ph}_{nm}(\gl)^*}\wh{T^{ph}_{nm}(\gl)}\right)d\gl.
\end{split}
\eed
Using \eqref{4.13} and definition \eqref{5.18} one gets
\bed
\begin{split}
&J^{ph}_{\rho_0,Q^{ph}} =\\
&
-\sum_{m,n\in \dN_0}
(m - n)\rho^{ph}(m)\frac{1}{2\pi}\int_\dR
\tr\left(\rho^{el}_{ac}(\gl-m\go)\wt{T^{ph}_{nm}(\gl)^*}\wt{T^{ph}_{nm}(\gl)}\right)d\gl\,.
\nonumber
\end{split}
\eed
Since $\rho^{el}_{ac} = \rho^{el}_l \oplus \rho^{el}_r$ where
$\rho^{el}_\ga$ is given by \eqref{5.1} we find
\bed
\begin{split}
&J^{ph}_{\rho_0,Q^{ph}} =\\
&
-\sum_{\substack{m,n\in \dN_0\\ \ga,\gk \in \{l,r\}}}
(m - n)\rho^{ph}(m)\frac{1}{2\pi}\int_\dR f(\gl - \mu_\ga - m\go)
\tr\left(\wt{T^{ph}_{n_\gk m_\ga}(\gl)^*}\wt{T^{ph}_{n_\gk m_\ga}(\gl)}\right)d\gl
\end{split}
\eed
where we have used \eqref{5.19}. Using
$\wt{\gs^{ph}_{n_\gk\,m_\ga}}(\gl) = \tr(\wt{T^{ph}_{n_\gk m_\ga}}(\gl)^*\wt{T^{ph}_{n_\gk m_\ga}}(\gl))$
we prove \eqref{5.20}.

(ii) If $\rho^{el}_{ac}$ commutes with $s_c$, then $\wh{\rho^{el}_{ac}}(\gl) = \rho^{el}_{ac}(\gl)$, $\gl \in \dR$
which yields that one can replace  $\wt{\gs^{ph}_{n_\gk m_\ga}}(\gl)$ by $\wh{\gs^{ph}_{n_\gk m_\ga}}(\gl)$, $\gl \in \dR$. Therefore \eqref{5.15} holds.

(iii) Obviously we have
\bea\la{5.17}
\lefteqn{
J^{ph}_{\rho_0,Q^{ph}} =}\\
& & \sum_{\substack{m,n \in \dN_0, n > m\\ \ga,\gk \in \{l,r\}}}(n-m)\rho^{ph}(m)
\frac{1}{2\pi}\int_\dR f(\gl - \mu_\ga - m\go) \wh{\gs^{ph}_{n_\gk
    m_\ga}}(\gl)d\gl+\nonumber \\
& &  \sum_{\substack{m,n \in \dN_0, n <m\\ \ga,\gk \in \{l,r\}}}(n-m)\rho^{ph}(m)
\frac{1}{2\pi}\int_\dR f(\gl - \mu_\ga - m\go) \wh{\gs^{ph}_{n_\gk
    m_\ga}}(\gl)d\gl\,.
\nonumber
\eea
Moreover, a straightforward computation shows that
\bead
\lefteqn{
\sum_{\substack{m,n \in \dN_0, n < m\\
\ga,\gk \in \{l,r\}}}(n-m)\rho^{ph}(m)
\frac{1}{2\pi}\int_\dR f(\gl - \mu_\ga - m\go) \wh{\gs^{ph}_{n_\gk
    m_\ga}}(\gl)d\gl =}\\
& &
\sum_{\substack{m,n \in \dN_0, n > m\\
\ga,\gk \in \{l,r\}}}(m-n)\rho^{ph}(n)
\frac{1}{2\pi}\int_\dR f(\gl - \mu_\gk - n\go) \wh{\gs^{ph}_{m_\ga n_\gk}}(\gl)d\gl.
\eead
Since $\pS = \{H,H_0\}$ is time reversible symmetric we find
\bea\la{5.18x}
\lefteqn{
\sum_{\substack{m,n \in \dN_0, n < m\\
\ga,\gk \in \{l,r\}}}(n-m)\rho^{ph}(m)
\frac{1}{2\pi}\int_\dR f(\gl - \mu_\ga - m\go) \wh{\gs^{ph}_{m_\ga n_\gk}}(\gl)d\gl =}\\
& &
\sum_{\substack{m,n \in \dN_0, n >m\\
\ga,\gk \in \{l,r\}}}(m-n)\rho^{ph}(n)
\frac{1}{2\pi}\int_\dR f(\gl - \mu_\gk - n\go) \wh{\gs^{ph}_{n_\gk m_\ga}}(\gl)d\gl.
\nonumber
\eea
Inserting \eqref{5.18x} into \eqref{5.17} we obtain \eqref{5.16x}.
\end{proof}
\bc\la{V.10a}
Let $\pS = \{H,H_0\}$ be the $JCL$-model
and let $f = f_{FD}$. If case $(E)$ of Proposition \ref{prop:III.4} is realized
and $\pS = \{H,H_0\}$ is time reversible symmetric, then
$J^{ph}_{\rho_0,Q^{ph}} \ge 0$.
\ec
\begin{proof}
We set $\mu := \mu_l = \mu_r$. One has
\bead
\lefteqn{
\rho^{ph}(m)f(\gl - \mu - m\go)
  - \rho^{ph}(n)f(\gl - \mu - n\go) = }\\
& &
e^{-m\gb\go}(1 -e^{-(n-m)\gb\go})f_{FD}(\gl - \mu - m\go)f_{FD}(\gl -
\mu - n\go) \ge 0
\eead
for $n > m$. From \eqref{5.16x} we get $J^{ph}_{\rho_0,Q^{ph}} \ge 0$.
\end{proof}

\begin{remark}
{\rm
Let us comment the results. If $J^{ph}_{\rho_0,Q^{ph}} \ge 0$, then
system $\pS$ is called light emitting. Similarly, if
$J^{ph}_{\rho_0,Q^{ph}} \le 0$, then we call it light absorbing.
Of course if $\pS$ is light emitting and absorbing, then
$J^{ph}_{\rho_0,Q^{ph}} =0$.

\begin{enumerate}

\item[(i)] If $\gb = \infty$, then $\rho^{ph}(m) = \gd_{0m}$, $m \in
  \dN_0$. Inserting this into \eqref{5.20} we get
\bed
J^{ph}_{\rho_0,Q^{ph}} = \sum_{\substack{n \in \dN_0\\
\ga,\gk \in \{l,r\}}}n
\frac{1}{2\pi}\int_\dR f(\gl - \mu_\ga) \wt{\gs^{ph}_{n_\gk
    0_\ga}}(\gl)d\gl \ge 0
\eed
Hence $\pS$ is light emitting.

\item[(ii)] Let us show $\pS$ might be light emitting even if $\gb <
  \infty$. We consider the case $(E)$ of
  Proposition \ref{prop:III.4}. If $\pS$ is time reversible symmetric, then it follows from
Corollary \ref{V.10a} that the system is light emitting.

If the system $\pS$ is time reversible and mirror symmetric, then
$J^{ph}_{\rho_0,Q^{el}_\ga} = 0$, $\ga \in \{l,r\}$, by Corollary \ref{V.6}(iii) . Since
$J^c_{\rho_0,Q^{el}} = 0$ by Proposition \ref{prop:III.4} we get that
$J^\pS_{\rho_0,Q^{el}_\ga} = 0$ but the photon
current is larger than zero. So our $JCL$-model is light
emitting by a zero total electron current
$J^\pS_{\rho_0,Q^{el}_\ga}$.

Let $v_r = 0$, $v_l = 2$ and $\go = 4$. Hence $\pS$ is not mirror
symmetric. Then we get from Remark
\ref{rem:III.2}(iii) that $J^{ph}_{\rho_0,Q^{el}_l} = -
J^{ph}_{\rho_0,Q^{el}_r} \le 0$.  Hence there is an electron current
from the left to the right lead. Notice that by Proposition
\ref{prop:III.4} $J^c_{\rho_0,Q^{el}_l} = 0$. Hence
$J^\pS_{\rho_0,Q^{el}_l} \le 0$.

\item[(iii)] To realize a light absorbing situation we consider the
  case $(S)$ of Proposition \ref{prop:III.4} and assume that $\pS$ is
  time reversible symmetric. Notice that by Lemma
  \ref{V.2a} $s_c$ commutes with $\rho^{el}$. We make the choice
\bed
v_r = 0, \quad v_l \ge 4, \quad \go = v_l, \quad \mu_l = 0,\quad \mu_r =
\go = v_l.
\eed
It turns out that with respect to the representation \eqref{5.16x} one has only to
$m = n-1$, $\gk = r$ and $\ga = l$. Hence
\bed
\begin{split}
&J^{ph}_{\rho_0,Q^{ph}} = \sum_{n \in \dN}
\frac{1}{2\pi}\int_\dR d\gl\times\\
& \left(\rho^{ph}(n-1)f(\gl - (n-1)\go) - \rho^{ph}(n)f(\gl - (n+1)\go)\right)\wh{\gs^{ph}_{n_l (n-1)_r}}(\gl)
\end{split}
\eed
Since $f(\gl) = f_{FD}(\gl)$ we find
\bead
\lefteqn{
\rho^{ph}(n-1)f(\gl - (n-1)\go) - \rho^{ph}(n)f(\gl - (n+1)\go) =}\\
& &
\rho^{ph}(n-1)f(\gl - (n-1)\go)f(\gl - (n+1)\go)\times\\
& &
\left(1 + e^{\gb(\gl - (n+1)\go)} - e^{-\gb\go}(1 + e^{\gb(\gl - \go(n-1))})\right)
\eead
or
\bead
\lefteqn{
\rho^{ph}(n-1)f(\gl - (n-1)\go) - \rho^{ph}(n)f(\gl - (n+1)\go) =}\\
& &
\rho^{ph}(n-1)f(\gl - (n-1)\go)f(\gl - (n+1)\go)(1 - e^{-\gb\go})(1 -
e^{\gb(\gl-\go n)}).
\eead
Since $\gl -n\go \ge 0$ we find $\rho^{ph}(n-1)f(\gl - (n-1)\go) -
\rho^{ph}(n)f(\gl - (n+1)\go) \le 0$ which yields
$J^{ph}_{\rho_0,Q^{ph}} \le 0$.

To calculate $J^{ph}_{\rho_0, Q^{el}_l}$ we use formula
\eqref{5.8}. Setting $\ga = l$ we get $\ga' = r$ which yields
\bed
\begin{split}
&J^{ph}_{\rho_0,Q^{el}_l}= -\sum_{m,n\in\dN_0}
\frac{\se}{2\pi}\int_\dR d\gl\times\\
& \left(\rho^{ph}(n)f(\gl-\mu_r-n\go)- \rho^{ph}(m)f(\gl-\mu_l-m\go)\right)\wh{\gs^{ph}_{n_l m_r}(\gl)},
\end{split}
\eed
One checks that $\wh{\gs^{ph}_{0_l 0_r}(\gl)} = 0$
and $\wh{\gs^{ph}_{n_l m_r}(\gl)} = 0$ for $m \not= n+1$, $n \in \dN$. Hence
\bed
\begin{split}
&J^{ph}_{\rho_0,Q^{el}_l}= -\sum_{n\in\dN}
\frac{\se}{2\pi}\int_\dR d\gl\times\\
& \left(\rho^{ph}(n)f(\gl-\mu_r-n\go)- \rho^{ph}(n-1)f(\gl-\mu_l-(n+1)\go)\right)\wh{\gs^{ph}_{n_l (n+1)_r}(\gl)},
\end{split}
\eed
Since $\mu_r = \go$ and $\mu_l = 0$ we find
\bed
\begin{split}
&J^{ph}_{\rho_0,Q^{el}_l}= -\sum_{n\in\dN}
\frac{\se}{2\pi}\int_\dR \times\\
& f(\gl-(n+1)\go)\rho^{ph}(n-1)(1 - e^{-\gb\go})\wh{\gs^{ph}_{n_l
    (n+1)_r}(\gl)}d\gl \le 0.
\end{split}
\eed
Hence there is a current going from the left to right induced by
photons. We recall that $J^c_{\rho_0, Q^{el}_l} = 0$.
\end{enumerate}
}
\end{remark}

\section*{Acknowledgments}

The first two authors would like to thank the University of Aalborg and the
Centre de Physique Th\'{e}orique - Luminy for hospitality and financial support.
In particular, we thank Horia D. Cornean for making us familiar with the $JCL$-model.


\begin{thebibliography}{10}

\bibitem{Pillet2007}
W.~Aschbacher, V.~Jak{\v{s}}i{\'c}, Y.~Pautrat, and C.-A. Pillet.
\newblock Transport properties of quasi-free fermions.
\newblock {\em J. Math. Phys.}, 48(3):032101, 28, 2007.

\bibitem{AJPI}
S.~Attal, A.~Joye, and C.-A. Pillet, editors.
\newblock {\em Open quantum systems. {I}}, volume 1880 of {\em Lecture Notes in
  Mathematics}.
\newblock Springer-Verlag, Berlin, 2006.
\newblock The Hamiltonian approach, Lecture notes from the Summer School held
  in Grenoble, June 16--July 4, 2003.

\bibitem{AJPII}
S.~Attal, A.~Joye, and C.-A. Pillet, editors.
\newblock {\em Open quantum systems. {II}}, volume 1881 of {\em Lecture Notes
  in Mathematics}.
\newblock Springer-Verlag, Berlin, 2006.
\newblock The Markovian approach, Lecture notes from the Summer School held in
  Grenoble, June 16--July 4, 2003.

\bibitem{AJPIII}
S.~Attal, A.~Joye, and C.-A. Pillet, editors.
\newblock {\em Open quantum systems. {III}}, volume 1882 of {\em Lecture Notes
  in Mathematics}.
\newblock Springer-Verlag, Berlin, 2006.
\newblock Recent developments, Lecture notes from the Summer School held in
  Grenoble, June 16--July 4, 2003.

\bibitem{Baro2004}
M.~Baro, H.-Chr. Kaiser, H.~Neidhardt, and J.~Rehberg.
\newblock A quantum transmitting {S}chr\"odinger-{P}oisson system.
\newblock {\em Rev. Math. Phys.}, 16(3):281--330, 2004.

\bibitem{Baumgaertel1983}
H.~Baumg{\"a}rtel and M.~Wollenberg.
\newblock {\em Mathematical scattering theory}.
\newblock Akademie-Verlag, Berlin, 1983.

\bibitem{Buettiker1985}
M.~B\"uttiker, Y.~Imry, R.~Landauer, and S.~Pinhas.
\newblock {G}eneralized many-channel conductance formula with application to
  small rings.
\newblock {\em Phys. Rev. B}, 31(10):6207--6215, May 1985.

\bibitem{Nenciu2008}
H.~D. Cornean, P.~Duclos, G.~Nenciu, and R.~Purice.
\newblock Adiabatically switched-on electrical bias and the
  {L}andauer-{B}\"uttiker formula.
\newblock {\em J. Math. Phys.}, 49(10):102106, 20, 2008.

\bibitem{Cornean2012b}
H.~D. Cornean, P.~Duclos, and R.~Purice.
\newblock Adiabatic non-equilibrium steady states in the partition free
  approach.
\newblock {\em Ann. Henri Poincar\'e}, 13(4):827--856, 2012.

\bibitem{Cornean2010c}
H.~D. Cornean, C.~Gianesello, and V.~A. Zagrebnov.
\newblock A partition-free approach to transient and steady-state charge
  currents.
\newblock {\em J. Phys. A}, 43(47):474011, 15, 2010.

\bibitem{Cornean2005}
H.~D. Cornean, A.~Jensen, and V.~Moldoveanu.
\newblock A rigorous proof of the {L}andauer-{B}\"uttiker formula.
\newblock {\em J. Math. Phys.}, 46(4):042106, 28, 2005.

\bibitem{Cornean2006}
H.~D. Cornean, A.~Jensen, and V.~Moldoveanu.
\newblock The {L}andauer-{B}\"uttiker formula and resonant quantum transport.
\newblock In {\em Mathematical physics of quantum mechanics}, volume 690 of
  {\em Lecture Notes in Phys.}, pages 45--53. Springer, Berlin, 2006.

\bibitem{CNWZ2012}
H.~D. Cornean, H.~Neidhardt, L.~Wilhelm, and V.~A. Zagrebnov.
\newblock {T}he {C}ayley transform applied to non-interacting quantum
  transport.
\newblock {\em arXiv, math-ph: 1212.4965v1}, 2012.
\newblock {T}o appear in {\em J. Funct. Anal.}

\bibitem{CNZ2009}
H.~D. Cornean, H.~Neidhardt, and V.~A. Zagrebnov.
\newblock The effect of time-dependent coupling on non-equilibrium steady
  states.
\newblock {\em Ann. Henri Poincar\'e}, 10(1):61--93, 2009.

\bibitem{Damak2006}
M.~Damak.
\newblock On the spectral theory of tensor product {H}amiltonians.
\newblock {\em J. Operator Theory}, 55(2):253--268, 2006.

\bibitem{GeKn2005}
C.~Gerry and P.~Knight.
\newblock {\em Introductory Quantum Optics}.
\newblock Cambridge University Press, Cambridge, 2005.

\bibitem{Ka1995}
T.~Kato.
\newblock {\em Perturbation theory for linear operators}.
\newblock Classics in Mathematics. Springer-Verlag, Berlin, 1995.
\newblock Reprint of the 1980 edition.

\bibitem{Landauer1957}
R.~Landauer.
\newblock {S}patial {V}ariation of {C}urrents and {F}ields {D}ue to {L}ocalized
  {S}catterers in {M}etallic {C}onduction.
\newblock {\em IBM J. Res. Develop.}, 1(3):223--231, 1957.

\bibitem{Nenciu2007}
G.~Nenciu.
\newblock {I}ndependent electron model for open quantum systems:
  {L}andauer-{B}\"uttiker formula and strict positivity of the entropy
  production.
\newblock {\em J. Math. Phys.}, 48(3):033302, 8, 2007.

\end{thebibliography}

\end{document}